\documentclass[11pt]{article}

\usepackage{natbib}
\usepackage[a4paper, margin=1in]{geometry}
\usepackage[nodisplayskipstretch]{setspace}
\usepackage{authblk}

\usepackage{graphicx, tabularx, booktabs}
\usepackage{amsthm, amsmath, amsfonts, amssymb}
\usepackage{xr, hyperref}
\usepackage[capitalize]{cleveref}
\usepackage{MnSymbol}
\usepackage{bbold}
\usepackage{xcolor}
\usepackage{enumitem, multicol}
\usepackage{pifont}
\usepackage[lined,boxed,commentsnumbered,ruled,linesnumbered,algosection]{algorithm2e}

\numberwithin{equation}{section}
\allowdisplaybreaks
\theoremstyle{plain}
\newtheorem{theorem}{Theorem}[section]
\newtheorem{lemma}[theorem]{Lemma}
\newtheorem{proposition}[theorem]{Proposition}
\newtheorem{corollary}[theorem]{Corollary}
\newcommand{\thistheoremname}{}
\newtheorem*{genericthm*}{\thistheoremname}
\newenvironment{namedthm*}[1]
{\renewcommand{\thistheoremname}{#1}%
	\begin{genericthm*}}
	{\end{genericthm*}}

\theoremstyle{remark}
\newtheorem{remark}{Remark}[section]
\newtheorem{assumption}{Assumption}
\newtheorem{definition}{Definition}
\newtheorem{example}{Example}
\newtheorem{algo}{Algorithm}

\setlist[enumerate]{label={(\alph*)}, nosep}
\setlist[itemize]{nosep}
\newcommand{\cmark}{\ding{51}}
\newcommand{\xmark}{\ding{55}}

\definecolor{myblue}{rgb}{0.1, 0.1, 0.7}
\hypersetup{
	hidelinks,
	linktocpage,
	colorlinks=true,
	linkcolor=myblue,
	citecolor=myblue,
	filecolor=myblue,
	urlcolor=myblue,
}

\crefformat{section}{Section~#2#1#3}
\crefformat{subsection}{Section~#2#1#3}
\crefformat{subsubsection}{Section~#2#1#3}
\crefrangeformat{section}{Sections~#3#1#4 to~#5#2#6}
\crefmultiformat{section}{Sections~#2#1#3}{ and~#2#1#3}{, #2#1#3}{ and~#2#1#3}
\crefformat{equation}{Equation~#2#1#3}
\crefrangeformat{equation}{Equations~#3#1#4 to~#5#2#6}
\crefmultiformat{equation}{Equations~#2#1#3}{ and~#2#1#3}{, #2#1#3}{ and~#2#1#3}
\newcommand{\Supp}{Appendix}

\newcommand{\di}{\mathrm{d}}
\newcommand{\T}{\intercal}
\newcommand{\argmin}{\mathop{\mathrm{argmin}}}
\newcommand{\diag}{\mathrm{diag}}

\newcommand{\I}{\mathbb{1}}
\newcommand{\E}{\mathbb{E}}
\newcommand{\Var}{\mathrm{Var}}
\newcommand{\Bias}{\mathrm{Bias}}
\newcommand{\MSE}{\mathrm{MSE}}
\newcommand{\norm}[1]{\left\lVert #1 \right\rVert}

\newcommand{\Corr}{\mathrm{Corr}}
\newcommand{\pr}{\mathbb{P}}

\newcommand{\Normal}{\mathrm{N}}
\newcommand{\Ber}{\mathrm{Bernoulli}}

\newcommand{\PSR}{\mathsf{PSR}}
\newcommand{\TSR}{\mathsf{TSR}}

\newcommand{\WFD}{\mathsf{WFD}}
\newcommand{\obm}{\mathsf{obm}}
\newcommand{\sub}{\mathsf{sub}}
\newcommand{\bart}{\mathsf{bart}}

\newcommand{\Lest}{\mathsf{L}}
\newcommand{\LASER}{\mathsf{LASER}}

\newcommand{\Z}{\mathbb{Z}}
\newcommand{\N}{\mathbb{N}_0}
\newcommand{\R}{\mathbb{R}}

\date{}

\begin{document}

\title{\bf Principles of Statistical Inference in Online Problems}
\author[1]{Man Fung Leung \thanks{Correspondence email: \href{mailto:mfleung2@illinois.edu}{\nolinkurl{mfleung2@illinois.edu}}}}
\author[2]{Kin Wai Chan}
\affil[1]{Department of Statistics, University of Illinois at Urbana-Champaign}
\affil[2]{Department of Statistics, The Chinese University of Hong Kong}
\maketitle

\begin{abstract}
To investigate a dilemma of statistical and computational efficiency faced by long-run variance estimators, we propose a decomposition of kernel weights in a quadratic form and some online inference principles.
These proposals allow us to characterize efficient online long-run variance estimators.
Our asymptotic theory and simulations show that this principle-driven approach leads to online estimators with a uniformly lower mean squared error than all existing works.
We also discuss practical enhancements such as mini-batch and automatic updates to handle fast streaming data and optimal parameters tuning.
Beyond variance estimation, we consider the proposals in the context of online quantile regression, online change point detection, Markov chain Monte Carlo convergence diagnosis, and stochastic approximation.
Substantial improvements in computational cost and finite-sample statistical properties are observed when we apply our principle-driven variance estimator to original and modified inference procedures.
\end{abstract}


\vfill
\newpage
\doublespacing

\section{Introduction} \label{sec:introduction}

Online problems arise naturally in many fields of statistics.
On top of them, modern computing allows intractable offline problems to be approached with online techniques.
To name a few, multidimensional integrals can be approximated by Markov chain Monte Carlo methods \citep{hastings1970mcmc};
deep neural networks can be trained using stochastic approximation \citep{robbins1951sgd}.
The usefulness of statistical inference in these problems motivates the study of online long-run variance estimators.
Existing works \citep{yeh2000dbm,wu2009recursive,rtavc,rtacm,chen2020sgd,zhu2023sgd} focus on how to select subsamples so that a batch means estimator can be updated online.
Nevertheless, their proposals are all dominated by the original offline batch means estimator and little is known about the reason behind.

We propose a novel framework to study online inference procedures in a principle-driven way.
\cref{sec:general-framework} introduces the framework, which consists of an interesting window decomposition and some general principles.
\cref{sec:long-run-variance-estimation} shows that the framework leads to online long-run variance estimators with higher statistical and computational efficiency than all existing works.
\cref{sec:beyond-variance-estimation} considers the framework in the context of online quantile regression, online change point detection, Markov chain Monte Carlo convergence diagnosis, and stochastic approximation.
\cref{sec:discussion} discusses our findings and some future directions.
All experiments are performed on Red Hat Enterprise Linux 7.9 with an Intel Xeon Gold 6148 CPU and R version 4.2.2.
An R-package \texttt{rlaser} that implements our estimators is available online.
Some algorithms and additional results are also deferred to the \Supp.

At the beginning of each section, we briefly describe or motivate subsequent subsections.
\cref{sec:review} reviews the setting of long-run variance estimation and introduces some notation.

\subsection{Review} \label{sec:review}

Suppose the sample mean $\bar{X}_n = \sum_{i=1}^n X_i/n$ is of interest.
For independent and identically distributed data $X_1, \ldots, X_n$, \citet{welford1962recursive} formulated the first online estimator of $n \Var(\bar{X}_n)$:
\begin{equation} \label{eq:welford}
	\hat{\sigma}^2_{n,\WFD}
	= \frac{1}{n} \sum_{i=1}^n (X_i -\bar{X}_n)^2
	= \frac{n-1}{n} \left\{ \hat{\sigma}^2_{n-1,\WFD} +\frac{1}{n}(X_n -\bar{X}_{n-1})^2 \right\},
\end{equation}
where $\hat{\sigma}^2_{0,\WFD} = 0$.
The \hyperref[eq:welford]{$\hat{\sigma}^2_{n,\WFD}$} is an online estimator because only a finite number of arithmetic operations ($O(1)$-time update) and statistics ($O(1)$-space update) are involved at each time.
It is actively used to analyze streaming data \citep{river}, despite the dependence structure (e.g., in \citealt{maleki2021anomaly} and \citealt{bassler2022anomaly}) may call for a robust estimator.

Autocorrelation robustness is a prerequisite for correct inference with many online techniques such as Markov chain Monte Carlo \citep{jones2006fw} and stochastic approximation \citep{chen2020sgd}.
Consider $\{X_i\}_{i \in \Z}$ to be stationary and ergodic with mean $\mu = \E(X_1)$ and autocovariance $\gamma_k = \E\{(X_0 -\mu)(X_k -\mu)\}$.
Under suitable conditions (e.g., in \citealt{wu2011asymptotic}), the long-run variance
\begin{equation} \label{eq:lrv}
	\sigma^2 = \lim_{n \to \infty} n \Var(\bar{X}_n) = \sum_{k \in \Z} \gamma_k
\end{equation}
accounts for possibly non-zero $\{\gamma_k\}_{k=1}^{\infty}$.
In contrast, \hyperref[eq:welford]{$\hat{\sigma}^2_{n,\WFD}$} only estimates $\gamma_0$.

Long-run variance estimation has a long history because it is crucial to the inference of dependent data.
Let $\ell_n$ be a batch size that is monotonically increasing in $n$.
A classical estimator utilizing the overlapping batch means \citep{meketon1984obm} is
\begin{equation} \label{eq:obm}
	\hat{\sigma}^2_{n,\obm}
	= \frac{\sum_{i=\ell_n}^n \left(\sum_{j=i-\ell_n+1}^i X_j -\ell_n \bar{X}_n \right)^2}{\sum_{i=\ell_n}^n \ell_n}.
\end{equation}
However, there is no simple relation like \hyperref[eq:welford]{$\hat{\sigma}^2_{n,\WFD}$} between \hyperref[eq:obm]{$\hat{\sigma}^2_{n,\obm}$} and \hyperref[eq:obm]{$\hat{\sigma}^2_{n-1,\obm}$} \citep{wu2009recursive}.

Online long-run variance estimation starts with the dynamic batch means proposed by \citet{yeh2000dbm}.
The dynamic batch means allows $O(1)$-time and $O(1)$-space update but may lead to inconsistency.
\citet{wu2009recursive} proposed the first consistent alternative based on a subsample selection rule,
which modifies \hyperref[eq:obm]{$\hat{\sigma}^2_{n,\obm}$} by replacing $\ell_n$ with a sequence of batch sizes $\{\ell_i\}_{i=1}^n$:
\begin{equation} \label{eq:subClass}
	\hat{\sigma}^2_{n,\sub}
	= \frac{\sum_{i=1}^n \left(\sum_{j=i-\ell_i+1}^i X_j -\ell_i \bar{X}_n \right)^2}{\sum_{i=1}^n \ell_i}.
\end{equation}
However, the mean squared error of Wu's estimator is 1.78 times of \hyperref[eq:obm]{$\hat{\sigma}^2_{n,\obm}$}.
\citet{rtavc,rtacm} proposed trapezoidal ($\TSR$) and parallelogrammatic ($\PSR$) selection rules to improve the mean squared error from 1.78 times (Wu) to 1.20 times ($\TSR$) and 1.12 times ($\PSR$) of \hyperref[eq:obm]{$\hat{\sigma}^2_{n,\obm}$}.

Finally, we introduce some notation.
Denote the floor, ceiling and indicator functions by $\lfloor \cdot \rfloor$, $\lceil \cdot \rceil$ and $\I_{\{\cdot\}}$.
For $a, b \in \R$, define $\sum_{k=a}^b x_k = \sum_{k=\lceil a \rceil}^{\lfloor b\rfloor} x_k$
and $a \lor b = \max(a,b)$.
Write $a_n \sim b_n$ if $\lim_{n \to \infty} a_n/b_n = 1$.
For $p \ge 1$, denote the $\mathcal{L}^p$-norm by $\norm{\cdot}_p = (\E|\cdot|^p)^{1/p}$.
Denote $\MSE(\cdot) = \norm{\cdot -\sigma^2}_2^2$ and $\Bias(\cdot) = \E(\cdot) -\sigma^2$.

\section{General Framework} \label{sec:general-framework}

Existing literature often regards online update as an algorithmic property and studies online estimators on an algorithm-by-algorithm basis.
This is not only true for variance but also other estimands, e.g., in \citet{xiao2011recursive} and \citet{huang2014recursive}.
Consequently, searching for efficient online estimators in a principle-driven way remains an open problem.
We study this problem with a focus on long-run variance estimation.
In \cref{sec:window-decomposition}, we propose a window decomposition that enables us to pinpoint the source of sub-efficiency in existing works.
While some refinements can be made already, we propose some principles that characterize efficient online estimators for the first time in \cref{sec:principle-driven-sufficient-conditions}.

\subsection{Window Decomposition} \label{sec:window-decomposition}

Consider the following quadratic form:
\begin{equation} \label{eq:genClass}
	\hat{\sigma}_n^2 = \hat{\sigma}_n^2(W)
	= \frac{1}{n} \sum_{i=1}^n \sum_{j=1}^n W_n(i, j) (X_i -\bar{X}_n) (X_j -\bar{X}_n),
\end{equation}
where $W_n(i,j)$ is a window function.
This choice is motivated by the fact that \cref{eq:genClass} can be used to show the asymptotic equivalence of
\hyperref[eq:obm]{$\hat{\sigma}^2_{n,\obm}$} and the celebrated Bartlett kernel estimator,
\begin{equation} \label{eq:bart}
	\hat{\sigma}^2_{n,\bart}
	= \sum_{k=-\ell_n}^{\ell_n} \left( 1 -\frac{|k|}{\ell_n} \right) \frac{1}{n} \sum_{i=|k|+1}^n (X_i -\bar{X}_n) (X_{i-|k|} -\bar{X}_n).
\end{equation}
To search for efficient online estimators of \hyperref[eq:lrv]{$\sigma^2$},
we decompose $W_n(i,j)$ in \cref{eq:genClass} as
\begin{equation} \label{eq:winDecom}
	W_n(i,j) = T\left( d_n^T(i,j) \right) S\left( d_n^S(i,j) \right),
\end{equation}
where $T(\cdot): [0,\infty) \rightarrow \R$ is called a tapering function,
$S(\cdot): [0,\infty) \rightarrow \{0,1\}$ is called a subsampling function,
and $d_n^T(i,j) > 0$ and $d_n^S(i,j) > 0$ are distances between times $i$ and $j$.
Since \cref{eq:winDecom} disentangles the tapering weight and subsampling frequency, it opens the door to principle-driven studies of online estimation.
For online data, the distances in \cref{eq:winDecom} are naturally
\begin{equation} \label{eq:distance}
	d^T_n(i,j) = \frac{|i-j|}{t_n(i,j)} \quad \text{and} \quad
	d^S_n(i,j) = \frac{|i-j|}{s_n(i,j)}
\end{equation}
to measure the standardized time lag $|i-j|$ between observations $X_i$ and $X_j$,
where $t_n(i,j)>0$ and $s_n(i,j)>0$ are some smoothing parameters.
The asymptotic equivalence of \hyperref[eq:obm]{$\hat{\sigma}^2_{n,\obm}$} and \hyperref[eq:bart]{$\hat{\sigma}^2_{n,\bart}$} can be reexplained by having the same window that satisfies \cref{eq:distance}.
In contrast, existing sub-efficient online estimators violate \cref{eq:distance};
see Example \ref{eg:on}.

The tapering function mainly controls the statistical efficiency as it determines how much weight to be placed on $X_i X_j$, which estimates $\gamma_{|i-j|}$.
To reduce variance, it is sensible to assign a lighter weight when $d^T_n(i,j)$ is large.
However, if $T(d^T_n(i,j))$ tapers $X_i X_j$ too much, a larger bias is introduced.
Hence, the taper governs the bias--variance tradeoff.

The subsampling function mainly controls the computational efficiency as it determines the members of each subsample.
Since $\{X_j\}_{j>i}$ is not observed at time $i$, rewrite
\begin{equation} \label{eq:genClassOn}
	\hat{\sigma}_n^2
	= \frac{1}{n} \sum_{i=1}^n
	(X_i -\bar{X}_n) \sum_{j=1}^i W'_n(i,j) (X_j -\bar{X}_n),
\end{equation}
so that the inner summand is adapted to time $i$,
where $W'_n(i,j) = (2 -\I_{i=j}) W_n(i,j)$ assuming $W_n(i,j) = W_n(j,i)$.
Denote the $i$-th subsample as $\mathcal{S}_n(i) = \{X_j : S( d^S_n(i,j) )=1, 1 \le j \le i\}$.
The computational complexity of $\hat{\sigma}_n^2$ depends on the size and regularity of $\mathcal{S}_n(1),\ldots,\mathcal{S}_n(n)$.

The next two examples use \cref{eq:winDecom} to identify the structural problem of existing online estimators.
As a remark, \cref{eq:winDecom} is a conceptual decomposition because $S(d_n^S(i,j))$ can be absorbed into the definition of $T(d_n^T(i,j))$, which reduces to the classical definition of window in Example \ref{eg:off}.

\begin{example}[Offline estimators] \label{eg:off}
	Let $T(x)=1-x$; $S(x) = \I_{x\le 1}$; and $d^T_n(i,j) = d^S_n(i,j) = |i-j|/\ell_n$,
	where $\ell_n = \Lambda n^\lambda$, $\Lambda \in \R^+$ and $\lambda \in (0, 1)$.
	Then, it gives the Bartlett window $W_{\bart}$:
	\begin{equation} \label{eq:winBart}
		W_{n,\bart}(i,j)
		=\left( 1-\frac{|i-j|}{\ell_n} \right) \I_{|i-j| \le \ell_n}.
	\end{equation}
	Other window functions $K(x) = T(x)\I_{x\le 1}$ can be formed similarly, e.g., the Tukey--Hanning window with $T(x)=\{1+\cos(\pi x)\}/2$.
	For these estimators, there are two important features:
	\begin{itemize}
		\item The smoothing parameters that determine the tapering and subsampling behaviors are identical and simultaneously controlled by $\ell_n$, i.e., $t_n(i,j)=s_n(i,j)=\ell_n$.
		\item The distances $d_n^T(i,j)$ and $d_n^S(i,j)$ depend not only on $i,j$ but also on $n$ because both distances are standardized by the global bandwidth $\ell_n$, which is a function of $n$.
	\end{itemize}
\end{example}

\begin{example}[Online estimators] \label{eg:on}
	Let $T(x) = 1-x$; $S(x) = \I_{x\le 1}$; $d^T_n(i,j) = (|i-j| + \overline{\ell}_n - \ell_{i \lor j})/\overline{\ell}_n$; and $d^S_n(i,j) = |i-j|/\ell_{i \lor j}$,
	where $\ell_n = \Lambda n^\lambda$, $\overline{\ell}_n = n^{-1} \sum_{i=1}^n \ell_i$, $\Lambda \in \R^+$ and $\lambda \in (0, 1)$.
	Then, it gives the window $W_{\PSR}$ of the best existing online estimator \citep{rtacm}:
	\begin{equation} \label{eq:winPSR}
		W_{n,\PSR}(i,j)
		= \left( 1- \frac{|i-j| +\overline{\ell}_n -\ell_{i \lor j}}{\overline{\ell}_n} \right) \I_{|i-j| \le \ell_{i \lor j}}.
	\end{equation}
	The~\Supp~provides a general window of other less-efficient online estimators and the derivation of \cref{eq:winPSR}.
	For these estimators, there are also two important features:
	\begin{itemize}
		\item The natural distance $|i-j|$ in \cref{eq:winPSR} is distorted by a counterintuitive amount of $\overline{\ell}_n - \ell_{i \lor j}$.
		\item The distances $d_n^T(i,j)$ and $d_n^S(i,j)$ are regularized by the ad-hoc global average bandwidth $\overline{\ell}_n$ and the local bandwidth $\ell_{i \lor j}$, respectively.
	\end{itemize}
\end{example}

\subsection{Principle-driven Sufficient Conditions} \label{sec:principle-driven-sufficient-conditions}

\cref{eq:winDecom} gives hints to refine existing works.
In particular, Example \ref{eg:on} suggests that the window $W_{\PSR}$ in \cref{eq:winPSR} can be improved immediately by using $W_{\Lest}$ defined as
\begin{equation} \label{eq:winL}
	W_{n,\Lest}(i,j)
	= \left( 1- \frac{|i-j|}{\ell_{i \lor j}} \right) \I_{|i-j| \le \ell_{i \lor j}}.
\end{equation}
We prove that $\hat{\sigma}^2_n(W_{\Lest})$ uniformly dominates the best existing online estimator $\PSR$ in the \Supp.
However, we can achieve more with \cref{eq:winDecom} by characterizing efficient online estimators in a principle-driven way.

\begin{namedthm*}{LASER principles} \label{prin:LASER}
	Consider \cref{eq:winDecom}.
	The principles below, which can be summarized as Local Alike Separated Exterior Ramping, unite the statistical and computational properties of \hyperref[eq:genClass]{$\hat{\sigma}_n^2$}.
	\begin{enumerate}
		\item[(L)] An online estimator should utilize local subsamples.
		\item[(A)] Under stationarity, if $(X_i,X_j)$ and $(X_{i'}, X_{j'})$ are included in the subsamples and $|i-j| = |i'-j'|$, they should be weighted alike.
		\item[(S)] The tapering and subsampling parameters should be separately chosen.
		\item[(E)] An $O(1)$-time estimator should be able to exteriorize (factor out) the global part of the tapering parameter from the weight of every subsample.
		\item[(R)] An $O(1)$-space estimator should ramp up (increase) the subsample size until it is too large.
	\end{enumerate}
\end{namedthm*}

Philosophically, Principle \hyperref[prin:LASER]{L} means that subsamples should be built locally and adapted to their current time to avoid reconstruction of subsamples when new data arrive.
In our context, the subsample $\mathcal{S}_n(i)$ should depend on its local time $i$ only.
It leads to the window $W_{\Lest}$ in \cref{eq:winL}.

Principle \hyperref[prin:LASER]{A} states that when the distances are the same, the data pairs contain the same amount of information under stationarity and so should be treated equally.
This is sensible and helps to explain the sub-efficiency of existing online estimators; see Example \ref{eg:on}.

Principle \hyperref[prin:LASER]{S} allows tuning statistical and computational properties separately.
To tune the tapering and subsampling behaviors,
$W_{\bart}$ in \cref{eq:winBart} uses the same $\ell_n$,
whereas $W_{\Lest}$ in \cref{eq:winL} uses the same $\ell_{i\vee j}$.
Consequently, they sacrifice computational and statistical efficiency, respectively.

Principle \hyperref[prin:LASER]{E} focuses on the relation between $O(1)$-time update and the taper.
If the $n$-related part of the tapering parameter $t_n(i,j)$ can be exteriorized (factored) out of all weights,
then the weighted sums will be local and $O(1)$-time update is possible.
For illustration, we exteriorize
\begin{equation} \label{eq:egE}
	\frac{1}{n} \sum_{i=1}^n X_i \sum_{j=1}^i \left( 1-\frac{|i-j|}{\ell_n} \right) X_j
	= \left( \frac{\ell_n}{n} \sum_{i=1}^n X_i \sum_{j=1}^i X_j \right)
	- \left\{ \frac{1}{n \ell_n} \sum_{i=1}^n X_i \sum_{j=1}^i (i-j) X_j \right\}.
\end{equation}
Then, we can update \cref{eq:egE} in $O(1)$ time because the inner sums $\sum_{j=1}^i X_j$ and $\sum_{j=1}^i (i-j) X_j$ are local after factoring out $\ell_n/n$ and $1/(n\ell_n)$, respectively.
In contrast, we cannot exteriorize $1/\ell_n$ out of $T(|i-j|/\ell_n)$ when $T(x)=\{1+\cos(\pi x)\}/2$.

Principle \hyperref[prin:LASER]{R} suggests a relation between $O(1)$-space update and the subsamples.
To illustrate, updating $\sum_{i=1}^n X_i$ to $\sum_{i=2}^{n+1} X_i$ requires storing the far observation $X_1$
so that $X_1$ can be removed from the sum.
Hence, trimming far observations frequently in the subsamples, e.g.,
$(X_1,X_2), (X_2,X_3), (X_3,X_4),$ 
$(X_4, X_5), \ldots$,
incurs a higher memory cost than ramping up the subsample size moderately, e.g.,
$(X_1,X_2), (X_1,X_2,X_3), (X_1,X_2,X_3,X_4), (X_4, X_5), \ldots$.
Hence, to allow $O(1)$-space update, we define the effective (ramped) subsampling parameter
\begin{equation} \label{eq:rampSub}
	s'_n = \left\{
	\begin{array}{ll}
		s'_{n-1} +1, &\quad \text{if} \quad s_{n-1} \le s'_{n-1} +1 < \phi s_{n-1}; \\
		s_n, &\quad \text{if} \quad s'_{n-1} +1 \ge \phi s_{n-1},
	\end{array}
	\right.
\end{equation}
where $s_n$ is the intended subsampling parameter, and $\phi \in [1, \infty)$ is the memory parameter that controls when $s_n'$ stops ramping up and resets to $s_n$.
For $\phi = 1$, $s'_n$ reduces to $s_n$.

Next, we provide a mathematical setup that relates the computational complexities of \hyperref[eq:genClass]{$\hat{\sigma}_n^2$} to the LASER principles.

\begin{proposition}[Online update] \label{prop:on}
	Consider \cref{eq:winDecom} with the distances in \cref{eq:distance} and a subsampling function $S(x) = \I_{x\le 1}$.
	Let $q \in \Z^+$; $a_0,\ldots,a_q \in \R$; and $\{s_n\}_{n \in \Z^+}$, $\{t_n\}_{n \in \Z^+}$ and $\{t'_n\}_{n \in \Z^+}$ be some non-zero sequences.
	Then \hyperref[eq:genClass]{$\hat{\sigma}_n^2$} can be updated in $O(1)$ time if
	\begin{enumerate}
		\item[(L)] the subsampling parameter is local, i.e., $s_n(i,j) = s_{i \lor j}$;
		\item[(A)] the tapering parameter is separable in $i \lor j$ and $n$, i.e., $t_n(i,j) = t_n t'_{i \lor j}$;
		\item[(S)] $\{s_n\}, \{t_n\}, \{t'_n\}$ are possibly different sequences; and
		\item[(E)] the tapering function is of the form $T(x) = \sum_{r=0}^q a_r x^r$.
	\end{enumerate}
	Suppose further that $\{s_n\}$ is a fixed or monotonically increasing random sequence.
	Then \hyperref[eq:genClass]{$\hat{\sigma}_n^2$} can be updated in $O(1)$ space if
	\begin{enumerate}
		\item[(R)] the subsampling parameter is ramped with $\phi \ge 2$, i.e., $s_n(i,j) = s'_{i \lor j}$ in \cref{eq:rampSub}.
	\end{enumerate}
\end{proposition}

Proposition \ref{prop:on} considers a more general setting than the implication of the LASER principles for two reasons.
First, many windows can be approximated by a polynomial basis.
Second, it includes the simple modification in \cref{eq:winL} as well as existing subsample selection rules, which can be written as a linear combination of several parts that satisfy Proposition \ref{prop:on}.
For $O(1)$-space update, we need to know the intended $s_n$ in advance to prepare for resets of $s'_n$.
Therefore, we consider a prespecified or data-driven but monotonically increasing sequence $\{s_n\}$.

To fully align with the LASER principles, we recommend the following settings.
In (A),  we use a constant sequence $\{t'_i\}$,
i.e., $t'_1=t'_2=\cdots$.
In (E), we set $a_0=1$, $a_q=-1$, and $a_1=\cdots=a_{q-1}=0$ so that $T(x) = 1-x^q$.
It can enhance the computational efficiency by updating fewer summary statistics,
while the window order $q$ and convergence rate remain unchanged;
see \citet{parzen1957kernel}.
By Proposition \ref{prop:on} with the recommended settings based on the LASER principles,
our proposed online estimator of \hyperref[eq:lrv]{$\sigma^2$} is
\begin{equation} \label{eq:LASER}
	\hat{\sigma}_{n,\LASER(q,\phi)}^2
	= \hat{\sigma}_n^2(W_{\LASER(q,\phi)})  \quad \text{with} \quad
	W_{n,\LASER(q,\phi)}(i,j)
	= \left( 1 -\frac{|i-j|^q}{t_n^q} \right) \I_{|i-j| \le s'_{i \lor j}},
\end{equation}
where $\Theta,\Psi \in \R^+$; $\theta,\psi \in (0,1)$; $q \in \Z^+$; $\phi \in [1,\infty)$; $t_n \sim \Theta n^{\theta}$; $s_n \sim \Psi n^{\psi}$; and $s'_n$ is defined in \cref{eq:rampSub}.
This proposal is always $O(1)$-time, and $O(1)$-space when $\phi \ge 2$.

\begin{table}[!t]
\centering
\scalebox{0.93}{
\begin{tabularx}{1.08\textwidth}{Xcccccc}
	\toprule[2pt]
	\multicolumn{3}{c}{Estimator}
	& \multicolumn{2}{c}{Statistical properties}
	& \multicolumn{2}{c}{Computational properties} \\
	\cmidrule(lr){1-3} \cmidrule(lr){4-5} \cmidrule(lr){6-7}
	$\hat{\sigma}_n^2$ & $q$ & $\phi$
	& Optimal MSE & Automatic
	& Time & Space  \\
	\midrule[1pt]
	$\bart$ \citep{andrews1991kernel} & $1$ & \slash
	& $B_n = 2.29 \sigma^4 \kappa_1^{2/3} n^{-2/3}$ & \xmark
	& $O(n)$ & $O(n)$ \\
	\citet{parzen1957kernel} & $3$ & \slash
	& $P_n = 3.39 \sigma^4 \kappa_3^{2/7} n^{-6/7}$ & \xmark
	& $O(n)$ & $O(n)$ \\
	\midrule[1pt]
	\citet{welford1962recursive} & $1$ & \slash
	& inconsistent for $\sigma^2$ & \slash
	& $O(1)$ & $O(1)$ \\
	\citet{yeh2000dbm} & $1$ & \slash
	& inconsistent for $\sigma^2$ & \xmark
	& $O(1)$ & $O(1)$ \\
	\citet{wu2009recursive} & $1$ & \slash
	& $1.78B_n$ & \xmark
	& $O(1)$ & $O(1)$ \\
	$\TSR$ \citep{rtavc} & $1$ & \slash
	& $1.20B_n$ & \xmark
	& $O(1)$ & $O(1)$ \\
	$\PSR$ \citep{rtavc} & $1$ & \slash
	& $1.12B_n$ & \xmark
	& $O(1)$ & $O(n^{1/3})$ \\
	$\PSR$ \citep{rtacm} & $3$ & \slash
	& $1.06P_n$ & \cmark
	& $O(1)$ & $O(n^{1/7})$ \\
	\midrule[1pt]
	$\Lest$ (\hyperref[eq:winL]{simple modification}) & $1$ & \slash
	& $1.08B_n$ & \cmark
	& $O(1)$ & $O(n^{1/3})$ \\
	$\LASER$ (\hyperref[eq:LASER]{proposal}) & $1$ & $2$
	& $1.01B_n$ & \cmark
	& $O(1)$ & $O(1)$ \\
	$\LASER$ (\hyperref[eq:LASER]{proposal}) & $1$ & $1$
	& $0.96B_n$ & \cmark
	& $O(1)$ & $O(n^{1/3})$ \\
	$\LASER$ (\hyperref[eq:LASER]{proposal}) & $3$ & $1$
	& $0.98P_n$ & \cmark
	& $O(1)$ & $O(n^{1/7})$ \\
	\bottomrule[2pt]
\end{tabularx}
}
\caption{Summary of properties of different long-run variance estimators.}
\label{tab:properties}
\end{table}

Finally, we summarize the properties of different estimators in Table \ref{tab:properties}.
While we have restricted our attention to online estimators, our construction can be used with many innovations in the offline literature and extended easily to other estimands.
We also enhance our proposals practically with mini-batch and automatic updates, which will be discussed in \cref{sec:practical-enhancements}.
Since automatic update frees users from smoothing parameter selection, they only need to decide $q$ and $\phi$.
In general, we recommend \hyperref[eq:LASER]{$\LASER(1,1)$}.
If memory is scarce, we recommend \hyperref[eq:LASER]{$\LASER(1,2)$}.
We conclude with an example that verifies the superiority of our principle-driven estimator.

\begin{example}[Online long-run variance estimation] \label{eg:mse}
	Consider the following models:
	
	\begin{enumerate}[label=(\Roman*)]
		\item Autoregressive moving average: Let $X_i = a X_{i-1} +b\varepsilon_{i-1} +\varepsilon_i$, where $\varepsilon_i \sim \Normal(0,1)$.
		Take $a = 0.5$ and $b = 0.5$,
		which results in a mildly autocorrelated linear time series.
		\item Bilinear: A strongly autocorrelated nonlinear time series is generated from
		\begin{equation} \label{eq:bilinear}
			X_i = (0.9 +0.1\varepsilon_i) X_{i-1} +\varepsilon_i, \quad \text{where} \quad \varepsilon_i \sim \Normal(0,1).
		\end{equation}
		\item Fractional Gaussian noise: Let $X_i$ be a zero-mean Gaussian process with $\gamma_k = a(k+b)^{-c}$.
		Take $a = 100$, $b = 5$ and $c = 5$, which leads to another nonlinear time series.
	\end{enumerate}
	
	Under Models I--III, we compute the mean squared errors of $\TSR$, $\PSR$, \hyperref[eq:LASER]{$\LASER(1,1)$} and
	\hyperref[eq:LASER]{$\LASER(1,2)$} based on 1000 replications.
	Oracle parameters are used to eliminate the effect of smoothing parameter selection, which will be investigated in Example \ref{eg:auto}.
	Figure \ref{fig:mse-ratio} confirms the improvements brought by our framework in finite sample.
\end{example}

\begin{figure}[!t]
	\centering
	\includegraphics[width=\linewidth]{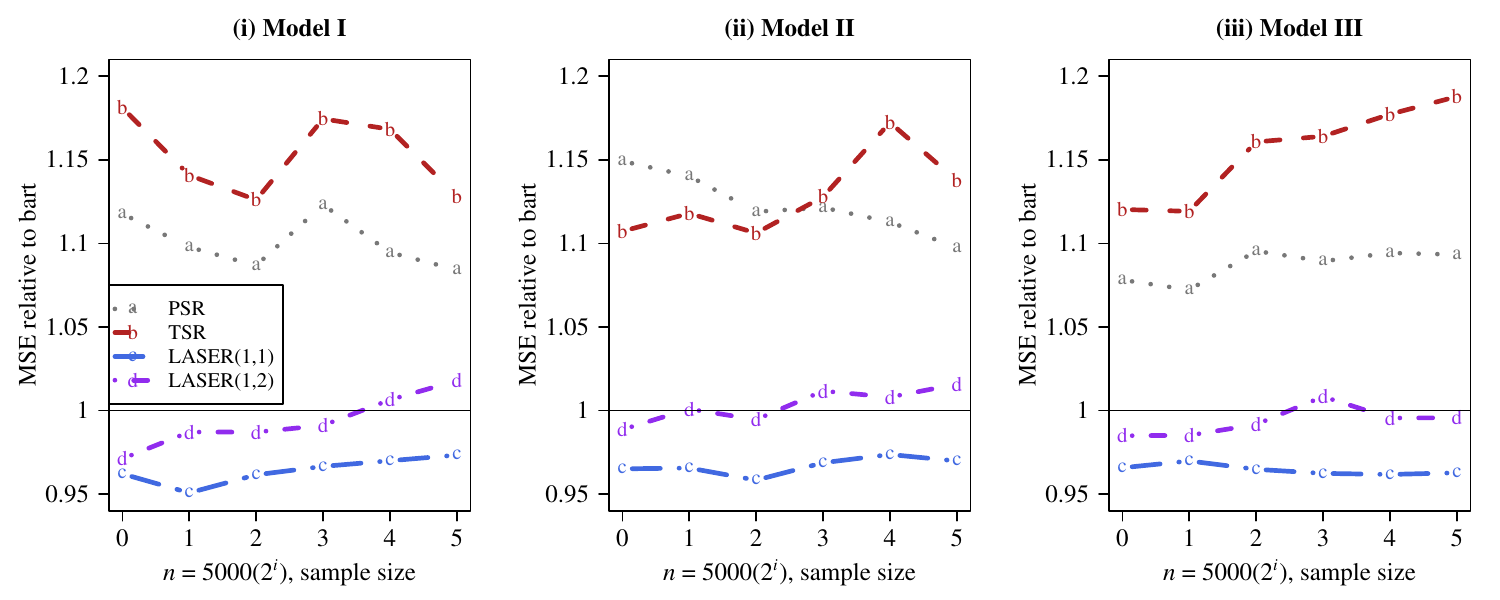}
	\caption{Long-run variance estimation using different online estimators: (a) $\TSR$ (dotted gray); (b) $\PSR$ (dashed red); (c) \hyperref[eq:LASER]{$\LASER(1,1)$} (longdash blue); (d) \hyperref[eq:LASER]{$\LASER(1,2)$} (dotdash purple). 
		The plots show the value of $\MSE(\cdot)/\MSE(\hat{\sigma}^2_{n,\bart})$.}
	\label{fig:mse-ratio}
\end{figure}

\begin{remark} \label{rmk:pd}
	The finite-sample estimates of \hyperref[eq:LASER]{$\LASER(q,\phi)$} may not be positive definite.
	To resolve this issue, we follow \citet{jentsch2015pd} and implement an adjustment that retains asymptotic properties; see the \Supp~and also \citet{vats2021kernel}.
\end{remark}

\section{Long-run Variance Estimation} \label{sec:long-run-variance-estimation}

To compare with existing works on the same basis, we develop the asymptotic theory of \hyperref[eq:genClass]{$\hat{\sigma}_n^2$} based on the dependence measures of \citet{wu2005asymptotic} in \cref{sec:consistency,sec:l2-optimal-convergence-rate}.
A surprising finding is that our proposed estimator can be ``super-optimal'' in the sense of \citet{rtacm}, i.e., it is online with a lower asymptotic mean squared error than \hyperref[eq:obm]{$\hat{\sigma}^2_{n,\obm}$} and \hyperref[eq:bart]{$\hat{\sigma}^2_{n,\bart}$}.
In \cref{sec:practical-enhancements}, we discuss mini-batch and automatic updates, which addresses some practical reasons why offline estimators remain common in online problems.

\subsection{Consistency} \label{sec:consistency}

Let $X_i = g(\mathcal{F}_i)$ for some measurable function $g$,
where $\mathcal{F}_i = (\dots, \epsilon_{i-1}, \epsilon_i)$ is the shift process of independent and identically distributed (iid) innovations $\{ \epsilon_i \}_{i \in \Z}$.
Let $\epsilon'_j$ be an iid copy of $\epsilon_j$,
$\mathcal{F}_{i,\{j\}} = (\mathcal{F}_{j-1}, \epsilon'_j, \epsilon_{j+1}, \dots, \epsilon_i)$
and $X_{i,\{j\}} = g(\mathcal{F}_{i,\{j\}})$.
\citet{wu2005asymptotic} defined $\delta_{i,p} = \norm{ X_i -X_{i,\{0\}} }_p$ (physical dependence measure)
and $\omega_{i,p} = \norm{ \E(X_i \mid \mathcal{F}_0) -\E(X_i \mid \mathcal{F}_{0,\{0\}}) }_p$ (predictive dependence measure).
These measures have a wide range of applications, which include but are not limited to long-run variance estimation.
Interested readers are referred to \citet{wu2011asymptotic}.

\begin{assumption}[Stability] \label{asum:stability}
	For some $\alpha > 2$,
	$\Delta_\alpha = \sum_{i=0}^\infty \delta_{i,\alpha} < \infty$.
\end{assumption}

\begin{assumption}[Summability of window] \label{asum:winSums}
	Let $\alpha' = \min(\alpha/2 ,2)$, $c \in (0, 1-1/\alpha')$,
	\begin{align*}
		G_{1,n} &= \max_{1 \le i \le n} \sum_{j=1}^n \left\{ \big|  W_n(i,j) \big| +W_n^2(i,j) \right\}, \\
		G_{2,n} &= \max_{1 \le i,j \le n} \big| W_n(i,j) \big| +\max_{1 \le i \le n} \left\{ \sum_{j=2}^n \big| W_n(i,j) -W_n(i,j-1) \big|^{\alpha'} \right\}^{\frac{1}{\alpha'}}.
	\end{align*}
	The window in \cref{eq:genClass} satisfies $G_{1,n} = o(n^{2-2/\alpha'})$ and $G_{2,n} = o(n^c)$.
\end{assumption}

\begin{assumption}[General window for long-run variance estimation] \label{asum:winGen}
	For $n\in\Z^+$ and $k=0,\ldots,n-1$, define
	$w_{n,k} = n^{-1} \sum_{i=k+1}^n W_n(i,i-k)$.
	The window in \cref{eq:genClass} satisfies
	\begin{enumerate}[label=(\alph*)]
		\item $W_n(i,j) = W_n(j,i)$ for all $n\in\Z^+$ and $i,j \in \{1, \ldots, n\}$;
		\item $W_n(i,i) = 1$ for all $n\in\Z^+$ and $i \in \{1, \ldots, n\}$; and
		\item there exists an increasing sequence $\{b_n\}_{n\in\Z^+}$ that diverges to $\infty$ as $n \to \infty$ such that \linebreak
		$\max_{1 \le k \le b_n} |w_{n,k}-1| = o(1)$ and
		$\max_{b_n < k < n} |w_{n,k}| = O(1)$.
	\end{enumerate}
\end{assumption}

Assumption \ref{asum:stability} is known as the stability condition in \citet{wu2005asymptotic}, which ensures $\sigma^2 < \infty$.
Assumptions \ref{asum:winSums} and \ref{asum:winGen} regulate the behavior of a general window so that $\hat{\sigma}_n^2$ is guaranteed to be precise and accurate, respectively.
These assumptions are mild and easily verifiable once $W_n(i,j)$, $s_n(i,j)$ and $t_n(i,j)$ are specified.
For instance, they hold under Definition \ref{def:winLASER}.

\begin{assumption}[$\tilde{q}$-th order serial dependence] \label{asum:uq}
	For some $\tilde{q} \in \Z^+$,
	$u_{\tilde{q}} = \sum_{k \in \Z} |k|^{\tilde{q}} |\gamma_k| < \infty$.
\end{assumption}

\begin{definition}[Parameters of $\LASER(q,\phi)$] \label{def:winLASER}
	Let $q \in \Z^+$ and $\phi \in [1,\infty)$ be fixed;
	$s'_n$ depends on $s_n$ according to \cref{eq:rampSub}; and
	$\alpha' = \min(\alpha/2, 2)$.
	The window function takes the form in \cref{eq:LASER}.
	The intended subsampling parameter and tapering parameter take the form
	$s_n = \min( \lfloor \Psi n^\psi \rfloor, n-1)$ and
	$t_n = \min( \lceil \Theta n^\theta \rceil, n)$.
	The coefficients satisfy $\Psi, \Theta \in \R^+$.
	The exponents satisfy $0 < \psi < \min \{ 2-2/\alpha', 1/(1+q) \}$, and either
	\begin{enumerate}[label=(\alph*)]
		\item $\psi \le \theta$; or
		\item $\max\{ \psi +(\psi -2+2/\alpha')/(2q), (q-\tilde{q}) \psi /q\} < \theta < \psi$,
	\end{enumerate}
	where $\tilde{q}$ is the order of serial dependence in Assumption \ref{asum:uq}.
\end{definition}

\begin{theorem}[Consistency] \label{thm:consistency}
	Let $\alpha > 2$.
	Suppose $\norm{X_1}_\alpha < \infty$ and Assumption \ref{asum:stability} holds.
	\begin{enumerate}[label=(\alph*)]
		\item If Assumptions \ref{asum:winSums} and \ref{asum:winGen} hold, then
		$\lVert \hat{\sigma}_n^2 -\sigma^2 \rVert_{\alpha/2} = o(1)$.
		\item Under Definition \ref{def:winLASER}(a) or \ref{def:winLASER}(b) with Assumption \ref{asum:uq},
		$\lVert \hat{\sigma}^2_{n, \LASER(q,\phi)} -\sigma^2 \rVert_{\alpha/2} = o(1)$.
	\end{enumerate}
\end{theorem}

Theorem \ref{thm:consistency}(a) shows the $\mathcal{L}^{\alpha/2}$-consistency of general \hyperref[eq:genClass]{$\hat{\sigma}^2_n$}.
It covers both offline and online estimators, e.g., all consistent estimators stated in Table \ref{tab:properties}.
Theorem \ref{thm:consistency}(b) is dedicated to our proposed estimator under Definition \ref{def:winLASER}.
When $\psi \le \theta$, i.e., under Definition \ref{def:winLASER}(a), the estimator is always consistent.
However, if $\psi > \theta$, i.e., under Definition \ref{def:winLASER}(b), $s_n/t_n$ diverges so the window is not absolutely bounded.
In this case, the estimator is inconsistent unless some condition is imposed on the serial dependence.
Definition \ref{def:winLASER}(b) handles it under Assumption \ref{asum:uq}.
Consequently, $\theta$ admits the lower bound $(q-\tilde{q}) \psi /q$.

\subsection{\texorpdfstring{$\mathcal{L}^2$-optimal Convergence Rate}{L2-optimal Convergence Rate}} \label{sec:l2-optimal-convergence-rate}

\begin{assumption}[$q$-th order weak stability] \label{asum:qWeakStable}
	For some $\alpha, q \ge 1$,
	$\Omega_\alpha^{(q)} = \sum_{j=0}^\infty j^q \omega_{j, \alpha} < \infty$.
\end{assumption}

Assumption \ref{asum:qWeakStable} is satisfied by a broad class of linear and nonlinear time series; see \citet{wu2011asymptotic} and the references therein.
This assumption implies that $u_q = \sum_{k \in \Z} |k|^q |\gamma_k| < \infty$ and
$v_q = \sum_{k \in \Z} |k|^q \gamma_k$ is well-defined \citep{wu2009recursive}, which are used in the following theorem.

\begin{theorem}[Exact convergence rate] \label{thm:mse}
	Suppose $\norm{X_1}_\alpha < \infty$ and Definition \ref{def:winLASER} holds.
	\begin{enumerate}[label=(\alph*)]
		\item Let $\alpha \ge 4$. If Assumption \ref{asum:stability} holds, then as $n \to \infty$,
		\begin{gather*}
			n^{1-\psi -\max\{ 2q(\psi-\theta), 0\}} \Var(\hat{\sigma}_n^2)
			\to \mathcal{V}_{\psi,\Psi,\theta,\Theta,q,\phi} \sigma^4,
			\qquad \text{where} \qquad \mathcal{V}_{\psi,\Psi,\theta,\Theta,q,\phi} = \\
			\left\{
			\begin{array}{ll}
				\frac{2 \Psi (\phi + 1)}{\psi + 1} \I_{\psi \le \theta}
				-\frac{8 \Psi^{q + 1} \Theta^{- q} (\phi^{q + 2} - 1)}{(\phi - 1) (q + 1) (q + 2) (\psi q + \psi + 1)} \I_{\psi = \theta}
				+\frac{2 \Psi^{2 q + 1} \Theta^{- 2 q} (\phi^{2 q + 2} - 1)}{(\phi - 1) (q + 1) (2 q + 1) (2 \psi q + \psi + 1)} \I_{\psi \ge \theta}, & \phi > 1; \\
				\frac{4 \Psi}{\psi + 1} \I_{\psi \le \theta}
				-\frac{8 \Psi^{q + 1} \Theta^{- q}}{(q + 1) (\psi q + \psi + 1)} \I_{\psi = \theta}
				+\frac{4 \Psi^{2 q + 1} \Theta^{- 2 q}}{(2 q + 1) (2 \psi q + \psi + 1)} \I_{\psi \ge \theta} , & \phi = 1. \\
			\end{array}
			\right.
		\end{gather*}
		\item Let $\alpha \ge 2$. If Assumption \ref{asum:qWeakStable} holds for the fixed $q$, then as $n \to \infty$,
		\[
		\Bias(\hat{\sigma}_n^2) \sim \left\{
		\begin{array}{ll}
			o(n^{-q\psi}), & \psi < \theta; \\
			-\Theta^{-q} n^{-q \theta} v_q, & \psi \ge \theta. \\
		\end{array}
		\right.
		\]
	\end{enumerate}
\end{theorem}

It follows from Theorem \ref{thm:mse} that $\psi = \theta = 1/(1+2q)$ optimize the order of $\MSE(\hat{\sigma}_n^2)$ for each $q$; see Remark \ref{rmk:mse}.
For the examples in Table \ref{tab:properties}, $\mathcal{V}_{1/3,\Psi,1/3,\Theta,1,2} = 2.5 \Psi^3/\Theta^2 - 5.6 \Psi^2/\Theta + 4.5 \Psi$,
$\mathcal{V}_{1/3,\Psi,1/3,\Theta,1,1} = (2/3) \Psi^3/\Theta^2 - 2.4 \Psi^2/\Theta + 3 \Psi$, and
$\mathcal{V}_{1/5,\Psi,1/5,\Theta,3,1} = (5/21) \Psi^7/\Theta^6 - (10/9) \Psi^4/\Theta^3 + (10/3) \Psi$.

\begin{corollary}[$\mathcal{L}^2$-optimal parameters] \label{coro:mse}
	Define $\kappa_q = |v_q|/\sigma^2$.
	Suppose that the conditions in Theorem \ref{thm:mse}(a) and (b) hold.
	If $\psi = \theta = 1/(1+2q)$, the $\mathcal{L}^2$-optimal $\Psi$ is
	\[
	\Psi_\star = \left\{
	\begin{array}{ll}
		\left\{ \frac{(\phi + 1) (2 q + 1)}{2q (q + 1)}
		- \frac{4 (\phi^{q + 2} -1) (2 q + 1)}{(\phi - 1) q (q + 1) (q + 2) (3 q + 2)}
		+ \frac{(\phi^{2 q + 2} - 1)}{2 (\phi - 1) q (q + 1) (2 q + 1)} \right\}^{-\frac{1}{1+2q}} \kappa_q^{\frac{2}{1+2q}}, & \phi > 1; \\
		\left\{ \frac{2 q + 1}{q (q + 1)}
		- \frac{4 (2 q + 1)}{q (q + 1) (3 q + 2)}
		+ \frac{1}{q(2 q + 1)} \right\}^{-\frac{1}{1+2q}} \kappa_q^{\frac{2}{1+2q}}, & \phi = 1. \\
	\end{array}
	\right.
	\]
	In addition, if $\Theta = \rho \Psi_\star$ is allowed for any $\rho \in \R^+$, then the $\mathcal{L}^2$-optimal $\Theta$ is
	\[
	\Theta_\star = \left\{
	\begin{array}{ll}
		\left\{ \frac{(q+2) (3q+2) (\phi^{2q+2}-1)}{4 (2q+1)^2 (\phi^{q+2}-1)}
		+ \frac{\Psi_\star^{-2q-1} \kappa_q^2 (\phi-1) (q+1) (q+2) (3q+2)}{4 (2q+1) (\phi^{q+2}-1)} \right\}^{\frac{1}{q}} \Psi_\star, & \phi > 1; \\
		\left\{ \frac{(q+1) (3q+2)}{2 (2q+1)^2}
		+ \frac{\Psi_\star^{-2q-1} \kappa_q^2 (q+1) (3q+2)}{4 (2q+1)} \right\}^{\frac{1}{q}} \Psi_\star, & \phi = 1. \\
	\end{array}
	\right.
	\]
	For $O(1)$-space update, the $\mathcal{L}^2$-optimal $\phi$ is $\phi_\star = 2$.
	Otherwise, the $\mathcal{L}^2$-optimal $\phi$ is $\phi_\star = 1$.
\end{corollary}

Corollary \ref{coro:mse} allows us to compare the optimal mean squared errors of \hyperref[eq:LASER]{$\LASER(q,\phi)$} for different $q$ and $\phi$ with existing offline and online estimators in Table \ref{tab:properties}.
It confirms that the structural problem of existing online estimators can be solved by our framework.

\begin{remark} \label{rmk:mse}
	When $1/(1+2q) = \psi < \theta$, the order of $\MSE(\hat{\sigma}_n^2)$ is also optimized.
	However, the window in Definition \ref{def:winLASER} becomes rectangular asymptotically.
	Consequently, there is no known expression for the leading-order term of $\MSE(\hat{\sigma}_n^2)$ and the $\mathcal{L}^2$-optimal parameters cannot be found \citep{hoc}.
\end{remark}

\subsection{Practical Enhancements} \label{sec:practical-enhancements}

\begin{algo}[$\LASER(1,1)$] \label{algo:las-cdn}
	Initialize $\mathcal{C}_1=\{1, 0, 1, X_1^2, X_1, \{0, \ldots, 0\}_{b=0,1}, \{X_1\}\}$.
	At time $n\in\Z^+$, store $\mathcal{C}_n = \{n, s_n, t_n, Q_n, \bar{X}_n, \{K_{n,b}, R_{n,b}, k_{n,b}, r_{n,b}, U_{n,b}, V_{n,b}\}_{b=0,1}, \{X_{n-b}\}_{b=0,\ldots,s_n} \}$.
	At time $n+1$, update $\mathcal{C}_n$ to $\mathcal{C}_{n+1}$ by:
	
	\noindent
	\begin{minipage}[t]{19pc}
		\begin{enumerate}[label=(\alph*)]
			\item $s_{n+1} = \min\{ \lfloor \Psi (n+1)^\psi \rfloor, n\}$;
			\item $t_{n+1} = \min\{ \lceil \Theta (n+1)^\theta \rceil, n+1\}$;
			\item $Q_{n+1} = Q_n +X_{n+1}^2$;
			\item $\bar{X}_{n+1} = (n\bar{X}_n +X_{n+1})/(n+1)$;
			\item $K_{n+1,0} = K_{n,0} +X_n -X_{n-s_n} \I_{s_{n+1}=s_n}$;
			\item $K_{n+1,1} = K_{n,1} +K_{n+1,0} -s_n X_{n-s_n} \I_{s_{n+1}=s_n}$;
		\end{enumerate}
	\end{minipage}
	\begin{minipage}[t]{16pc}
		\vspace{0pt}
		\begin{enumerate}
			\item[(g)] for each $b = 0,1$,
			\begin{enumerate}[label=(\roman*), leftmargin=24pt]
				\item $R_{n+1,b} = R_{n,b} +X_{n+1} K_{n+1,b}$;
				\item $k_{n+1,b} = k_{n,b} +s_{n+1}^b \I_{s_{n+1}=1+s_n}$;
				\item $r_{n+1,b} = r_{n,b} +k_{n+1,b}$;
				\item $U_{n+1,b} = U_{n,b} +k_{n+1,b} X_{n+1}$;
				\item $V_{n+1,b} = V_{n,b} +K_{n+1,b}$.
			\end{enumerate}
		\end{enumerate}
	\end{minipage}
	
	\[
	\text{Output:} \quad
	\hat{\sigma}_{n+1,\LASER(1,1)}^2
	= \frac{Q_{n+1} +2R_{n+1}^* +(2r_{n+1}^*-n-1)\bar{X}_{n+1}^2 -2\bar{X}_{n+1}(U_{n+1}^*+V_{n+1}^*)}{n+1},
	\]
	where $D_{n+1}^* = D_{n+1,0}-D_{n+1,1}/t_{n+1}$ for $D \in \{R,r,U,V\}$, e.g., $r_{n+1}^* = r_{n+1,0} -r_{n+1,1}/t_{n+1}$.
\end{algo}

Algorithm \ref{algo:las-cdn} shows that our recommended estimator can be updated in $O(1)$ time.
To further address users' needs, we enhance our proposal with mini-batch and automatic updates.

In some online problems, users are not interested in every estimate but estimates at time $n_1, n_2, \ldots$ instead.
However, none of the existing works address this practical need.
In light of it, we provide mini-batch update, which consists of two algorithmic modifications:

\begin{enumerate}[label=(\alph*)]
	\item removing redundant operations: we take Algorithm \ref{algo:las-cdn} as an example.
	When \hyperref[eq:LASER]{$\LASER(1,1)$} is updated at time $n_{j+1}$ from $n_j$, some intermediate statistics need to be computed at all $n_j+1, n_j+2, \ldots, n_{j+1}$.
	For example, $K_{i,0}$ should be updated by
	\[
	K_{i,0} = K_{i-1,0} +X_{i-1} -X_{i-s_i} \I_{s_i=s_{i-1}} \qquad \text{for} \qquad i=n_j+1, n_j+2, \ldots, n_{j+1}
	\]
	as they are used to compute other statistics, e.g., $R_{n_{j+1},0}$.
	However, other statistics that are only related to the output can be computed at time $n_{j+1}$ directly.
	\item applying vectorized operations: computation time can be further reduced through vectorization; see, e.g., \citet{wickham2019r}.
	While offline estimators can also be vectorized, online estimators have an edge that each update is computed from $n_j$ to $n_{j+1}$ only.
	In contrast, offline estimators may need to compute an update from $1$ to $n_{j+1}$.
\end{enumerate}

Another enhancement that we provide is automatic update.
By Corollary \ref{coro:mse}, the $\mathcal{L}^2$-optimal parameters depend on $\kappa_q = |v_q|/\sigma^2$ for some $q \in \Z^+$.
To choose $q$ accounting for the empirical serial dependence, we recommend $q=1$
($q=3$) for strongly (weakly) auto-correlated data.
Since an online estimate of $\sigma^2$ is available from the last iteration, it remains to handle $v_q = \sum_{k \in \Z} |k|^q \gamma_k$.
Classical methods such as parametric plug-in \citep{andrews1991kernel} are computationally inefficient.
While we can apply them on a pilot sample \citep{wu2009recursive}, it may be shortsighted in some cases.
To fully utilize the increasing sample size in online problems, a natural solution is to also estimate $v_q$ online.
Consider
\[
\hat{v}_{q,n} = \hat{v}_{q,n}(W)
= \frac{1}{n} \sum_{i=1}^n \sum_{j=1}^n W_n(i, j) |i-j|^q (X_i -\bar{X}_n) (X_j -\bar{X}_n).
\]
Since $\hat{v}_{q,n}$ has a similar form as \hyperref[eq:genClass]{$\hat{\sigma}^2_n$}, our framework suggests
\[
\hat{v}_{q,n,\LASER(p,\phi)}
= \hat{v}_{q,n}(W_{\LASER(p,\phi)}),
\]
where $W_{\LASER(p,\phi)}$ takes the form in \cref{eq:LASER}, $q$ and $\phi$ are inherited from \hyperref[eq:LASER]{$\hat{\sigma}_{n,\LASER(q,\phi)}^2$}, and $p$ is recommended to be 1.
By Proposition \ref{prop:on}, $\hat{v}_{q,n,\LASER(p,\phi)}$ can be updated in the same time and space complexities as \hyperref[eq:LASER]{$\hat{\sigma}_{n,\LASER(q,\phi)}^2$}.
Updating \hyperref[eq:LASER]{$\hat{\sigma}_{n,\LASER(q,\phi)}^2$} becomes automatic in the sense that users only need to supply the incoming observations,
but the smoothing parameters are optimally selected and the computational properties are preserved.
We implement this idea in our R-package and compare it with existing methods below.

\begin{example}[Smoothing parameter selection] \label{eg:auto}
	Consider the best existing work $\PSR$, and \hyperref[eq:LASER]{$\LASER(1,1)$} with different smoothing parameter selectors.
	
	\begin{enumerate}[label=(\alph*)]
		\item Best existing: $\PSR$ with the automatic update in \citet{rtacm} is included.
		\item Pilot: $\kappa_1$ is estimated with the first $500$ observations.
		Bartlett kernel estimators with asymptotically rate-optimal bandwidths $\lceil n^{1/5} \rceil$ (for $v_1$) and $\lceil n^{1/3} \rceil$ (for $\sigma^2$) are used.
		\item Auto: $\kappa_1$ is updated with the automatic procedure in this paper.
		\item Oracle: the theoretical value of $\kappa_1$, which is unknown in practice, is used.
	\end{enumerate}
	
	Under the bilinear model in \cref{eq:bilinear}, we obtain the trajectories of estimates of a typical realization with $10^5$ observations.
	We also compute the efficiency gained (in terms of $\sigma^2/\MSE(\cdot)$) for $m=500$ (mini-batch size) based on $1000$ replications.
	Figure \ref{fig:update} shows that the auto trajectory is very close to that of the oracle.
	Meanwhile, the pilot trajectory is obviously off from the true value.
	This is because realizations of \cref{eq:bilinear} are strongly autocorrelated, so $500$ observations and asymptotically rate-optimal bandwidths may be insufficient to obtain a good estimate of $\kappa_1$.
	In practice, a sufficient pilot sample size is unknown so our automatic update is advantageous compared with existing works.
\end{example}

\begin{figure}[!t]
	\centering
	\includegraphics[width=\linewidth]{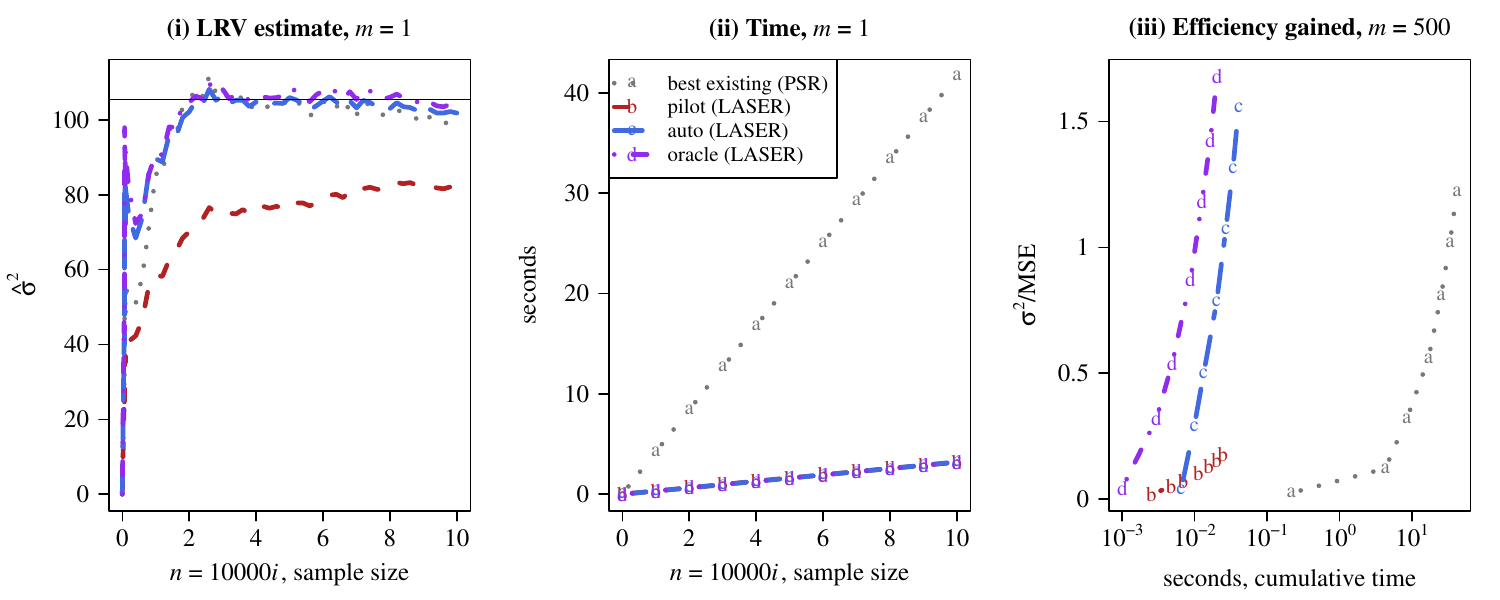}
	\caption{Estimation of $\sigma^2$ using different smoothing parameter selectors: (a) best existing (dotted gray); (b) pilot (dashed red); (c) auto (longdash blue); (d) oracle (dotdash purple).
		The horizontal line in plot (i) is the true value of $\sigma^2$ of \cref{eq:bilinear}.
	}
	\label{fig:update}
\end{figure}

\section{Beyond Variance Estimation} \label{sec:beyond-variance-estimation}

On top of efficient variance estimation, we can use the framework in \cref{sec:general-framework} to directly study online inference procedures.
In \cref{sec:online-quantile-regression}, we demonstrate how to improve an online quantile regression method in a principle-driven way.
In \cref{sec:online-change-point-detection}, we complement existing online change point detectors in time series by changing the standard approach in the field.
In \cref{sec:online-convergence-diagnosis,sec:online-optimization}, we illustrate the strength and flexibility of our proposals in Markov chain Monte Carlo convergence diagnosis and
stochastic approximation.

\subsection{Online Quantile Regression} \label{sec:online-quantile-regression}

In quantile regression, a classical estimator is $\tilde\beta_{\tau,n} = \argmin_{\beta \in \R^{d+1}} \sum_{i=1}^n \rho_\tau(y_i -x_i^\T \beta)$, where $\rho_\tau(x) = x(\tau -\I_{x \le 0})$ is the check function, $y_i$ is the response, $x_i^\T = (1, x_{i1}, \ldots, x_{id})$ is the covariate, and $\tau$ is the quantile level.
To handle a memory constraint of size $m$ in a distributed or online setting, \citet{chen2019rindicator} proposed a linear estimator that achieves the same asymptotic efficiency as $\tilde\beta_{\tau,n}$.
We review their online setting, where observations are divided into intervals by their indices.
Define $b_0 = 0$, $c_0 = -\infty$, and $c_{2k-1} = 2^{k-1} +1/2$ and $c_{2k} = 2^{k-1} +3/4$ for $k \ge 1$.
Let the first index in the $l$-th interval be $b_l = \lfloor m^{c_{l-1}} \rfloor +1$,
where $\{c_l\}$ was chosen such that $b_{l+1}-b_l$ is approximately $(b_{l-1}-b_{l-2})^2$ \citep{chen2019rindicator}.
When the $n$-th observation is in the $l$-th interval, their proposal is
\begin{equation} \label{eq:quant-original}
	\check\beta_{\tau,n} = \left( \sum_{i=b_{l-1}}^{b_l-1} \check{V}_{i,l-1} +\sum_{i=b_l}^n \check{V}_{i,l} \right)^{-1} \left( \sum_{i=b_{l-1}}^{b_l-1} \check{U}_{i,l-1} +\sum_{i=b_l}^n \check{U}_{i,l} \right),
\end{equation}
where
\begin{align*}
	\check{U}_{i,j} &= x_i \left\{ H\left(\frac{y_i -x_i^\T \check\beta_{\tau,b_j-1}}{h_j}\right) +\tau -1 +\frac{y_i}{h_j} H'\left(\frac{y_i -x_i^\T \check\beta_{\tau,b_j-1}}{h_j}\right) \right\}, \\
	\check{V}_{i,j} &= \frac{x_i x_i^\intercal}{h_j} H'\left(\frac{y_i -x_i^\T \check\beta_{\tau,b_j-1}}{h_j}\right),
\end{align*}
$H(\cdot)$ is a smooth approximation of $\I_{\cdot > 0}$,
$H'(\cdot)$ is the gradient of $H(\cdot)$,
$h_1 = \sqrt{d/m}$ and $h_j = \sqrt{d/m^{c_{j-1}}}$ for $j \ge 2$ are bandwidths, and
$\check\beta_{\tau,0}$ is the classical estimator computed with an initial sample of size $m$.

The linear estimator \hyperref[eq:quant-original]{$\check\beta_{\tau,n}$} is innovative but we can refine it following the ideas in this paper.
First, \hyperref[eq:quant-original]{$\check\beta_{\tau,n}$} mainly uses observations in the $(l-1)$-th and $l$-th intervals.
Earlier observations are not treated alike, which suggests that we can improve the finite-sample statistical properties by following Principle \hyperref[prin:LASER]{A}.
Second, \hyperref[eq:quant-original]{$\check\beta_{\tau,n}$} only utilizes the memory constraint of size $m$ during initialization.
We can reduce computational time further by considering mini-batch updates.
Therefore, we modify \cref{eq:quant-original} into
\begin{equation} \label{eq:quant-modified}
	\hat\beta_{\tau,n} = \left( \sum_{j=1}^{k-1} \sum_{i=a_j}^{a_{j+1}-1} \hat{V}_{i,j} +\sum_{i=a_k}^n \hat{V}_{i,k} \right)^{-1} \left( \sum_{j=1}^{k-1} \sum_{i=a_j}^{a_{j+1}-1} \hat{U}_{i,j} +\sum_{i=a_k}^n \hat{U}_{i,k} \right),
\end{equation}
where
\begin{align*}
	\hat{U}_{i,j} &= x_i \left\{ H\left(\frac{y_i -x_i^\T \hat\beta_{\tau,a_j-1}}{h_j}\right) +\tau -1 +\frac{y_i}{h_j} H'\left(\frac{y_i -x_i^\T \hat\beta_{\tau,a_j-1}}{h_j}\right) \right\}, \\
	\hat{V}_{i,j} &= \frac{x_i x_i^\intercal}{h_j} H'\left(\frac{y_i -x_i^\T \hat\beta_{\tau,a_j-1}}{h_j}\right),
\end{align*}
$k = \lfloor n/m \rfloor$ is the current number of intervals or mini-batches,
$a_0 = 0$, $a_j = \lfloor (j-1)m \rfloor +1$ is the first index in the $j$-th interval,
$h_j = \sqrt{d/(jm)}$ for $j \ge 1$ are bandwidths, and
$\hat\beta_{\tau,0} = \check\beta_{\tau,0}$ is the same initial estimator.
To compare the performance, we follow \citet{chen2019rindicator} and generate data from $y_i = x_i^\T \beta +\varepsilon_i$, where $\beta$ is a vector of ones, $d=10$,
\begin{itemize}
	\item independent case: $(x_{i1}, \ldots, x_{id})^\T$ follows a multivariate uniform distribution on $[0,1]^d$ with $\Corr(x_{ij}, x_{ik}) = 0.5^{|j-k|}$, and $\varepsilon_i \sim \Normal(0,1)$; or
	\item dependent case: $(x_{i1}, \ldots, x_{id}, \varepsilon_i)^\T$ follows a first-order vector autoregressive model with mean $(0, \ldots, 0)^\T$, coefficient matrix $\diag(0.5, \ldots, 0.5)$ and Gaussian noise covariance matrix $\Sigma_{jk} = 0.5^{|j-k|} \I_{j \ne d+1} \I_{k \ne d+1} +\I_{j = d+1} \I_{k = d+1}$ for $1 \le j, k \le d+1$.
\end{itemize}
The independent case exactly replicates a simulation in \citet{chen2019rindicator}, while the dependent case allows further comparison of different variance estimators.
Thus, we use the same $H(x) = \{1/2 +(15/16)(x -2x^3/3 +x^5/5)\} \I_{|x| < 1} +\I_{x \ge 1}$, $m=500$ and $\tau=0.1$.
Consider

\begin{enumerate}[label=(\alph*)]
	\item original and \hyperref[eq:welford]{$\WFD$}: the original method in \citet{chen2019rindicator} is used (see \cref{eq:quant-original}) with Welford's algorithm.
	\item modified and $\PSR$: (a) is modified by replacing $\check\beta_{\tau,n}$ in \cref{eq:quant-original} with $\hat\beta_{\tau,n}$ in \cref{eq:quant-modified}.
	$\PSR$ is used to account for possible serial dependence.
	\item modified and \hyperref[eq:LASER]{$\LASER(1,1)$}: same as (b) but \hyperref[eq:LASER]{$\LASER(1,1)$} is used instead.
\end{enumerate}

We compare point estimates and $95\%$ confidence intervals for the summed coefficient $J^\T \beta_\tau$ based on $1000$ replications,
where $J$ is a vector of ones and $\beta_\tau$ is computed by shifting $\varepsilon_i$ such that $\pr(\varepsilon_i \le 0 \mid x_i) = \tau$;
see the simulations in \citet{chen2019rindicator}.
For $\tau = 0.1$, the intercept $\beta_{\tau 1}$ in the independent and dependent cases are approximately $-0.28$ and $-0.48$, respectively.
Figure \ref{fig:quant} reports the results in the dependent case.
Despite \hyperref[eq:quant-original]{$\check\beta_{\tau,n}$} is asymptotically efficient, our modified estimator \hyperref[eq:quant-modified]{$\hat\beta_{\tau,n}$} outperforms it in a moderately large sample ($n \le 10^5$).
Our long-run variance estimator also allows practical autocorrelation-robust inference, whereas the best existing method $\PSR$ does not adapt to mini-batch inputs.
The results in the independent case are similar and thus deferred to the \Supp.

\begin{figure}[!t]
	\centering
	\includegraphics[width=\linewidth]{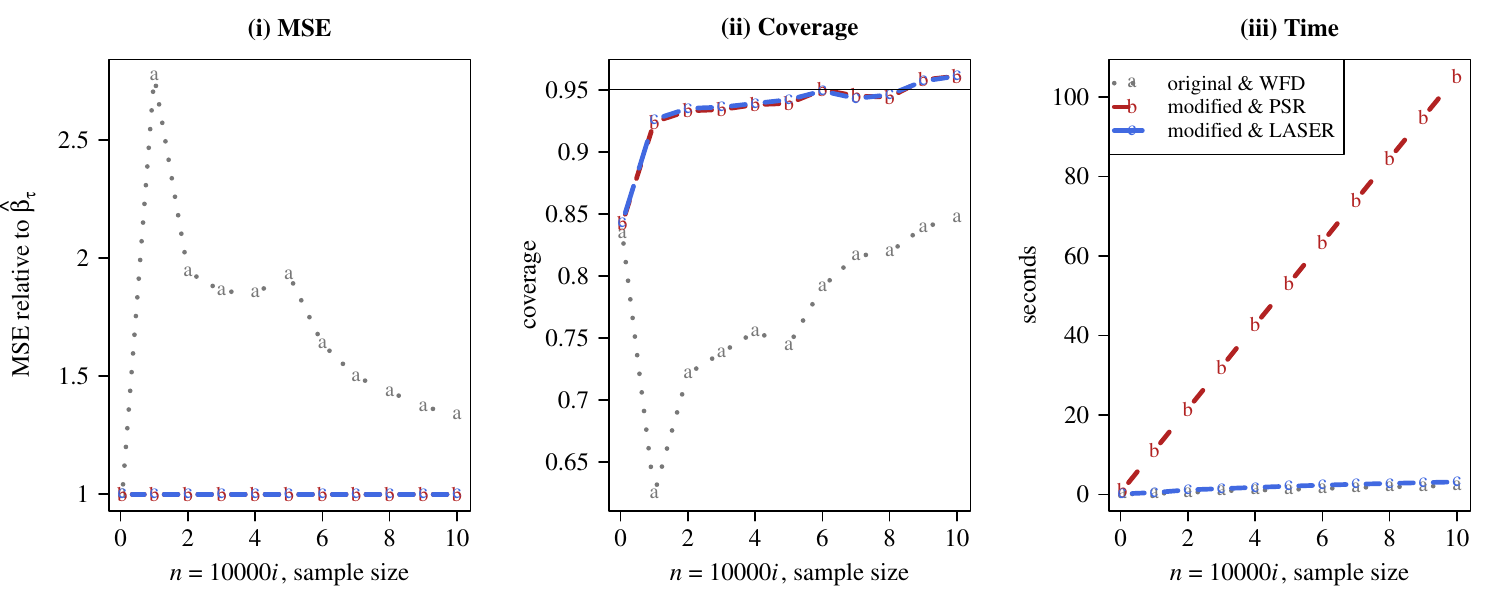}
	\caption{Online quantile regression under dependence using original and modified inference procedures at $5\%$ nominal size: 
		(a) original and \hyperref[eq:welford]{$\WFD$} (dotted gray); 
		(b) modified and $\PSR$ (dashed red); 
		(c) modified and \hyperref[eq:LASER]{$\LASER(1,1)$} (longdash blue).
		The plot in (i) shows the value of $\MSE(\cdot)/\MSE(\hat\beta_{\tau,n})$.}
	\label{fig:quant}
\end{figure}

\subsection{Online Change Point Detection} \label{sec:online-change-point-detection}

Recently, \citet{gosmann2021ocp} developed a new approach for online change point detection in an open-end scenario.
In their outlook, they mentioned that the standard approach in the field estimated the long-run variance with the initial data only,
but updating the estimate was logical particularly for stronger dependent models \citep{gosmann2021ocp}.
We echo their view here.
Consider the bilinear model in \cref{eq:bilinear}.
Let $\mu_i = \E(X_i)$; $k^* \in \Z^+$; and $\nu$ be the initial sample size.
We are interested in testing
\[
H_0: \mu_1 = \cdots = \mu_\nu = \mu_{\nu+1} = \cdots \quad \text{against} \quad
H_1: \mu_1 = \cdots = \mu_{\nu+k^*-1} \neq \mu_{\nu+k^*} = \mu_{\nu+k^*+1} = \cdots.
\]
In this case, the online change monitoring scheme in \citet{gosmann2021ocp} is
\[
\hat{E}_\nu(k) = \nu^{-1/2} \max_{0 \le j \le k-1} (k-j) \left| \bar{X}_{1,\nu+j} -\bar{X}_{\nu+j+1,\nu+k} \right|/\hat{\sigma},
\]
where $\bar{X}_{a,b} = (b-a+1)^{-1} \sum_{i=a}^b X_i$ and $\hat{\sigma}^2$ is a long-run variance estimator.
Following examples in \citet{gosmann2021ocp}, we use the same threshold function $w(t) = (1+t)^{-1}$, nominal size $\alpha=0.05$ and stopping point $n^*=4000$ to monitor $w(k/\nu) \hat{E}_\nu(k)$.
We also simulate $H_1$ by $X_i^{(\delta)} = X_i +\delta \I_{i \ge \nu+k^*}$ with $100$ burn-in and $\nu=400$ initial observations.
However, consider

\begin{enumerate}[label=(\alph*)]
	\item fix: the long-run variance is estimated with the initial data only, which is the standard approach.
	The \hyperref[eq:bart]{$\hat{\sigma}^2_{\nu,\bart}$} implemented in the R-package \texttt{sandwich} \citep{sandwich} without prewhitening and adjustment is used.
	\item offline: the long-run variance estimate is computed as in (a).
	However, the estimate is updated as data arrive.
	For robustness, the first-order difference statistics in \citet{dlrv} instead of the raw data are used.
	In favor of offline update, the smoothing parameter is first selected based on the initial data and then scaled according to \citet{dlrv}; see Remark \ref{rmk:dlrv}.
	\item online: the long-run variance estimate is updated as in (b) but using \hyperref[eq:LASER]{$\LASER(1,1)$}.
\end{enumerate}

We conduct the simulation for $\nu+k^* = 601, 1001, 1401$ and $\delta = 0, 0.5, \ldots, 5$ each with $1000$ replications.
Figure \ref{fig:cp} reports the results for $\nu+k^* = 1401$.
When the long-run variance estimate is updated, the type I error is considerably closer to $5\%$.
Furthermore, the positive predictive value improves substantially and is the best using online with negligible time cost.
The results for $\nu+k^* = 601, 1001$ are similar and so deferred to the \Supp.

\begin{figure}[!t]
	\centering
	\includegraphics[width=\linewidth]{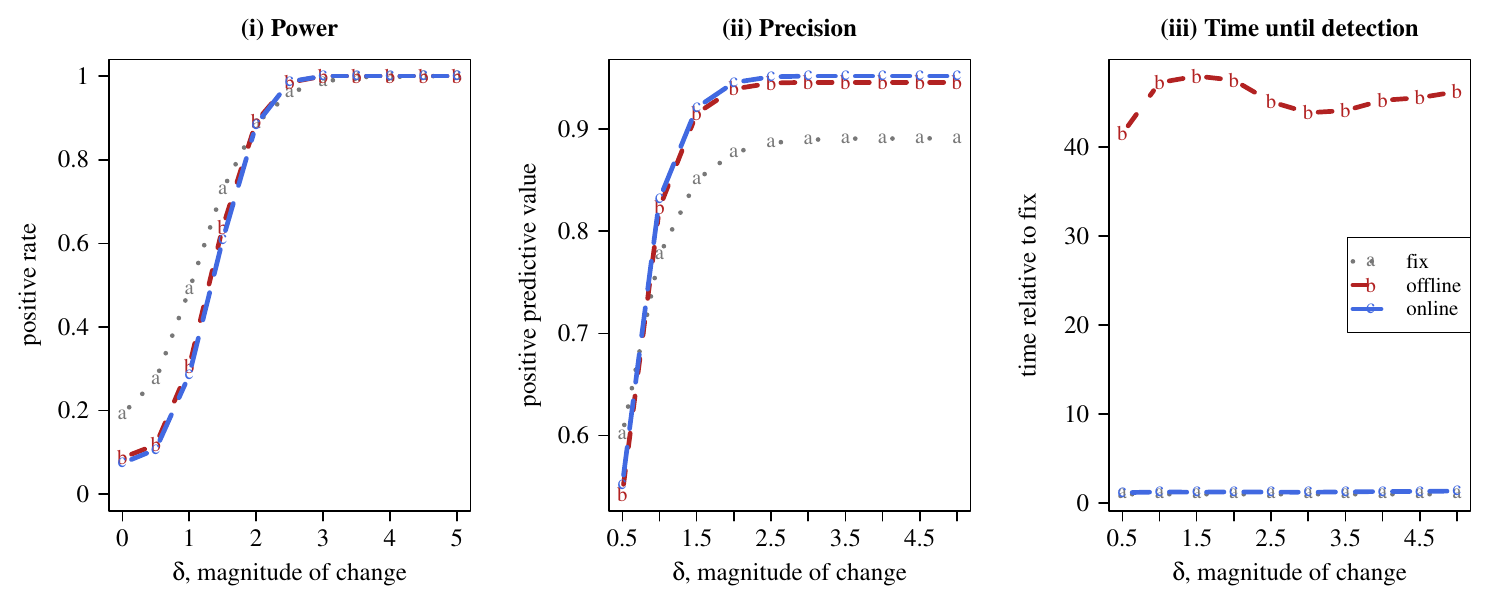}
	\caption{Online change point detection at $5\%$ nominal size using different long-run variance estimation methods: (a) fix (dotted gray); (b) offline (dashed red); (c) online (longdash blue).}
	\label{fig:cp}
\end{figure}

\begin{remark} \label{rmk:dlrv}
	The smoothing parameter selection here is favorable to offline because it reduces the time cost of offline much more than online.
	The scaling factor in \citet{dlrv} due to the use of difference statistics is also optimal to offline only.
	Investigating the scaling factor for online is of interest but beyond the scope here.
\end{remark}

\subsection{Online Convergence Diagnosis} \label{sec:online-convergence-diagnosis}

In Markov chain Monte Carlo methods, a common goal is to estimate $\E_{\zeta \sim \pi}\{h(\zeta)\}$ using $\bar{h}_n = n^{-1} \sum_{i=1}^n h(X_i)$ by generating a Markov chain $\{X_i\}$ that satisfies certain conditions for some distribution $\pi$ and target function $h$ \citep{flegal2010obm}.
However, the terminal sample size $n^*$ for a reasonably precise $\bar{h}_{n^*}$ is unknown a priori.
In light of it, \citet{jones2006fw} proposed the half-width test to terminate a simulation at
\[
n^* = \inf\left\{ n \in \Z^+: z_{1-\alpha/2} \hat{\sigma}_n/\sqrt{n} +p(n) < \epsilon \right\},
\]
where $\alpha \in (0,1)$ is the significance level,
$z_{1-\alpha/2}$ is the $100(1-\alpha/2)\%$ lower quantile of $\Normal(0,1)$,
$p(n)$ is a penalty function for $n$ that is too small, and
$\epsilon > 0$ is the maximum tolerable error.
Since $\bar{h}_n$ can be updated in $O(1)$ space, it is sensible to update its variance estimate also in $O(1)$ space, which we demonstrate with a classical example in \citet{hastings1970mcmc}.
To sample from $\Normal(0,1)$, \citet{hastings1970mcmc} used a Metropolis--Hastings algorithm with a random walk on $[-\delta, \delta]$ as the proposal.
Given a sample $\{ X_i \}_{i=1}^n$ generated with $\delta = 1$, $h(x) = x^2$, $\alpha = 0.05$, and $p(n) = \epsilon \I_{n < 1000}$,
we conduct the half-width test every $m=100$ more observations for $\epsilon=0.12, 0.11, \ldots, 0.02$ and $1000$ replications using different estimators (``R-packages''):

\begin{enumerate}[label=(\alph*)]
	\item \hyperref[eq:obm]{$\obm$} (``\texttt{mcmcse}''): the long-run variance is estimated by the \texttt{mcse} function \citep{mcmcse}.
	$\obm$ with lugsail parameter \texttt{r=1} is used.
	\item $\TSR$ (``\texttt{rTACM}''): the long-run variance is estimated by the \texttt{rTACM} function \citep{rtacm}. The best existing $O(1)$-space (and $O(1)$-time) estimator $\TSR$ is used.
	\item \hyperref[eq:LASER]{$\LASER(1,2)$} (``\texttt{rlaser}''): the long-run variance is estimated by the \texttt{lrv} function.
\end{enumerate}

Unmentioned arguments, such as the smoothing parameter selector, are left as default in all R-functions.
Figure \ref{fig:hastings} shows that the coverage rates using \hyperref[eq:LASER]{$\LASER(1,2)$} are comparable to \hyperref[eq:obm]{$\obm$} and generally better than $\TSR$, which are in line with Table \ref{tab:properties}.
While both $\TSR$ and \hyperref[eq:LASER]{$\LASER(1,2)$} use a constant amount of memory, $\TSR$ is only faster than \hyperref[eq:obm]{$\obm$} for large $n$ because it lacks mini-batch updates.
In contrast, \hyperref[eq:LASER]{$\LASER(1,2)$} is always faster than \hyperref[eq:obm]{$\obm$}.

\begin{figure}[!t]
	\centering
	\includegraphics[width=\linewidth]{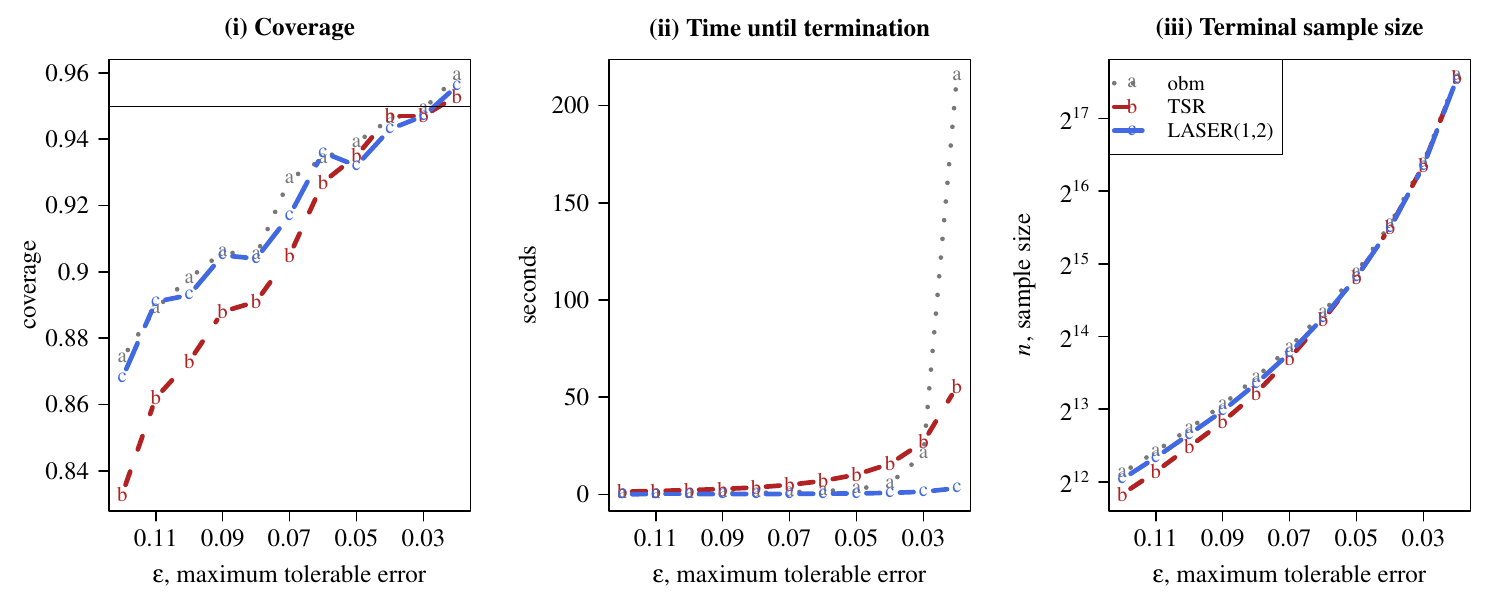}
	\caption{Online $95\%$-half-width test using different long-run variance estimators: (a) \hyperref[eq:obm]{$\obm$} (dotted gray); (b) $\TSR$ (dashed red); (c) \hyperref[eq:LASER]{$\LASER(1,2)$} (longdash blue).}
	\label{fig:hastings}
\end{figure}

\subsection{Online Optimization} \label{sec:online-optimization}

In machine learning, a classical problem is to estimate the model parameter $w$ based on \linebreak
$\min_{w \in \R^d} \E_{\zeta \sim \Pi} \{ h(w, \zeta)\}$,
where $h(w, \zeta)$ is a noisy measurement of the loss \citep{zhu2023sgd}.
The stochastic gradient descent algorithm proceeds with $w_i = w_{i-1} -\eta_i g_{i-1}$, where $w_0$ is the initial point,
$\eta_i$ is the learning rate at the $i$-th step, and
$g_i$ is the gradient of $h(w, \zeta_i)$ with respect to $w$ \citep{zhu2023sgd}.
To quantify the uncertainty after averaging, \citet{zhu2023sgd} modified \citet{wu2009recursive}.
To be specific, let $\eta_i = \eta_0 i^{-\alpha}$ for some $\eta_0 > 0$ and $\alpha \in (1/2, 1)$.
Then, the $k$-th block in Wu's estimator starts at $a_k = \lfloor A k^{2/(1-\alpha)} \rfloor$ for some $A > 0$.

On the other hand, it is certain that $\TSR$ and $\PSR$ dominate Wu's estimator under stationarity \citep{rtacm}.
Our Example \ref{eg:on} further reveals the structural problem of subsample selection rules.
Wondering whether similar results hold for averaged stochastic gradient descent, we consider a logistic regression example.
Let $\{\zeta_i = (x_i, y_i)\}_{i \in \Z^+}$ and $w^*$ denote a sequence of data and the true model parameter,
where $x_i \sim \Normal(0, I_d)$,
$y_i \mid x_i \sim \Ber(\{1+\exp(-x_i^{\T} w^*)\}^{-1})$,
$w^*$ is a $d$-dimensional vector linearly spaced between $-1$ and $1$, and $d=5$.
The loss function is $h(w, x_i, y_i) = (1-y_i) x_i^{\T} w +\ln\{1+\exp(-x_i^{\T} w)\}$.
The mini-batch size $m$ determines the update frequency of the variance estimate.
Following examples in \citet{zhu2023sgd}, we choose $\eta_0 = 0.5$, $\alpha = 0.505$ and $A = 1$ for:

\begin{enumerate}[label=(\alph*)]
	\item Wu's estimator (``\texttt{rTACM}''): $\TSR$ with $\xi=0$ and manual update is used, which is essentially \citeauthor{wu2009recursive}'s \citeyearpar{wu2009recursive} estimator in \citet{zhu2023sgd} as $a_k$ is set to be the same.
	\item $\TSR$ (``\texttt{rTACM}''): $\TSR$ with $\xi=1$ and manual update is used, which is proven to be more efficient than Wu's estimator under stationarity.
	\item online $\LASER$ (``\texttt{rlaser}''): ``$\LASER$'' with $q=1$ and $m=1$ is used.
	While \cref{eq:rampSub} does not apply, blocking with $a_k$ is compatible with the form in \cref{eq:winDecom}.
	In the $k$-th block, the subsampling parameter increases by 1 per iteration until it reaches the block maximum $\lceil 2A (1-\alpha)^{-1} k^{(1+\alpha)/(1-\alpha)} \rceil$, where it remains unchanged until the next block, and the tapering parameter is chosen to be the block maximum scaled according to Corollary \ref{coro:mse}.
	\item mini-batch $\LASER$ (``\texttt{rlaser}''): the estimator in (b) is used except that $m = 10^4$.
\end{enumerate}

\begin{figure}[!t]
	\centering
	\includegraphics[width=\linewidth]{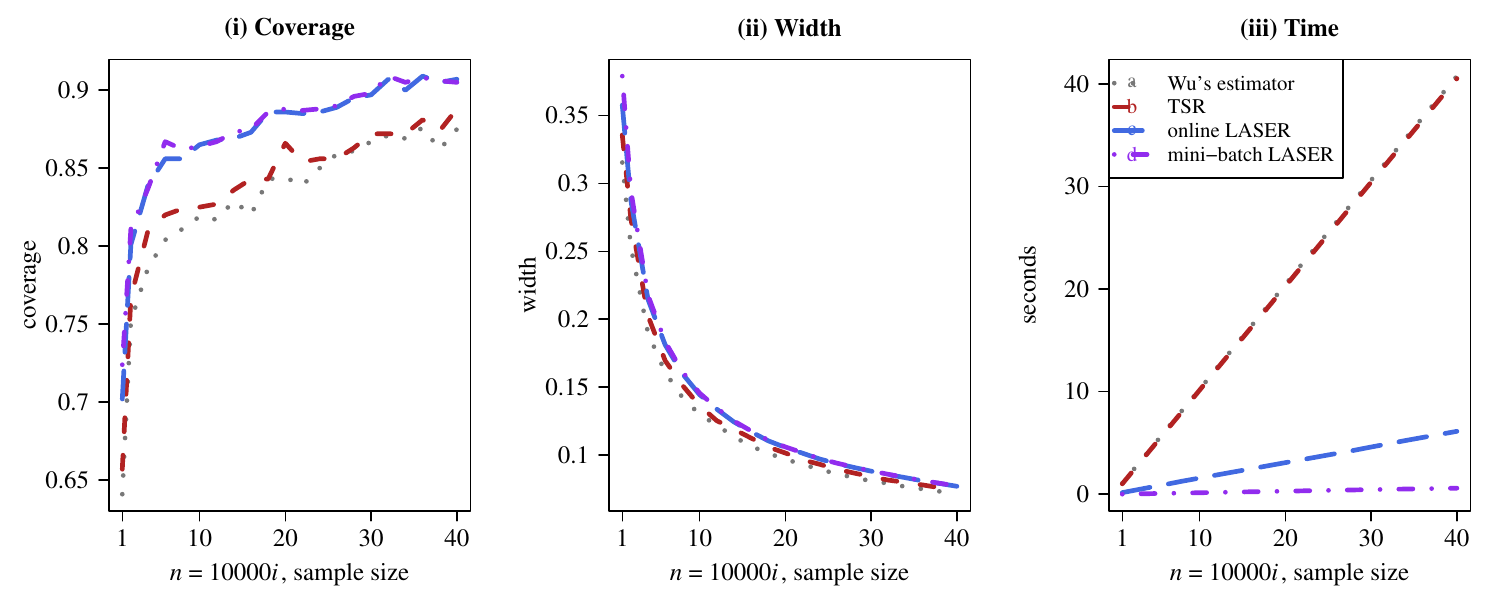}
	\caption{Confidence intervals for the summed stochastic gradient descent estimands using different variance estimators: (a) Wu's estimator (dotted gray); (b) $\TSR$ (dashed red); (c) online $\LASER$ (longdash blue); (d) mini-batch $\LASER$ (dotdash purple).}
	\label{fig:sgd_ci}
\end{figure}

We compare the above in constructing $95\%$ confidence intervals for the summed coefficient $J^{\T} w^*$ with a burn-in stage \citep{chen2020sgd} of size $500$ based on $1000$ replications,
where $J$ is a vector of ones.
In Figure \ref{fig:sgd_ci}, the performance of Wu's estimator is consistent with the studies reported in \citet{zhu2023sgd}.
Furthermore, there are notable improvements according to the theory under stationarity.
Real-time inference is possible using mini-batch updates without loss of statistical efficiency.
These results indicate that our framework can be used as a starting point for new research.
Finally, we remark that these intervals can be used as convergence diagnostics; see \cref{sec:online-convergence-diagnosis}.

\section{Discussion} \label{sec:discussion}

\citet{ljung1983recursive} mentioned that ``many people regard the area (of recursive identification/estimation) to be a `bag of tricks' rather than a theory''.
Indeed, existing online estimators are often found by sacrificing some statistical properties and studied on an algorithm-by-algorithm basis.
Contrary to the first leading approach studied in \citet{wu2009recursive},
we presented a novel framework to study long-run variance estimators in a principle-driven way.
To the best of our knowledge, this is the first attempt to characterize the form of online long-run variance estimators through sufficient conditions.
The study of necessary conditions would be a theoretically important future direction.
On the other hand, the proposed estimator achieved high performance in a wide range of applications.
The proposed principles and ideas also applied to problems beyond variance estimation.
Two examples are the estimation of $v_q$ and online quantile regression in \cref{sec:practical-enhancements,sec:online-quantile-regression}, respectively.
Other problems such as spectral density estimation \citep{xiao2011recursive} and nonparametric regression \citep{huang2014recursive} would be interesting for further investigation.

\section*{Acknowledgments}

The authors would like to thank Zhuohua Shen for his kind review and suggestions.

\section*{Funding}

The second author was supported in part by grants
GRF-14304420, GRF-14306421, and GRF-14307922
provided by the Research Grants Council of HKSAR.

\bibliographystyle{rss}
\bibliography{bib/learning,bib/lrv,bib/mcmc,bib/quant,bib/sgd,bib/ts}
\addcontentsline{toc}{section}{\refname}

\clearpage

\appendix

\section{Algorithms} \label{sec:algorithms}
\subsection{Overview of Notation} \label{sec:algo-overview}

In this subsection, we state and explain the notation that we are going to use throughout \cref{sec:algorithms}.
Most of them will be recalled later when they are used.
We begin with some simple components:
\[
\begin{aligned}
	K_{n,b} &= \sum_{k=1}^{s_n} k^b X_{n-k}, \\
	R_{n,b} &= \sum_{i=1}^n \sum_{k=1}^{s_i} k^b X_i X_{i-k}, \\
	U_{n,b} &= \sum_{i=1}^n \sum_{k=1}^{s_i} k^b X_i, \\
	Q_n &= \sum_{i=1}^n X_i^2,
\end{aligned}
\quad
\begin{aligned}
	k_{n,b} &= \sum_{k=1}^{s_n} k^b, \\
	r_{n,b} &= \sum_{i=1}^n \sum_{k=1}^{s_i} k^b, \\
	V_{n,b} &= \sum_{i=1}^n \sum_{k=1}^{s_i} k^b X_{i-k}, \\
	\bar{X}_n &= \frac{1}{n} \sum_{i=1}^n X_i,
\end{aligned}
\]
where $b \in \N = \{0, 1, \ldots\}$.
Here, the first subscript $n$ indicates the sample size and the second subscript $b$ indicates the exponent in the summation, which is usually related to the characteristic exponent $q$.
We capitalize the variables when they depends on the observations $X_1, \ldots, X_n$.
To lighten the notation in the algorithms, we define a shorthand
\[
R_n^{(a,b,c)} = R_{n,a} -\frac{R_{n,b}}{t_n^c}, \quad
R_n^{(a,b)} = R_n^{(a,b,b)} = R_{n,a} -\frac{R_{n,b}}{t_n^b},
\]
and similarly for $r_n^{(a,b,c)}, U_n^{(a,b,c)}, V_n^{(a,b,c)}$, e.g.,
$U_n^{(a,b,c)} = U_{n,a} -U_{n,b}/t_n^c$.
After we introduce the characteristic exponent $q$, we need the backward finite difference operator $\nabla^{(b)} \cdot$, which is defined by
\[
\nabla^{(1)} f(k) = f(k) -f(k-1) \quad \text{and} \quad
\nabla^{(b)} f(k) = \nabla^{(b-1)} f(k) -\nabla^{(b-1)} f(k-1),
\]
where $b \in \Z^+$ and $f$ is a function that takes an integer input $k$.
Then, we have
\[
d_{k,q}^{(b)} = \nabla^{(b)} k^q \quad \text{and} \quad
D_{n,q}^{(b)} = \sum_{k=1}^{s_n} d_{k,q}^{(b)} X_{n-k},
\]
where $b = 1, \ldots, q$.
Here, the superscript $(b)$ indicates the number of backward finite differences, which is used in $d_{k,q}^{(b)}$ and $D_{k,q}^{(b)}$ only.
When the memory parameter $\phi > 1$, the intended subsampling parameter $s_i$ is ramped up as $s'_i$ defined in \cref*{eq:rampSub}.
The components become
\[
\begin{aligned}
	K'_{n,b} &= \sum_{k=1}^{s'_n} k^b X_{n-k}, \\
	R'_{n,b} &= \sum_{i=1}^n \sum_{k=1}^{s'_i} k^b X_i X_{i-k}, \\
	U'_{n,b} &= \sum_{i=1}^n \sum_{k=1}^{s'_i} k^b X_i,
\end{aligned}
\quad
\begin{aligned}
	k'_{n,b} &= \sum_{k=1}^{s'_n} k^b, \\
	r'_{n,b} &= \sum_{i=1}^n \sum_{k=1}^{s'_i} k^b, \\
	V'_{n,b} &= \sum_{i=1}^n \sum_{k=1}^{s'_i} k^b X_{i-k},
\end{aligned}
\]
where $b \in \N$.
Since the prime symbol is used in the ramped subsampling parameter $s'_i$, we reserve it for the corresponding components.
The shorthand becomes
\[
R_n^{\prime(a,b,c)} = R'_{n,a} -\frac{R'_{n,b}}{t_n^c}, \quad
R_n^{\prime(a,b)} = R_n^{\prime(a,b,b)} = R'_{n,a} -\frac{R'_{n,b}}{t_n^b},
\]
and similarly for $r_n^{\prime(a,b,c)}, U_n^{\prime(a,b,c)}, V_n^{\prime(a,b,c)}$,
e.g., $U_n^{\prime(a,b,c)} = U'_{n,a} -U'_{n,b}/t_n^c$.
To perform $O(1)$-space update when $\phi \ge 2$, we need some ``precalculated'' components (denoted by the double prime symbol):
\begin{align*}
	a_n &= \lceil \phi s_n \rceil, \\
	K''_{n,b} &= \left\{
	\begin{array}{ll}
		K''_{n-1,b} +(a_{n-1} -s'_n)^b X_n, &s'_n \ge a_{n-1} -s_{n-1}\ \text{and}\ s'_n > 0; \\
		0, &s'_n < a_{n-1} -s_{n-1}\ \text{or}\ s'_n = 0, \\
	\end{array}
	\right. \\
	k''_{n,b} &= \left\{
	\begin{array}{ll}
		k''_{n-1,b} +(a_{n-1} -s'_n)^b, &s'_n \ge a_{n-1} -s_{n-1}\ \text{and}\ s'_n > 0; \\
		0, &s'_n < a_{n-1} -s_{n-1}\ \text{or}\ s'_n = 0.
	\end{array}
	\right.
\end{align*}

\subsection{\texorpdfstring{$\LASER(1,1)$: Known Zero-mean}{
		LASER(1,1): Known Zero-mean}} \label{sec:algo-las-zm}

Since the constructed estimator $\LASER(q,\phi)$ is based on a new approach, we shall show the derivation of the algorithm of $\LASER(1,1)$ when $\mu=0$ is known for illustrative purpose first.
The other algorithms can be derived using a similar procedure and will be discussed later.
To begin with, note that the estimator can be written as
\begin{align*}
	\bar{\sigma}^2_{n, \LASER(1,1)}
	&= \frac{1}{n} \sum_{i=1}^n X_i^2 + \frac{2}{n} \sum_{i=1}^n \sum_{k=1}^{s_i} \left( 1 -\frac{k}{t_n} \right) X_i X_{i-k} \\
	&= \frac{1}{n} \left( Q_n +2 R_{n,0} -\frac{2}{t_n} R_{n,1} \right),
\end{align*}
where
\[
Q_n = \sum_{i=1}^n X_i^2, \quad
R_{n,0} = \sum_{i=1}^n \sum_{k=1}^{s_i} X_i X_{i-k} \quad \text{and} \quad
R_{n,1} = \sum_{i=1}^n \sum_{k=1}^{s_i} k X_i X_{i-k}.
\]
The recursive formula for $Q_n$ is trivial so we shall focus on $R_{n,0}$ and $R_{n,1}$.
To lighten notation, write the subsampling parameter at the last iteration as $s = s_{n-1}$.
For $R_{n,0}$, we have
\begin{align*}
	R_{n,0}
	&= \sum_{i=1}^{n-1} \sum_{k=1}^{s_i} X_i X_{i-k} + X_n \sum_{k=1}^{s_n} X_{n-k} \\
	&= R_{n-1,0} + X_n K_{n,0},
\end{align*}
where $K_{n,0} = \sum_{k=1}^{s_n} X_{n-k}$.
Note that $K_{n,0}$ can be updated easily since
\[
K_{n,0} = \left\{
\begin{array}{ll}
	0, &s_n = s = 0; \\
	K_{n-1,0} +X_{n-1} -X_{n-s-1}, &s_n = s > 0; \\
	K_{n-1,0} +X_{n-1}, &s_n = s+1. \\
\end{array}
\right.
\]
Similarly, for $R_{n,1}$, observe that
\begin{align*}
	R_{n,1}
	&= \sum_{i=1}^{n-1} \sum_{k=1}^{s_i} k X_i X_{i-k} + X_n \sum_{k=1}^{s_n} k X_{n-k} \\
	&= R_{n-1,1} + X_n K_{n,1},
\end{align*}
where $K_{n,1} = \sum_{k=1}^{s_n} k X_{n-k}$.
We can update $K_{n,1}$ recursively by
\[
K_{n,1} = \left\{
\begin{array}{ll}
	0, &s_n = s = 0; \\
	K_{n-1,1} +K_{n,0} -s X_{n-s-1}, &s_n = s > 0; \\
	K_{n-1,1} +K_{n,0}, &s_n = s+1. \\
\end{array}
\right.
\]
We summarize the recursive formulas for different components in the order of their updates:

\begin{proposition} \label{prop:las-zm}
	Suppose the subsampling parameter $\{s_i \in \N\}_{i \in \Z^+}$ is a monotonically increasing sequence such that $\sup_{i \in \Z^+} |s_{i+1}-s_i| \le 1$.
	Write its value at the last iteration as $s = s_{n-1}$.
	Then, the following components can be updated in $O(1)$ time:
	\begin{align*}
		K_{n,0} &= \sum_{k=1}^{s_n} X_{n-k}
		= \left\{
		\begin{array}{ll}
			0, &s_n = s = 0; \\
			K_{n-1,0} +X_{n-1} -X_{n-s-1}, &s_n = s > 0; \\
			K_{n-1,0} +X_{n-1}, &s_n = s+1, \\
		\end{array}
		\right. \\
		K_{n,1} &= \sum_{k=1}^{s_n} k X_{n-k}
		= \left\{
		\begin{array}{ll}
			0, &s_n = s = 0; \\
			K_{n-1,1} +K_{n,0} -s X_{n-s-1}, &s_n = s > 0; \\
			K_{n-1,1} +K_{n,0}, &s_n = s+1, \\
		\end{array}
		\right. \\
		R_{n,b} &= \sum_{i=1}^n \sum_{k=1}^{s_i} k^b X_i X_{i-k}
		= R_{n-1,b} + X_n K_{n,b} \quad \text{for} \quad b=0,1, \\
		Q_n &= \sum_{i=1}^n X_i^2 = Q_{n-1} +X_n^2.
	\end{align*}
\end{proposition}

When the data arrive sequentially, we need to maintain a vector $\vec{x}$ in order to store the last $s_n$ observations to update the moving sum $K_{n,0}$.
This can be implemented efficiently using a queue data structure as it is first-in-first-out; see Algorithm \ref{algo:las-zm}.

\begin{algorithm}[!t]
	\caption{$\LASER(1,1)$, known zero-mean} \label{algo:las-zm}
	\SetAlgoVlined
	\DontPrintSemicolon
	\SetNlSty{texttt}{[}{]}
	\small
	\textbf{initialization}: \;
	Set $n = 1, s = 0, Q_n = X_1^2, K_{n,0} = K_{n,1} = R_{n,0} = R_{n,1} = 0$ \;
	Push $X_1$ into $\vec{x}$ \;
	\Begin{
		Receive $X_{n+1}$ \;
		Set $n = n + 1$ \;
		Compute $s_n, t_n$ \;
		Retrieve $X_{n-1}, X_{n-s-1}$ from $\vec{x}$ \tcc*[f]{last and first elements in $\vec{x}$} \;
		\eIf{$s_n == s$}{
			Update $K_{n,0}, K_{n,1}$ with Prop. \ref{prop:las-zm} \;
			Pop $X_{n-s-1}$ from $\vec{x}$ \;
		}{
			Update $K_{n,0}, K_{n,1}$ with Prop. \ref{prop:las-zm} \;
		}
		Push $X_n$ into $\vec{x}$ \;
		Update $R_{n,0}, R_{n,1}, Q_n$ with Prop. \ref{prop:las-zm} \;
		Set $s = s_n$ \;
		Output $\bar{\sigma}^2_{n, \LASER(1,1)} = (Q_n +2 R_{n,0} -2R_{n,1}/t_n) /n$ \;
	}
\end{algorithm}

\subsection{\texorpdfstring{$\LASER(1,1)$: Unknown General Mean}{
		LASER(1,1): Unknown General Mean}} \label{sec:algo-las}

When $\mu \in \R$ is unknown, we can subtract the sample mean $\bar{X}_n = n^{-1} \sum_{i=1}^n X_i$ from each observation.
Therefore, the estimator can be written as
\begin{align*}
	\hat{\sigma}^2_{n, \LASER(1,1)}
	={}& \frac{1}{n} \left\{ \sum_{i=1}^n (X_i -\bar{X}_n)^2 + 2 \sum_{i=1}^n \sum_{k=1}^{s_i} \left( 1 -\frac{k}{t_n} \right) (X_i -\bar{X}_n) (X_{i-k} -\bar{X}_n) \right\} \\
	={}& \frac{1}{n} \left\{ \sum_{i=1}^n X_i^2 -n\bar{X}_n^2 + 2 \sum_{i=1}^n \sum_{k=1}^{s_i} \left( 1 -\frac{k}{t_n} \right) (X_i X_{i-k} -X_i\bar{X}_n -X_{i-k}\bar{X}_n +\bar{X}_n^2) \right\} \\
	={}& \frac{1}{n} \left[ Q_n -n\bar{X_n^2} +2R_{n,0} -\frac{2}{t_n}R_{n,1}
	-2\bar{X}_n \left\{ \sum_{i=1}^n \sum_{k=1}^{s_i} \left( 1 -\frac{k}{t_n} \right) (X_i +X_{i-k} -\bar{X}_n) \right\} \right] \\
	={}& \frac{1}{n} \bigg\{ Q_n +2 R_{n,0} -\frac{2}{t_n} R_{n,1} +\left( 2r_{n,0} -\frac{2}{t_n} r_{n,1} -n \right) \bar{X}_n^2 \\
	& -2\bar{X}_n \left(U_{n,0} -\frac{1}{t_n} U_{n,1} +V_{n,0} -\frac{1}{t_n} V_{n,1} \right) \bigg\},
\end{align*}
where the additional components are
\[
r_{n,b} = \sum_{i=1}^n \sum_{k=1}^{s_i} k^b, \quad
U_{n,b} = \sum_{i=1}^n \sum_{k=1}^{s_i} k^b X_i \quad \text{and} \quad
V_{n,b} = \sum_{i=1}^n \sum_{k=1}^{s_i} k^b X_{i-k} \quad \text{for} \quad b=0,1.
\]
To lighten notation, denote $R_n^{(a,b)} = R_{n,a} -R_{n,b}/t_n^b$
and similarly for $r_n^{(a,b)}, U_n^{(a,b)}, V_n^{(a,b)}$, e.g.,
$U_n^{(a,b)} = U_{n,a} -U_{n,b}/t_n^b$.
Then, the estimator can be expressed as
\[
\hat{\sigma}^2_{n, \LASER(1,1)}
= \frac{1}{n} \left\{ Q_n +2R_n^{(0,1)} +(2r_n^{(0,1)}-n) \bar{X}_n^2 -2\bar{X}_n (U_n^{(0,1)} +V_n^{(0,1)}) \right\}.
\]
Using the same procedure as in \cref{sec:algo-las-zm}, we can derive the recursive formulas for the additional components:

\begin{proposition} \label{prop:las}
	Suppose the subsampling parameter $\{s_i \in \N\}_{i \in \Z^+}$ is a monotonically increasing sequence such that $\sup_{i \in \Z^+} |s_{i+1}-s_i| \le 1$.
	Write its value at the last iteration as $s = s_{n-1}$.
	For each $b=0,1$, the following components can be updated in $O(1)$ time:
	\begin{align*}
		k_{n,b} &= \sum_{k=1}^{s_n} k^b
		= \left\{
		\begin{array}{ll}
			0, &s_n = s = 0; \\
			k_{n-1,b}, &s_n = s > 0; \\
			k_{n-1,b} +s_n^b, &s_n = s+1, \\
		\end{array}
		\right. \\
		r_{n,b} &= \sum_{i=1}^n \sum_{k=1}^{s_i} k^b
		= r_{n-1,b} +k_{n,b}, \\
		U_{n,b} &= \sum_{i=1}^n \sum_{k=1}^{s_i} k^b X_i
		= U_{n-1,b} +k_{n,b} X_n, \\
		V_{n,b} &= \sum_{i=1}^n \sum_{k=1}^{s_i} k^b X_{i-k}
		= V_{n-1,b} +K_{n,b}, \\
		\bar{X}_n &= \frac{1}{n} \sum_{i=1}^n X_i
		= \frac{(n-1)\bar{X}_{n-1}+X_n}{n}.
	\end{align*}
\end{proposition}

We notice that it may not be necessary to maintain the components separately, e.g., $U_{n,0}$ and $U_{n,1}$.
Instead, we can maintain the components like
$\mathcal{U}_n = \sum_{i=1}^n \sum_{k=1}^{s_i} (1-k/t_n) X_i$ directly.
However, it may violate Principle E.
Consequently, the recursive formulas for these components will depend on the tapering parameter $t_n$ and cannot be generalized to the settings in  \cref{sec:algo-lase} or \ref{sec:algo-mb} without restricting changes in $t_n$.
Therefore, it is better to maintain the components separately by factoring out $1/t_n$.
We state the recursive algorithm for $\hat{\sigma}^2_{n, \LASER(1,1)}$ in Algorithm \ref{algo:las}.

\begin{algorithm}[!t]
	\caption{$\LASER(1,1)$, unknown general mean} \label{algo:las}
	\SetAlgoVlined
	\DontPrintSemicolon
	\SetNlSty{texttt}{[}{]}
	\small
	\textbf{initialization}: \;
	Set $n = 1, s = 0, Q_n = X_1^2, \bar{X}_n = X_1$ \;
	Set $K_{n,b} = R_{n,b} = k_{n,b} = r_{n,b} = U_{n,b} = V_{n,b} = 0$ for $b=0,1$ \;
	Push $X_1$ into $\vec{x}$ \;
	\Begin{
		Receive $X_{n+1}$ \;
		Set $n = n + 1$ \;
		Compute $s_n, t_n$ \;
		Retrieve $X_{n-1}, X_{n-s-1}$ from $\vec{x}$ \tcc*[f]{last and first elements in $\vec{x}$} \;
		\eIf{$s_n == s$}{
			Update $K_{n,b}$ with Prop. \ref{prop:las-zm} for $b=0,1$ \;
			Pop $X_{n-s-1}$ from $\vec{x}$ \;
		}{
			Update $K_{n,b}$ with Prop. \ref{prop:las-zm} and $k_{n,b}$ with Prop. \ref{prop:las} for $b=0,1$ \;
		}
		Push $X_n$ into $\vec{x}$ \;
		Update $R_{n,b}, Q_n$ with Prop. \ref{prop:las-zm} and $r_{n,b}, U_{n,b}, V_{n,b}, \bar{X}_n$ with Prop. \ref{prop:las} for $b=0,1$ \;
		Set $s=s_n$ \;
		Output $\hat{\sigma}^2_{n, \LASER(1,1)} = \{ Q_n +2R_n^{(0,1)} +(2r_n^{(0,1)}-n) \bar{X}_n^2 -2\bar{X}_n (U_n^{(0,1)} +V_n^{(0,1)}) \} /n$ \;
	}
\end{algorithm}

\subsection{\texorpdfstring{$\LASER(q,1)$: Characteristic Exponent}{
		LASER(q,1): Characteristic Exponent}} \label{sec:algo-lase}

In this subsection, we illustrate how to derive recursive formulas for $\hat{\sigma}^2_{n,\LASER(q,1)}$.
We are going to focus on the $q$-th order term as similar formulas apply to lower order terms except that a separate set of components is needed for each term.

To begin with, we follow the same procedure as in \cref{sec:algo-las-zm,sec:algo-las} to write the estimator as
\begin{align*}
	\hat{\sigma}^2_{n, \LASER(q,1)}
	&= \frac{1}{n} \sum_{i=1}^n (X_i -\bar{X}_n)^2
	+\frac{2}{n} \sum_{i=1}^n \sum_{k=1}^{s_i} \left( 1 -\frac{k^q}{t_n^q} \right) (X_i -\bar{X}_n) (X_{i-k} -\bar{X}_n) \\
	&= \frac{1}{n} \left\{ Q_n +2R_n^{(0,q)} +(2r_n^{(0,q)}-n) \bar{X}_n^2 -2\bar{X}_n (U_n^{(0,q)} +V_n^{(0,q)}) \right\},
\end{align*}
where the shorthand, e.g., $R_n^{(0,q)} = R_{n,0} -R_{n,q}/t_n^q$, is introduced in \cref{sec:algo-las}; see \cref{sec:algo-overview} for an overview of notation.
Now, we need to work on a recursive formula for
$K_{n,q} = \sum_{k=1}^{s_n} k^q X_{n-k}$ in order to update $R_{n,q}$.
To this end, we can use the backward finite difference to be discussed in \ref{sec:proof-on}.
Recall the backward finite difference operator $\nabla^{(b)} \cdot$ is defined by \[
\nabla^{(1)} f(k) = f(k) -f(k-1) \quad \text{and} \quad
\nabla^{(b)} f(k) = \nabla^{(b-1)} f(k) -\nabla^{(b-1)} f(k-1),
\]
where $b \in \Z^+$ and $f$ is a function that takes an integer input $k$.
When $f(k) = k^q$, we have the recursive relation
\[
\nabla^{(1)} k^q = k^q -(k-1)^q \quad \text{and} \quad
\nabla^{(b)} k^q = \nabla^{(b-1)} k^q -\nabla^{(b-1)} (k-1)^q.
\]
To lighten notation, denote $d_{k,q}^{(b)} = \nabla^{(b)} k^q$.
From \cref{sec:proof-on}, we know that $d_{k,q}^{(q)}$ is a non-zero constant and well-defined when $k > q$.
Therefore, we can expand the recursive relation and put $k=q+1$ to obtain 
\begin{align*}
	d_{q+1,q}^{(q)}
	&= (q+1)^q -\binom{q}{1} q^q +\binom{q}{2} (q-1)^q - \cdots +(-1)^q 1^q \\
	&= \sum_{r=0}^q (-1)^r \binom{q}{r} (q-r+1)^q.
\end{align*}
Similarly, the initial values $d_{b+1,q}^{(b)}$ for $b = 1, \ldots, q-1$ are given by
\[
d_{b+1,q}^{(b)} = \sum_{r=0}^b (-1)^r \binom{b}{r} (b-r+1)^q.
\]
It remains to find $d_{k,q}^{(b)}$ when $k>b+1$ for $b=1, \ldots, q-1$.
By rearranging $\nabla^{(b+1)} k^q = \nabla^{(b)} k^q -\nabla^{(b)} (k-1)^q$, we have
\begin{align*}
	d_{k,q}^{(b)} &= \nabla^{(b)} k^q
	= \nabla^{(b)} (k-1)^q +\nabla^{(b+1)} k^q \\
	&= d_{k-1,q}^{(b)} +d_{k,q}^{(b+1)}.
\end{align*}
When $k$ increases, we can update $d_{k,q}^{(b)}$ recursively from $b=q-1$ to $b=1$ since $d_{k-1,q}^{(b)}$ comes from the last iteration and
$d_{k,q}^{(b+1)}$ is a constant or has been updated.
We summarize the recursive formulas for $d_{k,q}^{(b)}$ below:
\[
d_{k,q}^{(b)} = \left\{
\begin{array}{ll}
	0, &k<b+1 \ \text{and} \ b=1,2,\ldots,q; \\
	\sum_{r=0}^b (-1)^r \binom{b}{r} (b-r+1)^q, &k=b+1 \ \text{and} \ b=1,2,\ldots,q; \\
	d_{q+1,q}^{(q)}, &k>b+1 \ \text{and} \ b=q; \\
	d_{k-1,q}^{(b)} +d_{k,q}^{(b+1)}, &k>b+1 \ \text{and} \ b=1,2,\ldots,q-1. \\
\end{array}
\right.
\]
It is also possible to compute $d_{k,q}^{(b)}$ directly using $d_{k,q}^{(b)} = \sum_{r=0}^b (-1)^r \binom{b}{r} (k-r)^q$.
Nonetheless, this approach requires more arithmetic operations when $k$ changes.
Therefore, it is more desirable to use the recursive formulas for $d_{k,q}^{(b)}$.
Analogous to Proposition \ref{prop:las-zm}, the recursive formulas for the zero-mean components are summarized below.

\begin{proposition} \label{prop:lase-zm}
	Suppose the subsampling parameter $\{s_i \in \N\}_{i \in \Z^+}$ is a monotonically increasing sequence such that $\sup_{i \in \Z^+} |s_{i+1}-s_i| \le 1$.
	Write its value at the last iteration as $s = s_{n-1}$.
	Then, the following components can be updated in $O(1)$ time:
	\begin{align*}
		c_q^{(b)} &= \sum_{r=0}^b (-1)^r \binom{b}{r} (b-r+1)^q
		\quad \text{for} \quad b = 1,2,\ldots,q, \\
		d_{s_n,q}^{(b)} &= \left\{
		\begin{array}{ll}
			0, &s_n < b+1 \ \text{and} \ b=1,2,\ldots,q; \\
			c_q^{(b)}, &s_n = b+1 \ \text{and} \ b=1,2,\ldots,q; \\
			d_{s,q}^{(b)}, &s_n = s > b+1 \ \text{and} \ b=1,2,\ldots,q; \\
			d_{s,q}^{(b)}, &s_n = s+1 > b+1 \ \text{and} \ b=q; \\
			d_{s,q}^{(b)} +d_{s+1,q}^{(b+1)}, &s_n = s+1 > b+1 \ \text{and} \ b=1,2,\ldots,q-1, \\
		\end{array}
		\right. \\
		D_{n,q}^{(q)} &= \sum_{k=1}^{s_n} d_{k,q}^{(q)} X_{n-k} \\
		&= \left\{
		\begin{array}{ll}
			0, &s_n < q+1; \\
			D_{n-1,q}^{(q)} +c_q^{(q)} X_{n-q-1} -d_{s,q}^{(q)} X_{n-s-1}, &s_n = s \ge q+1; \\
			D_{n-1,q}^{(q)} +c_q^{(q)} X_{n-q-1}, &s_n = s+1 \ge q+1, \\
		\end{array}
		\right. \\
		D_{n,q}^{(b)} &= \sum_{k=1}^{s_n} d_{k,q}^{(b)} X_{n-k}
		\quad \text{for} \quad b = q-1,q-2,\ldots,1, \\
		&= \left\{
		\begin{array}{ll}
			0, &s_n < b+1; \\
			D_{n-1,q}^{(b)} +c_q^{(b)} X_{n-b-1} +D_{n,q}^{(b+1)} -d_{s,q}^{(b)} X_{n-s-1}, &s_n = s \ge b+1; \\
			D_{n-1,q}^{(b)} +c_q^{(b)} X_{n-b-1} +D_{n,q}^{(b+1)}, &s_n = s+1 \ge b+1, \\
		\end{array}
		\right. \\
		K_{n,q} &= \sum_{k=1}^{s_n} k^q X_{n-k}
		= \left\{
		\begin{array}{ll}
			0, &s_n = s = 0; \\
			K_{n-1,q} +X_{n-1} +D_{n,q}^{(1)} -s^q X_{n-s-1}, &s_n = s > 0;\\
			K_{n-1,q} +X_{n-1} +D_{n,q}^{(1)}, &s_n = s+1, \\
		\end{array}
		\right. \\
		K_{n,0} &= \sum_{k=1}^{s_n} X_{n-k}
		= \left\{
		\begin{array}{ll}
			0, &s_n = s = 0; \\
			K_{n-1,0} +X_{n-1} -X_{n-s-1}, &s_n = s > 0; \\
			K_{n-1,0} +X_{n-1}, &s_n = s+1, \\
		\end{array}
		\right. \\
		R_{n,b} &= \sum_{i=1}^n \sum_{k=1}^{s_i} k^b X_i X_{i-k}
		= R_{n-1,b} + X_n K_{n,b} \quad \text{for} \quad b=0,q, \\
		Q_n &= \sum_{i=1}^n X_i^2 = Q_{n-1} +X_n^2.
	\end{align*}
\end{proposition}

Analogous to Proposition \ref{prop:las}, the recursive formulas for the additional components are as follows.

\begin{proposition} \label{prop:lase}
	Suppose the subsampling parameter $\{s_i \in \N\}_{i \in \Z^+}$ is a monotonically increasing sequence such that $\sup_{i \in \Z^+} |s_{i+1}-s_i| \le 1$.
	Write its value at the last iteration as $s = s_{n-1}$.
	For each $b=0,q$, the following components can be updated in $O(1)$ time:
	\begin{align*}
		k_{n,b} &= \sum_{k=1}^{s_n} k^b
		= \left\{
		\begin{array}{ll}
			0, &s_n = s = 0; \\
			k_{n-1,b}, &s_n = s > 0; \\
			k_{n-1,b} +s_n^b, &s_n = s+1, \\
		\end{array}
		\right. \\
		r_{n,b} &= \sum_{i=1}^n \sum_{k=1}^{s_i} k^b
		= r_{n-1,b} +k_{n,b}, \\
		U_{n,b} &= \sum_{i=1}^n \sum_{k=1}^{s_i} k^b X_i
		= U_{n-1,b} +k_{n,b} X_n, \\
		V_{n,b} &= \sum_{i=1}^n \sum_{k=1}^{s_i} k^b X_{i-k}
		= V_{n-1,b} +K_{n,b}, \\
		\bar{X}_n &= \frac{1}{n} \sum_{i=1}^n X_i
		= \frac{(n-1)\bar{X}_{n-1}+X_n}{n}.
	\end{align*}
\end{proposition}

We state the recursive algorithm for $\hat{\sigma}^2_{n, \LASER(q,1)}$ in Algorithm \ref{algo:lase}.

\begin{algorithm}[!t]
	\caption{$\LASER(q,1)$} \label{algo:lase}
	\SetAlgoVlined
	\DontPrintSemicolon
	\SetNlSty{texttt}{[}{]}
	\small
	\textbf{initialization}: \;
	Set $n = 1, s = 0, Q_n = X_1^2, \bar{X}_n = X_1$ \;
	Set $K_{n,b} = R_{n,b} = k_{n,b} = r_{n,b} = U_{n,b} = V_{n,b} = 0$ for $b=0,q$ \;
	Compute $c_q^{(b)}$, and set $D_{n,q}^{(b)} = d_{s_n,q}^{(b)}=0$ for $b=1,2,\ldots,q$ \;
	Push $X_1$ into $\vec{x}$ \;
	\Begin{
		Receive $X_{n+1}$ \;
		Set $n = n + 1$ \;
		Compute $s_n, t_n$ \;
		Retrieve $X_{n-1}, \ldots, X_{n-1-\max\{\min(s_n-1,q),0\}}$ and $X_{n-s-1}$ from $\vec{x}$ \;
		\eIf{$s_n == s$}{
			Update $D_{n,q}^{(b)}$ with Prop. \ref{prop:lase-zm} for $b=q,q-1,\ldots,1$ \;
			Update $K_{n,b}$ with Prop. \ref{prop:lase-zm} for $b=0,q$ \;
			Pop $X_{n-s}$ from $\vec{x}$ \;
		}{
			Update $d_{s_n,q}^{(b)}, D_{n,q}^{(b)}$ with Prop. \ref{prop:lase-zm} for $b=q,q-1,\ldots,1$ \;
			Update $K_{n,b}$ with Prop. \ref{prop:lase-zm} and $k_{n,b}$ with Prop. \ref{prop:lase} for $b=0,q$ \;
		}
		Push $X_n$ into $\vec{x}$ \;
		Update $R_{n,b}, Q_n$ with Prop. \ref{prop:lase-zm} and $r_{n,b}, U_{n,b}, V_{n,b}, \bar{X}_n$ with Prop. \ref{prop:lase} for $b=0,q$ \;
		Set $s=s_n$ \;
		Output $\hat{\sigma}^2_{n, \LASER(q,1)} = \{ Q_n +2R_n^{(0,q)} +(2r_n^{(0,q)}-n) \bar{X}_n^2 -2\bar{X}_n (U_n^{(0,q)} +V_n^{(0,q)}) \} /n$ \;
	}
\end{algorithm}

\subsection{\texorpdfstring{$\LASER(q,\phi)$: Memory Parameter}{
		LASER(q,phi): Memory Parameter}} \label{sec:algo-laser}

In this subsection, we illustrate how to derive recursive formulas for
$\hat{\sigma}^2_{n, \LASER(q,\phi)}$ with $q=1$ and $\phi \ge 2$.
We take these values to focus on how to perform $O(1)$-space update.
The general algorithm for $q \in \Z^+$ and $\phi \ge 2$, which is implemented in our R package, can be obtained by combining the ideas in \cref{sec:algo-lase} and here.

To begin with, we follow the same procedure as in \cref{sec:algo-las-zm,sec:algo-las} to write the estimator as
\begin{align*}
	\hat{\sigma}^2_{n, \LASER(1,\phi)} &= \frac{1}{n} \sum_{i=1}^n (X_i -\bar{X}_n)^2
	+\frac{2}{n} \sum_{i=1}^n \sum_{k=1}^{s'_i} \left( 1 -\frac{k}{t_n} \right) (X_i -\bar{X}_n) (X_{i-k} -\bar{X}_n) \\
	&= \frac{1}{n} \left\{ Q_n +2R_n^{\prime(0,1)} +(2r_n^{\prime(0,1)}-n) \bar{X}_n^2 -2\bar{X}_n (U_n^{\prime(0,1)} +V_n^{\prime(0,1)}) \right\},
\end{align*}
where the shorthand, e.g., $R_n^{\prime(0,1)} = R'_{n,0} -R'_{n,1}/t_n$, is similarly introduced \cref{sec:algo-las}; see \cref{sec:algo-overview} for an overview of notation.
Since $q=1$, the components are analogous to those in Propositions \ref{prop:las-zm} and \ref{prop:las} with the only difference in the use of the ramped subsampling parameter $s'_i$.
When $s'_i$ is increasing, there is no difference in the recursive formulas.
However, when $s'_i$ is too large that it resets to the value of $s_i$, we need another way to update the components recursively.
For $\phi \ge 2$, this can be done by ``precalculating'' some components.
We summarize the recursive formulas for the zero-mean components below.

\begin{proposition} \label{prop:laser-zm}
	Suppose the intended subsampling parameter $\{s_i \in \N\}_{i \in \Z^+}$ is a monotonically increasing sequence such that $\sup_{i \in \Z^+} |s_{i+1}-s_i| \le 1$.
	Furthermore, suppose the ramped subsampling parameter defined in \cref*{eq:rampSub}, i.e., $\{s'_i\}_{i \in \Z^+}$, is ramped with $\phi \ge 2$.
	Write its value at the last iteration as $s' = s'_{n-1}$.
	Denote the ramping upper bound by $a_n = \lceil \phi s_n \rceil$.
	To keep the branching simple, assume $s_n$ and $a_n$ only increase when $s'_n$ resets.
	Then, the following components can be updated in $O(1)$ time and $O(1)$ space:	
	\begin{align*}
		K'_{n,0} &= \sum_{k=1}^{s'_n} X_{n-k}
		= \left\{
		\begin{array}{ll}
			0, &0 = s'_n = s'; \\
			K'_{n-1,0} +X'_{n-1}, &0 < s'_n = s'+1; \\
			K''_{n-1,0}, &0 < s'_n = s_{n-1} < s'; \\
			K''_{n-1,0} +X''_{n-1}, &0 < s'_n = s_{n-1}+1 < s', \\
		\end{array}
		\right. \\
		K'_{n,1} &= \sum_{k=1}^{s'_n} k X_{n-k}
		= \left\{
		\begin{array}{ll}
			0, &0 = s'_n = s'; \\
			K'_{n-1,1} +K'_{n,0}, &0 < s'_n = s'+1; \\
			K''_{n-1,1}, &0 < s'_n = s_{n-1} < s'; \\
			K''_{n-1,1} +s_n X''_{n-1}, &0 < s'_n = s_{n-1}+1 < s', \\
		\end{array}
		\right. \\
		K''_{n,0} &= \left\{
		\begin{array}{ll}
			K''_{n-1,0} +X_n, &s'_n \ge a_{n-1} -s_{n-1}\ \text{and}\ s'_n > 0; \\
			0, &s'_n < a_{n-1} -s_{n-1}\ \text{or}\ s'_n=0, \\
		\end{array}
		\right. \\
		K''_{n,1} &= \left\{
		\begin{array}{ll}
			K''_{n-1,1} +(a_{n-1} -s'_n) X_n, &s'_n \ge a_{n-1} -s_{n-1}\ \text{and}\ s'_n > 0; \\
			0, &s'_n < a_{n-1} -s_{n-1}\ \text{or}\ s'_n=0, \\
		\end{array}
		\right. \\
		R'_{n,b} &= \sum_{i=1}^n \sum_{k=1}^{s'_i} k^b X_i X_{i-k}
		= R'_{n-1,b} + X_n K'_{n,b} \quad \text{for} \quad b=0,1, \\
		Q_n &= \sum_{i=1}^n X_i^2 = Q_{n-1} +X_n^2, \\
		X''_n &= \left\{
		\begin{array}{ll}
			X'_{n-1}, &s'_n = a_n -s_n; \\
			X''_{n-1}, &s'_n \ne a_n -s_n, \\
		\end{array}
		\right. \\
		X'_n &= X_n.
	\end{align*}
\end{proposition}

Note that $K''_{n-1,0}, K''_{n-1,1}, X'_{n-1}, X''_{n-1}$ are the ``precalculated'' versions of $K'_{n-1,0}, K'_{n-1,1}, X_{n-1}$, $X_{n-s_n}$, respectively, when $s'_n$ resets to $s_n$.
Using these components, we do not need to store the past observations explicitly, e.g., by maintaining the vector $\vec{x}$ in Algorithm \ref{algo:las-zm}.
As a result, only a constant amount of memory is involved.
The restriction that $s_n$ and $a_n$ only increase when $s'_n$ resets also does not affect the asymptotic efficiency because it just slightly delays the increment of $s_n$.
For the additional components, their recursive formulas are as follows:

\begin{proposition} \label{prop:laser}
	Suppose the intended subsampling parameter $\{s_i \in \N\}_{i \in \Z^+}$ is a monotonically increasing sequence such that $\sup_{i \in \Z^+} |s_{i+1}-s_i| \le 1$.
	Furthermore, suppose the ramped subsampling parameter defined in \cref*{eq:rampSub}, i.e., $\{s'_i\}_{i \in \Z^+}$, is ramped with $\phi \ge 2$.
	Write its value at the last iteration as $s' = s'_{n-1}$.
	Denote the ramping upper bound by $a_n = \lceil \phi s_n \rceil$. To keep the branching simple, assume $s_n$ and $a_n$ only increase when $s'_n$ resets.
	For each $b=0,1$, the following components can be updated in $O(1)$ time and $O(1)$ space:
	\begin{align*}
		k'_{n,b} &= \sum_{k=1}^{s'_n} k^b
		= \left\{
		\begin{array}{ll}
			0, &0 = s'_n = s'; \\
			k'_{n-1,b} +(s'_n)^b, &0 < s'_n = s'+1; \\
			k''_{n-1,b}, &0 < s'_n = s_{n-1} < s'; \\
			k''_{n-1,b} +s_n^b, &0 < s'_n = s_{n-1}+1 < s', \\
		\end{array}
		\right. \\
		k''_{n,b} &= \left\{
		\begin{array}{ll}
			k''_{n-1,b} +(a_{n-1} -s'_n)^b, &s'_n \ge a_{n-1} -s_{n-1}\ \text{and}\ s'_n > 0; \\
			0, &s'_n < a_{n-1} -s_{n-1}\ \text{or}\ s'_n = 0, \\
		\end{array}
		\right. \\
		r'_{n,b} &= \sum_{i=1}^n \sum_{k=1}^{s'_i} k^b
		= r'_{n-1,b} +k'_{n,b}, \\
		U'_{n,b} &= \sum_{i=1}^n \sum_{k=1}^{s'_i} k^b X_i
		= U'_{n-1,b} +k'_{n,b} X_n, \\
		V'_{n,b} &= \sum_{i=1}^n \sum_{k=1}^{s'_i} k^b X_{i-k}
		= V'_{n-1,b} +K'_{n,b}, \\
		\bar{X}_n &= \frac{1}{n} \sum_{i=1}^n X_i
		= \frac{(n-1)\bar{X}_{n-1}+X_n}{n}.
	\end{align*}
\end{proposition}

We state the recursive algorithm for $\hat{\sigma}^2_{n, \LASER(1,\phi)}$ in Algorithm \ref{algo:laser}.

\begin{algorithm}[!t]
	\caption{$\LASER(1,\phi)$} \label{algo:laser}
	\SetAlgoVlined
	\DontPrintSemicolon
	\SetNlSty{texttt}{[}{]}
	\small
	\textbf{initialization}: \;
	Set $n = 1, s = 0, Q_n = X_1^2, \bar{X}_n = X_1$ \;
	Set $K'_{n,b} = K''_{n,b} = R'_{n,b} = k'_{n,b} = k''_{n,b} = r'_{n,b} = U'_{n,b} = V'_{n,b} = 0$ for $b=0,1$ \;
	Set $s' = a = 0, X'_n = X''_n = X_1$ \;
	\Begin{
		Receive $X_{n+1}$ \;
		Set $n = n + 1$ \;
		Compute $s_n, t_n, a_n$ \;
		\eIf{$s'+1 < a$}{
			Set $s'_n = s' + 1, s_n = s, a_n = a$ \tcc*[f]{restrict $s_n$ when $s'_n>s'$} \;
			Update $K'_{n,b}$ with Prop. \ref{prop:laser-zm} and $k'_{n,b}$ with Prop. \ref{prop:laser} for $b=0,1$ \;
			\If{$s'_n \ge a - s$}{
				Update $K''_{n,b}$ with Prop. \ref{prop:laser-zm} and $k''_{n,b}$ with Prop. \ref{prop:laser} for $b=0,1$ \;
			}
		}{
			Set $s'_n = s_n, a = a_n$ \;
			\eIf{$s'_n == s$}{
				Update $K'_{n,b}$ with Prop. \ref{prop:laser-zm} and $k'_{n,b}$ with Prop. \ref{prop:laser} for $b=0,1$ \;
			}{
				Update $K'_{n,b}$ with Prop. \ref{prop:laser-zm} and $k'_{n,b}$ with Prop. \ref{prop:laser} for $b=0,1$ \;
			}
			Set $K''_{n,b} = k''_{n,b} = 0$ for $b=0,1$ \;
		}
		Update $R'_{n,b}, Q_n$ with Prop. \ref{prop:laser-zm} and $r'_{n,b}, U'_{n,b}, V'_{n,b}, \bar{X}_n$ with Prop. \ref{prop:laser} for $b=0,1$ \;
		\If{$s'_n == a_n - s_n$}{
			Set $X''_n = X'_n$ \;
		}
		Set $s'=s'_n, s=s_n, X'_n = X_n$ \;
		Output $\scriptstyle \hat{\sigma}^2_{n,\LASER(1,\phi)} = \{ Q_n +2R_n^{\prime(0,1)} +(2r_n^{\prime(0,1)}-n) \bar{X}_n^2 -2\bar{X}_n (U_n^{\prime(0,1)} +V_n^{\prime(0,1)}) \} /n$ \;
	}
\end{algorithm}

\subsection{Mini-batch Update} \label{sec:algo-mb}

In this subsection, we present and discuss the mini-batch algorithm for $\hat{\sigma}^2_{n_j, \LASER(1,1)}$.
The general algorithm for $q \in \Z^+$ or $\phi \ge 2$, which is implemented in our R package, can be obtained by combining the ideas in \cref{sec:algo-lase,sec:algo-laser} and here.

Following the discussion in \cref*{sec:practical-enhancements}, we modify Algorithm \ref{algo:las} in two ways to obtain Algorithm \ref{algo:mb}. 
First, we may retrieve $X_{i-1}$ or $X_{i-s-1}$ from the input $X_{n+1}, X_{n+2}, \ldots, X_{n_{j+1}}$.
If we simply apply Algorithm \ref{algo:las} for $i=n+1,\ldots,n_{j+1}$, there will be a considerable number of redundant operations associated with the input and output of $\vec{x}$.
Second, we update $Q_i$ and $\bar{X}_i$ from $i=n_j$ to $i=n_{j+1}$ directly, which allows vectorized operations to be employed.
It is also possible to update $R_{i,b}, r_{i,b}, U_{i,b}, V_{i,b}$ from $i=n_j$ to $i=n_{j+1}$ directly by storing $\{K_{i,b}, k_{i,b}\}$.
Although we do not use this design in our implementation, it can be efficient if the mini-batch size $m=n_{j+1}-n_j$ that determines the storage is prespecified.

\begin{algorithm}[!t]
	\caption{Mini-batch $\LASER(1,1)$} \label{algo:mb}
	\SetAlgoVlined
	\DontPrintSemicolon
	\SetNlSty{texttt}{[}{]}
	\small
	\textbf{initialization}: \;
	Set $n = 1, s = j = 0, Q_n = X_1^2, \bar{X}_n = X_1$ \;
	Set $K_{n,b} = R_{n,b} = k_{n,b} = r_{n,b} = U_{n,b} = V_{n,b} = 0$ for $b=0,1$ \;
	Push $X_1$ into $\vec{x}$ \;
	\Begin{
		Receive $X_{n+1}, X_{n+2}, \ldots, X_{n_{j+1}}$ \;
		\For{$i = n+1$ \KwTo $n_{j+1}$}{		
			Compute $s_i$ \;
			Retrieve $X_{i-1}, X_{i-s-1}$ from $\vec{x}$ or the latest input \;
			\eIf{$s_i == s$}{
				Update $K_{i,b}$ with Prop. \ref{prop:las-zm} for $b=0,1$ \;
				Pop $X_{i-s-1}$ from $\vec{x}$ \;
			}{
				Update $K_{i,b}$ with Prop. \ref{prop:las-zm} and $k_{i,b}$ with Prop. \ref{prop:las} for $b=0,1$ \;
			}
			Update $R_{i,b}$ with Prop. \ref{prop:las-zm} and $r_{i,b}, U_{i,b}, V_{i,b}$ with Prop. \ref{prop:las} for $b=0,1$ \;
			Set $s=s_n$ \;
		}
		Set $n = n_{j+1}$ and $j = j +1$ \;
		Update $Q_n$ with Prop. \ref{prop:las-zm} and $\bar{X}_n$ with Prop. \ref{prop:las} \;
		Set $\vec{x} = (X_{n-s},\ldots,X_n)^\T$ \;
		Compute $t_n$ \;
		Output $\hat{\sigma}^2_{n, \LASER(1,1)} = \{ Q_n +2R_n^{(0,1)} +(2r_n^{(0,1)}-n) \bar{X}_n^2 -2\bar{X}_n (U_n^{(0,1)} +V_n^{(0,1)}) \} /n$ \;
	}
\end{algorithm}

\subsection{Nuisance Parameter Estimation} \label{sec:algo-vq}

In this subsection, we illustrate how to derive recursive formulas for
$\hat{v}_{q,n,\LASER(p,\phi)}$ with $p,q \in \Z^+$ and $\phi=1$.
The general algorithms under ramping or a mini-batch setting, which are implemented in our R package, can be obtained by combining the ideas in \cref{sec:algo-laser,sec:algo-mb} and here.

To begin with, we follow the same procedure as in \cref{sec:algo-las-zm,sec:algo-las} to write the estimator as
\begin{align*}
	\hat{v}_{q,n,\LASER(p,1)}
	&= \frac{2}{n} \sum_{i=1}^n \sum_{k=1}^{s_i} \left( 1 -\frac{k^p}{t_n^p} \right) k^q (X_i -\bar{X}_n) (X_{i-k} -\bar{X}_n) \\
	&= \frac{2}{n} \left\{ R_n^{(q,p+q,p)} +r_n^{(q,p+q,p)} \bar{X}_n^2 -\bar{X}_n (U_n^{(q,p+q,p)} +V_n^{(q,p+q,p)}) \right\},
\end{align*}
where the shorthand, e.g., $R_n^{(q,p+q,p)} = R_{n,q} -R_{n,p+q}/t_n^p$,
is the full version of the one introduced in \cref{sec:algo-las}; see \cref{sec:algo-overview} for an overview of notation.
The components are analogous to those in Propositions \ref{prop:lase-zm} and \ref{prop:lase} with the only difference in the new subscript $p+q$.
We summarize the recursive formulas for the zero-mean components below.

\begin{proposition} \label{prop:vq-zm}
	Suppose the subsampling parameter $\{s_i \in \N\}_{i \in \Z^+}$ is a monotonically increasing sequence such that $\sup_{i \in \Z^+} |s_{i+1}-s_i| \le 1$.
	Write its value at the last iteration as $s = s_{n-1}$.
	For each $\rho=q,p+q$, the following components can be updated in $O(1)$ time:	
	\begin{align*}
		c_\rho^{(b)} &= \sum_{r=0}^b (-1)^r \binom{b}{r} (b-r+1)^\rho \quad \text{for} \quad b = 1,2,\ldots,\rho, \\
		d_{s_n,\rho}^{(b)} &= \left\{
		\begin{array}{ll}
			0, &s_n < b+1 \ \text{and} \ b=1,2,\ldots,\rho; \\
			c_\rho^{(b)}, &s_n = b+1 \ \text{and} \ b=1,2,\ldots,\rho; \\
			d_{s,\rho}^{(b)}, &s_n = s > b+1 \ \text{and} \ b=1,2,\ldots,\rho; \\
			d_{s,\rho}^{(b)}, &s_n = s+1 > b+1 \ \text{and} \ b=\rho; \\
			d_{s,\rho}^{(b)} +d_{s+1,\rho}^{(b+1)}, &s_n = s+1 > b+1 \ \text{and} \ b=1,2,\ldots,\rho-1, \\
		\end{array}
		\right. \\
		D_{n,\rho}^{(\rho)} &= \sum_{k=1}^{s_n} d_{k,\rho}^{(\rho)} X_{n-k} \\
		&= \left\{
		\begin{array}{ll}
			0, &s_n < \rho+1; \\
			D_{n-1,\rho}^{(\rho)} +c_\rho^{(\rho)} X_{n-\rho-1} -d_{s,\rho}^{(\rho)} X_{n-s-1}, &s_n = s \ge \rho+1; \\
			D_{n-1,\rho}^{(\rho)} +c_\rho^{(\rho)} X_{n-\rho-1}, &s_n = s+1 \ge \rho+1, \\
		\end{array}
		\right. \\
		D_{n,\rho}^{(b)} &= \sum_{k=1}^{s_n} d_{k,\rho}^{(b)} X_{n-k} \quad \text{for} \quad b = \rho-1,\rho-2,\ldots,1, \\
		&= \left\{
		\begin{array}{ll}
			0, &s_n < b+1; \\
			D_{n-1,\rho}^{(b)} +c_\rho^{(b)} X_{n-b-1} +D_{n,\rho}^{(b+1)} -d_{s,\rho}^{(b)} X_{n-s-1}, &s_n = s \ge b+1; \\
			D_{n-1,\rho}^{(b)} +c_\rho^{(b)} X_{n-b-1} +D_{n,\rho}^{(b+1)}, &s_n = s+1 \ge b+1, \\
		\end{array}
		\right. \\
		K_{n,\rho} &= \sum_{k=1}^{s_n} k^\rho X_{n-k}
		= \left\{
		\begin{array}{ll}
			0, &s_n = s = 0; \\
			K_{n-1,\rho} +X_{n-1} +D_{n,\rho}^{(1)} -s^\rho X_{n-s-1}, &s_n = s > 0;\\
			K_{n-1,\rho} +X_{n-1} +D_{n,\rho}^{(1)}, &s_n = s+1, \\
		\end{array}
		\right. \\
		R_{n,\rho} &= \sum_{i=1}^n \sum_{k=1}^{s_i} k^\rho X_i X_{i-k}
		= R_{n-1,\rho} + X_n K_{n,\rho}.
	\end{align*}
\end{proposition}

For the additional components, their recursive formulas are as follows:

\begin{proposition} \label{prop:vq}
	Suppose the subsampling parameter $\{s_i \in \N\}_{i \in \Z^+}$ is a monotonically increasing sequence such that $\sup_{i \in \Z^+} |s_{i+1}-s_i| \le 1$.
	Write its value at the last iteration as $s = s_{n-1}$.
	For each $\rho=q,p+q$, the following components can be updated in $O(1)$ time:
	\begin{align*}
		k_{n,\rho} &= \sum_{k=1}^{s_n} k^\rho
		= \left\{
		\begin{array}{ll}
			0, &s_n = s = 0; \\
			k_{n-1,\rho}, &s_n = s > 0; \\
			k_{n-1,\rho} +s_n^\rho, &s_n = s+1, \\
		\end{array}
		\right. \\
		r_{n,\rho} &= \sum_{i=1}^n \sum_{k=1}^{s_i} k^\rho
		= r_{n-1,\rho} +k_{n,\rho}, \\
		U_{n,\rho} &= \sum_{i=1}^n \sum_{k=1}^{s_i} k^\rho X_i
		= U_{n-1,\rho} +k_{n,\rho} X_n, \\
		V_{n,\rho} &= \sum_{i=1}^n \sum_{k=1}^{s_i} k^\rho X_{i-k}
		= V_{n-1,\rho} +K_{n,\rho}, \\
		\bar{X}_n &= \frac{1}{n} \sum_{i=1}^n X_i
		= \frac{(n-1)\bar{X}_{n-1}+X_n}{n}.
	\end{align*}
\end{proposition}

The recursive algorithm for $\hat{v}_{q,n,\LASER(p,1)}$ is similar to Algorithm \ref{algo:lase} once we have Propositions \ref{prop:vq-zm} and \ref{prop:vq}.
Therefore, we omit its pseudocode here.

\subsection{Automatic Update} \label{sec:algo-auto}

In this subsection, we discuss how to select the smoothing parameters $s_n$ and $t_n$ automatically when we update $\hat{\sigma}^2_{n,\LASER(q,\phi)}$ with $q \in \Z^+$ and $\phi=1$.
The general algorithms under ramping or a mini-batch setting are similar and thus omitted.

To begin with, denote the set of components used to update $\hat{\sigma}_n^2$ and $\hat{v}_{q,n}$ by $\mathcal{C}(\hat{\sigma}^2_n)$ and $\mathcal{C}(\hat{v}_{q,n})$.
For instance,
\[
\mathcal{C}(\hat{\sigma}^2_{n,\LASER(1,1)})
= \left\{ n, s_n, t_n, Q_n, \bar{X}_n, \{K_{n,b}, R_{n,b}, k_{n,b}, r_{n,b}, U_{n,b}, V_{n,b}\}_{b=0,1}, \{X_{n-b}\}_{b=0,\ldots,s_n} \right\};
\]
see Algorithm \ref{algo:las}. 
We also define $\Psi_\diamond = \Psi_\star/\kappa^{2/(1+2q)}$ and
$\Theta_\diamond = \Theta_\star/\kappa^{2/(1+2q)}$, which denote the optimal coefficients for computing $s_n$ and $t_n$ excluding the unknown $\kappa$; see Corollary \ref{coro:mse}.
Then, the automatic optimal parameters selector is given in Algorithm \ref{algo:auto}.

\begin{algorithm}[!t]
	\caption{Automatic optimal parameters selector for $\LASER(q,1)$} \label{algo:auto}
	\SetAlgoVlined
	\DontPrintSemicolon
	\SetNlSty{texttt}{[}{]}
	\SetNoFillComment
	\small
	\textbf{input for the selector}: \;
	(i) $n, s_n, t_n$ -- the sample size, subsampling and tapering parameters from the last iteration \;
	(ii) $\Psi_\diamond, \psi_\star$ --  the optimal coefficient (excluding $\kappa$) and exponent for computing $s_n$ \;
	(iii) $\Theta_\diamond, \theta_\star$ --  the optimal coefficient (excluding $\kappa$) and exponent for computing $t_n$ \;
	(iv) $s_0, t_0$ -- the minimum values of subsampling and tapering parameters \;
	\textbf{input for the estimators}: \;
	(v) $X_{n+1}$ -- the new observation \;
	(vi) $\mathcal{C}(\hat{\sigma}^2_{n,\LASER(q,1)}), \mathcal{C}(\hat{v}_{q,n,\LASER(p,1)})$ -- the components for recursive estimation \;
	\Begin{
		Set $\hat{\kappa} = \left| \hat{v}_{q,n,\LASER(p,1)} \right|/\hat{\sigma}^2_{n,\LASER(q,1)}$ \;
		Set $s_{n+1} = \lfloor \Psi_\diamond \hat{\kappa}^{2/(1+2q)} (n+1)^{\psi_\star} \rfloor$ and $t_{n+1} = \lfloor \Theta_\diamond \hat{\kappa}^{2/(1+2q)} (n+1)^{\theta_\star} \rfloor$ \;
		\eIf{$s_n < \max(s_{n+1}, s_0)$}{
			Set $s_{n+1} = s_n +1$ \;
		}{
			Set $s_{n+1} = s_n$ \;
		}
		\eIf{$t_n < \max(t_{n+1}, t_0)$}{
			Set $t_{n+1} = t_n +1$ \;
		}{
			Set $t_{n+1} = t_n$ \;
		}
		Update $\hat{\sigma}^2_{n,\LASER(q,1)}$ with $X_{n+1}, \mathcal{C}(\hat{\sigma}^2_{n,\LASER(q,1)}), s_{n+1}, t_{n+1}$ \;
		Update $\hat{v}_{q,n,\LASER(p,1)}$ with $X_{n+1}, \mathcal{C}(\hat{v}_{q,n,\LASER(p,1)})$ \;
	}
\end{algorithm}

There are several remarks about Algorithm \ref{algo:auto}.
First, the nuisance parameter estimate $\hat{\kappa}$ seems to use fewer observations than the updated long-run variance estimate $\hat{\sigma}^2_{n+1}$.
This is natural because the selector uses the previous long-run variance estimate $\hat{\sigma}^2_n$ as the denominator for $\hat{\kappa}$.
If the numerator $\hat{v}_{q,n}$ was updated first, the nuisance parameter estimate could be unstable in a mini-batch setting where $\hat{\kappa} = \left| \hat{v}_{q,n_{j+1}} \right|/\hat{\sigma}^2_{n_j}$.
Alternatively, we can update $\hat{\kappa}$ as follows:

\begin{enumerate}[label=(\arabic*)]
	\item Set $\hat{\kappa}_n = \left| \hat{v}_{q,n} \right|/\hat{\sigma}^2_n$.
	\item Update $\hat{\sigma}^2_n$ with $X_{n+1}, \hat{\kappa}_n$.
	\item Update $\hat{v}_{q,n}$.
	\item Set $\hat{\kappa}_{n+1} = \left| \hat{v}_{q,n+1} \right|/\hat{\sigma}^2_{n+1}$.
	\item Reupdate $\hat{\sigma}^2_n$ with $X_{n+1}, \hat{\kappa}_{n+1}$.
\end{enumerate}

Using a similar procedure, it is further possible to update $\hat{\kappa}$ iteratively.
However, we do not recommend doing so because the computation time will increase significantly with little improvement in statistical efficiency.

Second, updating $\hat{v}_{q,n}$ also requires smoothing parameters selection.
Denote these ancillary smoothing parameters by $s_n^\diamond$ and $t_n^\diamond$.
We notice that it is not possible to select the coefficients of $s_n^\diamond$ and $t_n^\diamond$ based on asymptotic theory as another nuisance parameter $\kappa_{p+q}$ will arise.
This circular problem is common and well-known in nonparametric estimation.
If we just consider the exponents, choosing $s_n^\diamond = \lfloor n^{1/(1+2p)} \rfloor$ may be too small and lead to poor empirical performance.
Therefore, we suggest
\[
s_n^\diamond
= \left\{
\begin{array}{ll}
	\lfloor n^{1/2} \rfloor, &n \le n_0; \\
	\left\lfloor \max\big[ n_0^{1/2}, (p+q) n^{1/(1+2p)} \big] \right\rfloor, &n > n_0, \\
\end{array}
\right.
\]
where $n_0 = 1000$ is a threshold for small sample size and the empirical adjustment $\lfloor n^{1/2} \rfloor$ is recommended by \citet{jones2006fw}.
The ancillary tapering parameter $t_n^\diamond$ may admit the same form except that the floor functions are replaced by the ceiling functions as $t_n^\diamond$ cannot be zero.
Although it is possible to select the coefficients according to some benchmark models such as an autoregressive model instead of using $p+q$, this kind of benchmarking is still optimal to a particular stochastic process only.
Hence we recommend the simple choice $p+q$ to solve the problem of small exponents when $p+q$ is large.

Third, the main smoothing parameters $s_n$ and $t_n$ can be lower bounded by $s_0$ and $t_0$ in finite sample.
This is because if users wrongly choose a large characteristic exponent $q$ without accounting for the empirical serial dependence, the computed $s_n$ and $t_n$ may be too small in finite sample.
In light of a similar problem, \citet{jones2006fw} suggested using $\lfloor n^{1/2} \rfloor$.
However, this is not adaptive and Theorem \ref{thm:mse} states that this is not $\mathcal{L}^2$-optimal.
To balance theoretical and practical considerations, we provide the possibility to lower bound $s_n$ and $t_n$ by $s_0$ and $t_0$ in our R package, where the default is $s_0 = t_0 = 5$.
The empirical adjustment $\lfloor n^{1/2} \rfloor$ is implemented for the aforementioned ancillary smoothing parameters $s_n^\diamond$ and $t_n^\diamond$ in small sample only.
Example \ref{eg:auto} shows that recursive estimation with Algorithm \ref{algo:auto} performs very close to that with the oracle and much better than that with a pilot study.

\begin{remark} \label{rmk:vq}
	By Theorem \ref{thm:mse}, a larger $q$ improves the convergence rate of $\hat{\sigma}_{n,\LASER(q,\phi)}^2$ subject to the empirical serial dependence.
	A larger $p$ improves $\hat{v}_{q,n,\LASER(p,\phi)}$ similarly but this improvement has little effect on the long-run variance estimate.
\end{remark}

\subsection{Multivariate Extension} \label{sec:algo-multi}

In the main text, we focus on long-run variance estimation to introduce new ideas.
Nevertheless, our framework is also applicable to long-run covariance matrix estimation.
Consider a $d$-dimensional time series $\{ X_i=(X_{i}^{(1)}, \ldots, X_{i}^{(d)})^\T \}_{i\in\Z}$, where $d\in\Z^+$.
The long-run covariance matrix is defined as $\Sigma = \lim_{n\to\infty} n\Var(\bar{X}_n)$.
We can estimate $\Sigma$ by
\begin{equation} \label{eq:genClassMulti}
	\hat{\Sigma}_n = \hat{\Sigma}_n(W)
	= \frac{1}{n} \sum_{i=1}^n \sum_{j=1}^n W_n(i, j) (X_i -\bar{X}_n)
	(X_j -\bar{X}_n)^\T.
\end{equation}
Using the symmetry of \cref{eq:genClassMulti}, the $(h,k)$-th entry of $\hat{\Sigma}_n$ can be expressed as
\begin{align}
	\hat{\Sigma}_n^{(h,k)}
	&= \frac{1}{2n} \sum_{i=1}^n \sum_{j=1}^n W_n(i, j)
	\left( \hat{X}_i^{(h)} \hat{X}_j^{(k)}
	+\hat{X}_i^{(k)} \hat{X}_j^{(h)} \right) \nonumber \\
	&= \frac{1}{2n} \sum_{i=1}^n \sum_{j=1}^n W_n(i, j)
	\left\{ \left( \hat{X}_i^{(h)}+\hat{X}_i^{(k)}\right)
	\left( \hat{X}_j^{(h)}+\hat{X}_j^{(k)} \right)
	-\hat{X}_i^{(h)}\hat{X}_j^{(h)}
	-\hat{X}_i^{(k)}\hat{X}_j^{(k)} \right\} \nonumber \\
	&= \frac{1}{2} \left\{ \hat{\sigma}^2_n(W;X^{(h)}+X^{(k)})
	-\hat{\sigma}^2_n(W;X^{(h)})
	-\hat{\sigma}^2_n(W;X^{(k)}) \right\}, \label{eq:genClassMultiEntry}
\end{align}
where $\hat{X}_i^{(k)} = X_i^{(k)} - \bar{X}_n^{(k)}$ is the demeaned version of $X_i^{(k)}$ for $i=1, \ldots,n$ and $k=1,\ldots,d$; and
$\hat{\sigma}^2_{n}(W;Y)$ denotes the long-run variance estimator $\hat{\sigma}_n^2(W)$ if (univariate) data $Y_{1:n}$ are used.
In view of \cref{eq:genClassMultiEntry}, \hyperref[eq:genClassMulti]{$\hat{\Sigma}_n$} can be expressed as a function of $d(d+1)/2$ long-run variance estimators $\hat{\sigma}_n^2$ with different input samples.
Hence, all computational and statistical properties discussed before are also enjoyed by
\hyperref[eq:genClassMulti]{$\hat{\Sigma}_n$} without the need of deriving new formulas or theories.

In practice, users are usually interested in a scalar statistic $c^\T \bar{X}_n$, instead of the whole vector $\bar{X}_n$, for some $c\in\R^d$.
Then, the smoothing parameters in $\hat{\Sigma}_n$ can be selected so that the long-run variance estimator of the reference statistic, i.e., $\hat{\sigma}^2_{n}(W; c^\T X)$, is optimized.
The vector $c$ can be interpreted as the relative importance of the $d$ entries in $X$.
The selected smoothing parameters can be used for all entries, which enable efficient matrix operations and enhance computational efficiency.
An additional advantage of this approach is that users only need to specify the relative importance of the $d$ entries in $X$.
In contrast, \citet{rtacm} required users to specify the relative importance of $d\times d$ entries in $\Sigma$.
It is not only difficult to interpret but also hard to determine the relative importance.

\section{\texorpdfstring{Proof of Theorem \ref{thm:consistency}}{
		Proof of Theorem 3.1}} \label{sec:proof-consistency} 

The proof is divided into 5 steps, which are stated in \crefrange{sec:proof-consistency-mean}{sec:proof-consistency-special}.
Each step requires some technical lemmas, whose proofs are deferred to \cref{sec:proof-of-lemmas}.

\subsection{Replacement of Sample Mean} \label{sec:proof-consistency-mean}

We prove Theorem \ref{thm:consistency}(a) first.
By Assumption \ref{asum:winGen}(b) and the symmetry of $\gamma_k$, we can focus on the following rewritten form:
\begin{equation} \label{eq:sigmaClass}
	\hat{\sigma}^2_n
	= \frac{1}{n} \sum_{i=1}^n (X_i -\bar{X}_n)^2 + \frac{2}{n} \sum_{i=2}^n \sum_{j=1}^{i-1} W_n(i,j) (X_i -\bar{X}_n) (X_j -\bar{X}_n).
\end{equation}
Define the known-mean version of \hyperref[eq:sigmaClass]{$\hat{\sigma}^2_n$} as
\begin{equation} \label{eq:sigmaMean}
	\bar{\sigma}^2_n
	= \frac{1}{n} \sum_{i=1}^n (X_i -\mu)^2 + \frac{2}{n} \sum_{i=2}^n \sum_{j=1}^{i-1} W_n(i,j) (X_i -\mu) (X_j -\mu).
\end{equation}
By the Minkowski inequality, we have
\[
\norm{ \hat{\sigma}^2_n - \sigma^2 }_{\frac{\alpha}{2}}
\le \norm{ \hat{\sigma}^2_n -\bar{\sigma}^2_n }_{\frac{\alpha}{2}}
+ \norm{ \bar{\sigma}^2_n -\sigma^2 }_{\frac{\alpha}{2}}.
\]
Lemma \ref{lem:eqvMean} states that
$\norm{ \hat{\sigma}^2_n -\bar{\sigma}^2_n }_{\alpha/2} = o(1)$.
Hence it suffices to show that the known-mean version satisfies
$\norm{ \bar{\sigma}^2_n -\sigma^2 }_{\alpha/2} = o(1)$.

Without loss of generality, assume $\mu = 0$ for the remaining of this proof.
Otherwise, consider the demeaned series $\{X_i -\mu\}$.

\subsection{\texorpdfstring{Approximation by $m$-dependent Process}{
		Approximation by m-dependent Process}} \label{sec:proof-consistency-mDepend}

Define the $m$-dependent process and $m$-dependent process approximated version of \hyperref[eq:sigmaMean]{$\bar{\sigma}^2_n$} as
\begin{align}
	\tilde{X}_i
	&= \E(X_i \mid \epsilon_{i-m}, \ldots, \epsilon_i), \label{eq:xMDepend} \\
	\tilde{\sigma}^2_n
	&= \frac{1}{n} \sum_{i=1}^n \tilde{X}_i^2 + \frac{2}{n} \sum_{i=2}^n \sum_{j=1}^{i-1} W_n(i,j) \tilde{X}_i \tilde{X}_j, \label{eq:sigmaMDepend}
\end{align}
respectively.
By the Minkowski inequality, we have
\[
\norm{ \bar{\sigma}^2_n -\sigma^2 }_{\frac{\alpha}{2}}
\le \norm{ \bar{\sigma}^2_n -\E(\bar{\sigma}^2_n) -\tilde{\sigma}^2_n +\E(\tilde{\sigma}^2_n) }_{\frac{\alpha}{2}}
+ \norm{ \tilde{\sigma}^2_n -\E(\tilde{\sigma}^2_n) +\E(\bar{\sigma}^2_n) -\sigma^2 }_{\frac{\alpha}{2}}.
\]
Lemma \ref{lem:eqvMDepend} states that
$\norm{ \bar{\sigma}^2_n -\E(\bar{\sigma}^2_n) -\tilde{\sigma}^2_n +\E(\tilde{\sigma}^2_n) }_{\alpha/2} = o(1)$.
Thus, it remains to show that \linebreak
$\norm{ \tilde{\sigma}^2_n -\E(\tilde{\sigma}^2_n) +\E(\bar{\sigma}^2_n) -\sigma^2 }_{\alpha/2} = o(1)$.

\subsection{Approximation by Martingale Difference} \label{sec:proof-consistency-mart}

For $i=1,2,\ldots,n$, define
\begin{align}
	M_i &= \sum_{h=0}^\infty \E(\tilde{X}_{i+h} \mid \mathcal{F}_i), \label{eq:aMart} \\
	D_i &= M_i -\E(M_i \mid \mathcal{F}_{i-1}). \label{eq:dMart}
\end{align}
Note that $\{D_i\}$ is a martingale difference sequence with respect to the filtration $\mathcal{F}_i$.
Therefore, the martingale difference approximated version of \hyperref[eq:sigmaMDepend]{$\tilde{\sigma}^2_n$} is
\begin{equation} \label{eq:sigmaMart}
	\breve{\sigma}^2_n = \frac{1}{n} \sum_{i=1}^n D_i^2 + \frac{2}{n} \sum_{i=2}^n \sum_{j=1}^{i-1} W_n(i,j) D_i D_j.
\end{equation}
Applying the Minkowski inequality again,
\[
\norm{ \tilde{\sigma}^2_n -\E(\tilde{\sigma}^2_n) +\E(\bar{\sigma}^2_n) -\sigma^2 }_{\frac{\alpha}{2}}
\le \norm{ \tilde{\sigma}^2_n -\E(\tilde{\sigma}^2_n) -\breve{\sigma}^2_n +\E(\breve{\sigma}^2_n) }_{\frac{\alpha}{2}}
+\norm{ \breve{\sigma}^2_n -\E(\breve{\sigma}^2_n) +\E(\bar{\sigma}^2_n) -\sigma^2 }_{\frac{\alpha}{2}}.
\]
By Lemma \ref{lem:eqvMart}, $\norm{ \tilde{\sigma}^2_n -\E(\tilde{\sigma}^2_n) -\breve{\sigma}^2_n +\E(\breve{\sigma}^2_n) }_{\alpha/2} = o(1)$.
Therefore, it suffices to prove that \linebreak
$\norm{ \breve{\sigma}^2_n -\E(\breve{\sigma}^2_n) +\E(\bar{\sigma}^2_n) -\sigma^2 }_{\alpha/2} = o(1)$.

\subsection{Application of Ergodic Theorem} \label{sec:proof-consistency-ergodic}

Since $D_i$'s are martingale differences, $\E(D_i D_j) = 0$ for $i \neq j$.
We can also choose $m$ such that $m \to \infty$ as $n \to \infty$ by Assumption \ref{asum:winSums}; see \cref{sec:proof-eqvMart}.
Therefore, by Lemma \ref{lem:varDMart},
$\E(\breve{\sigma}^2_n) \to \sigma^2$ as $n \to \infty$.
Using the Ergodic Theorem (see, e.g., \citealt{durrett2019probability}), we obtain
$\norm{ \breve{\sigma}^2_n -\sigma^2 }_{\alpha/2} = o(1)$.
On the other hand,
\begin{align*}
	\E(\bar{\sigma}^2_n)
	&= \gamma_0 +\frac{2}{n} \sum_{i=2}^n \sum_{j=1}^{i-1} W_n(i,j) \gamma_{i-j}
	= \gamma_0 +\frac{2}{n} \sum_{i=2}^n \sum_{k=1}^{i-1} W_n(i,i-k) \gamma_k \\
	&= \gamma_0 +2 \sum_{k=1}^{n-1} \frac{1}{n} \sum_{i=k+1}^n W_n(i,i-k) \gamma_k \\
	&= \gamma_0 +2 \sum_{1 \le k \le b_n} w_{n,k} \gamma_k +2 \sum_{b_n < k < n} w_{n,k} \gamma_k,
\end{align*}
where $w_{n,k} = n^{-1} \sum_{i=k+1}^n W_n(i,i-k)$ is defined in Assumption \ref{asum:winGen}.
By the definition of long-run variance $\sigma^2 = \gamma_0 +2\sum_{k \in \Z^+} \gamma_k$ and Minkowski inequality, we have
\begin{align*}
	\left| \E(\bar{\sigma}^2_n) -\sigma^2 \right|
	&= \left| 2 \sum_{1 \le k \le b_n} (w_{n,k}-1) \gamma_k +2 \sum_{b_n < k < n} (w_{n,k}-1) \gamma_k -2 \sum_{k=n}^\infty \gamma_k \right| \\
	&\le 2 \sum_{1 \le k \le b_n} |w_{n,k}-1| |\gamma_k|  +2 \sum_{b_n < k < n} (|w_{n,k}|+1) |\gamma_k| +2 \sum_{k=n}^\infty |\gamma_k|.
\end{align*}
Since $\sum_{k \in \Z} |\gamma_k| < \infty$ under Assumption \ref{asum:stability} \citep{wu2009recursive},
we have $\sum_{k=n}^\infty |\gamma_k| = o(1)$,
\begin{align*}
	\sum_{1 \le k \le b_n} |w_{n,k}-1| |\gamma_k|
	&\le \max_{1 \le k \le b_n} |w_{n,k}-1| \sum_{1 \le k \le b_n} |\gamma_k|
	= o(1), \\
	\sum_{b_n < k < n} (|w_{n,k}|+1) |\gamma_k|
	&= \left( \max_{b_n < k < n} |w_{n,k}|+1 \right) \sum_{b_n < k < n} |\gamma_k|
	= o(1).
\end{align*}
It follows that $\left| \E(\bar{\sigma}^2_n) -\sigma^2 \right| = o(1)$ and
\[
\left| \E(\bar{\sigma}^2_n) -\E(\breve{\sigma}^2_n) \right|
\le \left| \E(\bar{\sigma}^2_n) -\sigma^2 \right| +\left| \E(\breve{\sigma}^2_n) -\sigma^2 \right|
= o(1).
\]
Combining the results using the Minkowski inequality completes the proof of Theorem \ref{thm:consistency}(a):
\[
\norm{ \breve{\sigma}^2_n -\E(\breve{\sigma}^2_n) +\E(\bar{\sigma}^2_n) -\sigma^2 }_{\frac{\alpha}{2}}
\le \norm{ \breve{\sigma}^2_n -\sigma^2 }_{\frac{\alpha}{2}} + \left| \E(\bar{\sigma}^2_n) -\E(\breve{\sigma}^2_n) \right|
= o(1).
\]

\subsection{Verification of Special Cases} \label{sec:proof-consistency-special}

Theorem \ref{thm:consistency}(b) is a special case of Theorem \ref{thm:consistency}(a).
Therefore, it suffices to verify Assumptions \ref{asum:winSums} and \ref{asum:winGen} under Definition \ref{def:winLASER}.
We verify Assumption \ref{asum:winSums} first.
For $G_{1,n}$, we have
\begin{align}
	G_{1,n} &= \max_{1 \le i \le n} \sum_{j=1}^n \left\{ \left| 1 -\frac{|i-j|^q}{t_n^q} \right| \I_{|i-j|\le s'_{i \lor j}}
	+\left( 1 -\frac{|i-j|^q}{t_n^q} \right)^2 \I_{|i-j|\le s'_{i \lor j}} \right\} \nonumber \\
	&\le \max_{1 \le i \le n} \sum_{j=1}^n \left| 1 -\frac{|i-j|^q}{t_n^q} \right| \I_{|i-j|\le \phi s_n}
	+\max_{1 \le i \le n} \sum_{j=1}^n \left( 1 -\frac{|i-j|^q}{t_n^q} \right)^2 \I_{|i-j|\le \phi s_n} \label{eq:rampG1} \\
	&\le 2 \sum_{k=0}^{\phi s_n} \left| 1 -\frac{k^q}{t_n^q} \right| +2 \sum_{k=0}^{\phi s_n} \left( 1 -\frac{k^q}{t_n^q} \right)^2 \nonumber \\
	&\le O(1) \cdot \int_0^{\phi s_n} \left| 1 -\frac{k^q}{t_n^q} \right| \,\di k +O(1) \cdot \int_0^{\phi s_n} \left( 1 -\frac{k^q}{t_n^q} \right)^2 \,\di k \label{eq:riemannG1} \\
	&\le O(n^\psi +n^{q(\psi-\theta)+\psi}) +O(n^\psi +n^{2q(\psi-\theta)+\psi})
	= O(n^{\psi +\max\{ 2q(\psi-\theta), 0\}}), \label{eq:boundG1}
\end{align}
where \cref{eq:rampG1} follows from the fact that $s'_{i \lor j} \le \phi s_n$ for $1 \le i,j \le n$, and
\cref{eq:riemannG1} follows from approximating the sums by Riemann integrals.
When $\psi \le \theta$, check that $\psi +\max\{ 2q(\psi-\theta), 0 \} < 2-2/\alpha'$.
When $\psi > \theta$, we have
\[
\psi +\max\{ 2q(\psi-\theta), 0 \}
< \psi -2q \cdot \frac{\psi-2+2/\alpha'}{2q}
= 2 -\frac{2}{\alpha'}.
\]
For $G_{2,n}$, observe that
\begin{align}
	G_{2,n} &= \max_{1 \le i,j \le n} \big| W_n(i,j) \big|
	+\max_{1 \le i \le n} \left\{ \sum_{j=2}^n \big| W_n(i,j) -W_n(i,j-1) \big|^{\alpha'} \right\}^{\frac{1}{\alpha'}} \nonumber \\
	&\le \max\left( \left| 1-\frac{\phi^q s_n^q}{t_n^q} \right|, 1 \right) +2 \sum_{k=0}^{\phi s_n} \left| \frac{(k+1)^q -k^q}{t_n^q} \right| \nonumber \\
	&\le O( n^{\max\{ q(\psi -\theta), 0 \}} ) + O(1) \cdot \int_0^{\phi s_n} \frac{k^{q-1}}{t_n^q} \,\di k \label{eq:riemannG2} \\
	&= O( n^{\max\{ q(\psi -\theta), 0 \}} +n^{q(\psi -\theta)} )
	= O( n^{\max\{ q(\psi -\theta), 0 \}} ), \label{eq:boundG2}
\end{align}
where \cref{eq:riemannG2} follows from approximating the sum by a Riemann integral.
When $\psi \le \theta$, check that $c \in (0, 1-1/\alpha')$ satisfies
$\max\{ q(\psi -\theta), 0 \} = 0 < c$.
When $\psi > \theta$, there exists $c \in (1-1/\alpha'-\psi/2, 1-1/\alpha')$ such that
\[
\max\{ q(\psi -\theta), 0 \}
< -q \cdot \frac{\psi-2+2/\alpha'}{2q} < c.
\]
Next, we verify Assumption \ref{asum:winGen}.
It is clear that $W_n(i,j) = W_n(j,i)$ and $W_n(i,i)=1$ for all $n\in\Z^+$ and $i,j \in \{1, \ldots, n\}$ under Definition \ref{def:winLASER}.
In addition, note that
\begin{align*}
	\max_{1 \le k \le b_n} |w_{n,k}-1|
	&= \max_{1 \le k \le b_n} \left| \frac{1}{n} \sum_{i=k+1}^n \left( 1 -\frac{k^q}{t_n^q} \right) \I_{k \le s'_i} -1 \right| \\
	&= \max_{1 \le k \le b_n} \left| -\frac{1}{n} \sum_{i=k+1}^n \frac{k^q}{t_n^q} \I_{k \le s'_i} -\frac{1}{n} \sum_{i=k+1}^n \I_{k > s'_i} -\frac{k}{n} \right| \\
	&\le \max_{1 \le k \le b_n} \left| \frac{1}{n} \sum_{i=k+1}^n \frac{k^q}{t_n^q} \I_{k \le s'_i} \right| +\frac{1}{n} \sum_{i=1}^n \I_{b_n > s'_i} +\frac{b_n}{n}, \\
	\max_{b_n < k < n} |w_{n,k}|
	&= \max_{b_n < k < n} \left| \frac{1}{n} \sum_{i=k+1}^n \left( 1 -\frac{k^q}{t_n^q} \right) \I_{k \le s'_i} \right| \\
	&\le \max_{b_n < k < n} \left| \frac{1}{n} \sum_{i=k+1}^n \frac{k^q}{t_n^q} \I_{k \le s'_i} \right| +\frac{1}{n} \sum_{i=1}^n \I_{b_n \le s'_i}.
\end{align*}
When $\psi \le \theta$, note that $s'_i \ge C n^{2\psi/3}$ for some $C \in \R^+$ and $i > n^{2/3}$.
We can choose $b_n = O(n^{\psi/2})$ such that
\[
\frac{1}{n} \sum_{i=1}^n \I_{b_n > s'_i}
\le \frac{n^{2/3}}{n} +\frac{1}{n} \sum_{n^{2/3} < i \le n} \I_{b_n > C n^{2\psi/3}}
= o(1).
\]
Therefore,
\begin{align*}
	\max_{1 \le k \le b_n} \left| \frac{1}{n} \sum_{i=k+1}^n \frac{k^q}{t_n^q} \I_{k \le s'_i} \right| +\frac{1}{n} \sum_{i=1}^n \I_{b_n > s'_i} +\frac{b_n}{n}
	&\le \frac{1}{n} \sum_{i=1}^n \frac{b_n^q}{t_n^q} +o(1) +o(1)
	= o(1), \\
	\max_{b_n < k < n} \left| \frac{1}{n} \sum_{i=k+1}^n \frac{k^q}{t_n^q} \I_{k \le s'_i} \right| +\frac{1}{n} \sum_{i=1}^n \I_{b_n \le s'_i}
	&\le \frac{\phi^q}{n} \sum_{i=1}^n \frac{s_n^q}{t_n^q} +O(1)
	= O(1).
\end{align*}
When $\psi > \theta$, Assumption \ref{asum:winGen}(c) does not hold.
However, Assumption \ref{asum:winGen}(c) is only used to show that $\left| \E(\bar{\sigma}^2_n) -\sigma^2 \right| = o(1)$ in \cref{sec:proof-consistency-ergodic}.
We prove the same thing under Definition \ref{def:winLASER}(b) and Assumption \ref{asum:uq}.
Consider the case $\phi=1$, by \cref{eq:interBias} we have
\[
\E(\bar{\sigma}_n^2)
\sim \sigma^2 - 2 \sum_{h=1}^{s_n} \frac{h^q}{t_n^q} \gamma_h -2 \sum_{h=1}^{s_n} \left( \frac{h}{s_n} \right)^{\frac{1}{\psi}} \left( 1 -\frac{h^q}{t_n^q} \right) \gamma_h.
\]
Since $u_{\tilde{q}} < \infty$ under Assumption \ref{asum:uq}, Kronecker's lemma gives
\begin{align*}
	\sum_{h=1}^{s_n} \frac{h^q}{t_n^q} \gamma_h
	&= \frac{s_n^{q-\tilde{q}}}{t_n^q} \sum_{h=1}^{s_n} \frac{h^{q-\tilde{q}}}{s_n^{q-\tilde{q}}} h^{\tilde{q}} \gamma_h
	= O( n^{q(\psi-\theta)-\tilde{q} \psi} ) o(1)
	= o( n^{q(\psi-\theta)-\tilde{q} \psi} ), \\
	\sum_{h=1}^{s_n} \left( \frac{h}{s_n} \right)^{\frac{1}{\psi}} \gamma_h
	&= \frac{1}{s_n^{\tilde{q}}} \sum_{h=1}^{s_n} \frac{h^{1/\psi}}{s_n^{1/\psi-\tilde{q}}} \gamma_h \\
	&= \frac{1}{s_n^{\tilde{q}}} \sum_{h=1}^{\lceil s_n^{1-\tilde{q}\psi} \rceil} \frac{h^{1/\psi}}{\big(s_n^{1-\tilde{q}\psi} \big)^{1/\psi}} \gamma_h
	+\frac{1}{s_n^{\tilde{q}}} \sum_{h=\lceil s_n^{1-\tilde{q}\psi} \rceil +1}^{s_n} \frac{\big(h^{1-\tilde{q}\psi} \big)^{1/\psi}}{\big(s_n^{1-\tilde{q}\psi} \big)^{1/\psi}} h^{\tilde{q}} \gamma_h \\
	&= O(n^{-\tilde{q} \psi}) o(1) +O(n^{-\tilde{q} \psi}) o(1)
	= o(n^{-\tilde{q} \psi}), \\
	\sum_{h=1}^{s_n} \left( \frac{h}{s_n} \right)^{\frac{1}{\psi}} \frac{h^q}{t_n^q} \gamma_h
	&= \frac{s_n^{q-\tilde{q}}}{t_n^q} \sum_{h=1}^{s_n} \frac{h^{q-\tilde{q}+1/\psi}}{s_n^{q-\tilde{q}+1/\psi}} h^{\tilde{q}} \gamma_h
	= o( n^{q(\psi-\theta)-\tilde{q} \psi} ).
\end{align*}
Check that
$
q(\psi-\theta)-\tilde{q}\psi < 0
\Leftrightarrow (q-\tilde{q})/q \psi < \theta,
$
which always holds under Definition \ref{def:winLASER}(b).
The case $\phi>1$ can be similarly proved using \cref{eq:interBiasRamp}.
\hfill$\qedsymbol$

\section{\texorpdfstring{Proof of Theorem \ref{thm:mse}}{
		Proof of Theorem 3.2}} \label{sec:proof-mse} 

The proof is divided into 6 steps, which are stated in \crefrange{sec:proof-var-mean}{sec:proof-bias}.
Each step requires some technical lemmas, whose proofs are deferred to \cref{sec:proof-of-lemmas}.
Before beginning the proof, we explain some blocking variables that will be used later:

\begin{align}
	\nu_k &= \left\lceil \frac{1}{(\phi-1) \psi} \left( \frac{1}{\Psi} \right)^{\frac{1}{\psi}} k^{\frac{1}{\psi}-2} \right\rceil, \label{eq:blockNo} \\
	\varpi_k &= \lceil (\phi-1)k \rceil, \label{eq:blockWidth} \\
	\eta_k &= \left\{
	\begin{array}{ll}
		1 +\sum_{h=1}^{k-1} \varpi_h \cdot \nu_h, &\quad \phi \in (1, \infty); \\
		\left\lceil \left( \frac{k}{\Psi} \right)^{\frac{1}{\psi}} \right\rceil, &\quad \phi = 1. \\
	\end{array}
	\right. \label{eq:blockStart}
\end{align}
The first block includes both $\{i \in \Z^+: s_i=0\}$ and $\{i \in \Z^+: s_i=1\}$.
For $\phi > 1$, we can partition $\{1,\ldots,n\}$ into $s_n$ blocks of data with ramps as the subsample size will be ramped:
\[
\{1, \ldots, n\}
= \underbrace{\Bigg\{ \bigcup_{h=1}^{s_n-1}
	\overbrace{\bigcup_{i=1}^{\nu_h}
		\overbrace{\bigcup_{j=\eta_h+(i-1)\varpi_h}^{\eta_h+i\varpi_h-1} \{j\}}^{i \text{-th ramp}}
	}^{h \text{-th block}} \Bigg\}}_{s_n-1 \text{ complete block}}
\bigcup \underbrace{\Bigg\{ \bigcup_{j=\eta_{s_n}}^n \{j\} \Bigg\}}_\text{last incomplete block},
\]
where $\cup_{i=a}^{b} E_i = \emptyset$ if $a > b$, for any sets $E_i$'s.
We visualize an example in Figure \ref{fig:blocking}.
Note that the physical meaning of $\eta_k$ is the minimum sample size such that the intended subsample size is $k$.
For $\phi=1$, this can be seen from
\[
\Psi\eta_k^\psi \sim k
\Leftrightarrow \eta_k \sim \left( \frac{k}{\Psi} \right)^{\frac{1}{\psi}}.
\]
For $\phi>1$, this interpretation also holds but we define $\eta_k$ based on $\nu_k$ and $\varpi_k$ instead, where
$\nu_k$ is the number of ramps and
$\varpi_k$ is the width of a ramp in the $k$-th block:
\[
(\phi-1)k \cdot \nu_k \sim \eta_{k+1} -\eta_k
\Leftrightarrow \nu_k \sim \frac{1}{(\phi-1) \psi} \left( \frac{1}{\Psi} \right)^{\frac{1}{\psi}} k^{\frac{1}{\psi}-2}.
\]

\begin{figure}[!t]
	\centering
	\includegraphics[width=0.6\linewidth]{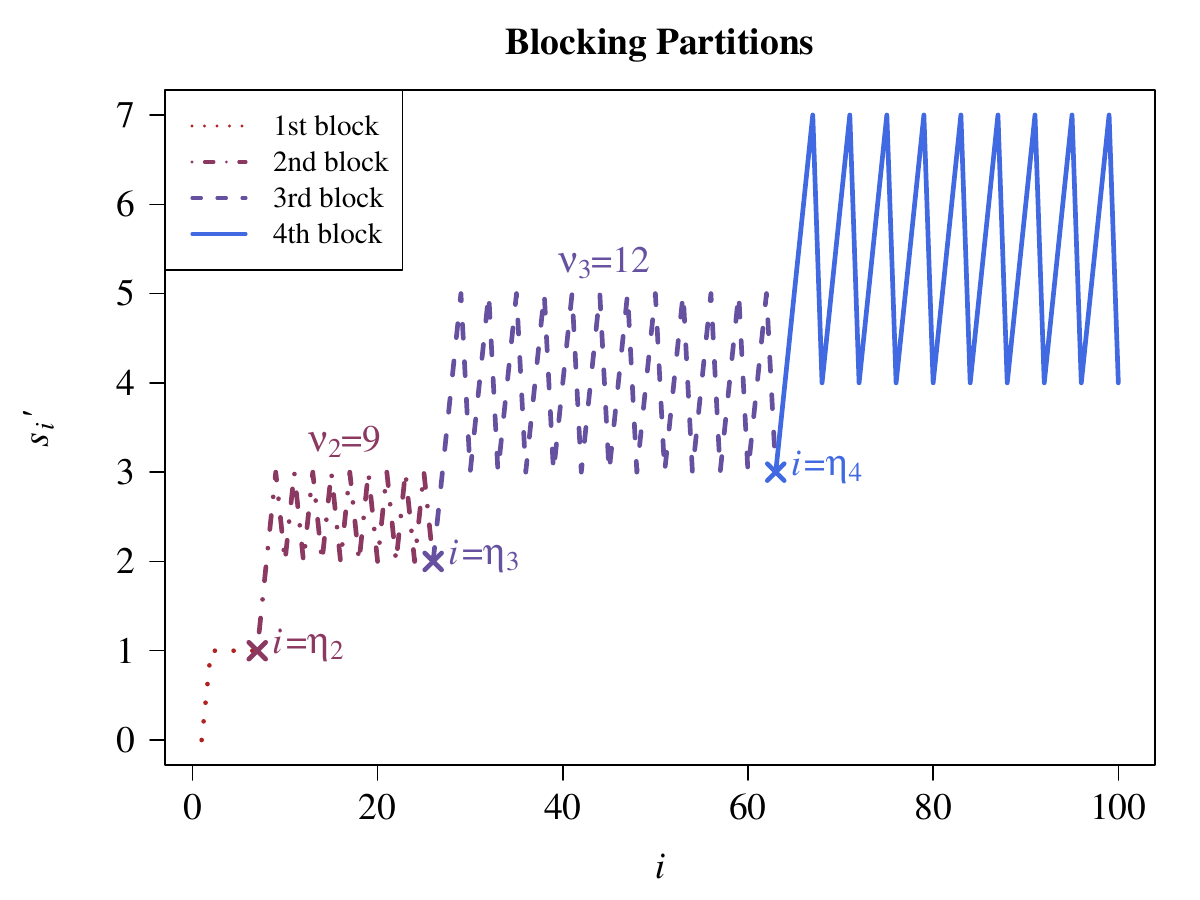}
	\caption{A blocking example with $n=100$, $\Psi=1$, $\psi=1/3$ and $\phi=2$.
		The 4th block (solid blue) is incomplete so we do not need to find $\nu_4$.}
	\label{fig:blocking}
\end{figure}

\subsection{Replacement of Sample Mean} \label{sec:proof-var-mean}

We find the exact convergence rate of variance first.
Recall the definitions of unknown-mean version $\hat{\sigma}_n^2$ and known-mean version $\bar{\sigma}_n^2$ in \cref{eq:sigmaClass,eq:sigmaMean}, respectively.
For simplicity, write $\norm{ \cdot }_2 = \norm{ \cdot }$.
We are going to approximate $\norm{\hat{\sigma}_n^2-\E(\hat{\sigma}_n^2)}$ by
$\norm{\bar{\sigma}_n^2-\E(\bar{\sigma}_n^2)}$.
To this end, we can use the following procedure:

First, by the Minkowski inequality,
\begin{align}
	\norm{ \hat{\sigma}^2_n -\E(\hat{\sigma}^2_n) }
	&\le \norm{ \bar{\sigma}^2_n -\E(\bar{\sigma}^2_n) } + \norm{ \hat{\sigma}^2_n -\bar{\sigma}^2_n } + \left| \E(\hat{\sigma}^2_n) - \E(\bar{\sigma}^2_n) \right| \nonumber \\
	&= \norm{ \bar{\sigma}^2_n -\E(\bar{\sigma}^2_n) } +H_1 +H_2. \label{eq:interVarMean}
\end{align}
By similar steps as in \cref{eq:boundG1},
\[
G_{3,n} = O(n^{2\psi +\max\{ 2q(\psi-\theta), 0\} +1}) \quad \text{and} \quad
G_{4,n} = O(n^{\psi +\max\{ q(\psi-\theta), 0\} +1}),
\]
where $G_{3,n}$ and $G_{4,n}$ are defined in Lemma \ref{lem:eqvMean}.
By Lemma \ref{lem:eqvMean} and the Jensen's inequality,
\begin{align*}
	H_1 &= O(n^{-1}) +O\left( n^{-\frac{3}{2}} G_{3,n}^{\frac{1}{2}} \right) +O\left( n^{-2} G_{4,n} \right)
	= O\left( n^{\psi +\max\{ q(\psi-\theta), 0\} -1} \right), \\
	H_2 &\le \E \left| \hat{\sigma}^2_n - \bar{\sigma}^2_n \right|
	\le H_1 = O\left( n^{\psi +\max\{ q(\psi-\theta), 0\} -1} \right).
\end{align*}
\cref{eq:interVarMean} now becomes
\begin{equation} \label{eq:upperVarMean}
	\norm{ \hat{\sigma}^2_n -\E(\hat{\sigma}^2_n) }
	\le \norm{ \bar{\sigma}^2_n -\E(\bar{\sigma}^2_n) } + O\left( n^{\psi +\max\{ q(\psi-\theta), 0\} -1} \right).
\end{equation}
On the other hand, applying the Minkowski inequality to
$\norm{ \bar{\sigma}^2_n -\E(\bar{\sigma}^2_n) }$ yields
\begin{align*}
	\norm{ \bar{\sigma}^2_n -\E(\bar{\sigma}^2_n) }
	&\le \norm{ \hat{\sigma}^2_n -\E(\hat{\sigma}^2_n) } + H_1 + H_2 \\
	&= \norm{ \hat{\sigma}^2_n -\E(\hat{\sigma}^2_n) } +O\left( n^{\psi +\max\{ q(\psi-\theta), 0\} -1} \right),
\end{align*}
which implies
\begin{equation} \label{eq:lowerVarMean}
	\norm{ \hat{\sigma}^2_n -\E(\hat{\sigma}^2_n) }
	\ge \norm{ \bar{\sigma}^2_n -\E(\bar{\sigma}^2_n) } + O\left( n^{\psi +\max\{ q(\psi-\theta), 0\} -1} \right).
\end{equation}
Check that
$2\psi +\max\{ 2q(\psi-\theta), 0\} -2 < \psi +\max\{ 2q(\psi-\theta), 0\} -1$
always holds under Definition \ref{def:winLASER} because $\psi < 1/(1+q) < 1$.
Hence, it suffices to show that
\begin{equation} \label{eq:boundVarMean}
	\norm{ \bar{\sigma}^2_n -\E(\bar{\sigma}^2_n) }^2
	\sim \mathcal{V}_{\psi,\Psi,\theta,\Theta,q,\phi} \sigma^4 n^{\psi +\max\{ 2q(\psi-\theta), 0\} -1}.
\end{equation}

Without loss of generality, assume $\mu = 0$ for the remaining of this proof.
Otherwise, consider the demeaned series $\{X_i -\mu\}$.

\subsection{Approximation by Blocking} \label{sec:proof-var-block}

When $\phi > 1$, we can partition the subsamples into $s_n$ blocks using \crefrange{eq:blockNo}{eq:blockStart}:
\begin{align}
	\bar{\sigma}^2_n
	={}& \frac{1}{n} \sum_{i=1}^n X_i^2 + \frac{2}{n} \sum_{i=2}^n \sum_{j=1}^{i-1} W_n(i,j) X_i X_j \nonumber \\
	={}& \frac{1}{n} \sum_{i=1}^n X_i^2 + \frac{2}{n} \sum_{i=2}^n \sum_{k=1}^{s'_i} \left( 1 -\frac{k^q}{t_n^q} \right) X_i X_{i-k} \nonumber \\
	={}& \frac{2}{n} \sum_{h=1}^{s_n-1} \sum_{i=1}^{\nu_h} \sum_{j=\eta_h+(i-1)\varpi_h}^{\eta_h+i\varpi_h-1} \sum_{k=1}^{j-\eta_h-(i-1)\varpi_h+h} \left( 1 -\frac{k^q}{t_n^q} \right) X_j X_{j-k} \nonumber \\
	& +\frac{1}{n} \left\{ \sum_{i=1}^n X_i^2 + 2 \sum_{j=\eta_{s_n}}^n \sum_{k=1}^{s'_j} \left( 1 -\frac{k^q}{t_n^q} \right) X_j X_{j-k} \right\} \nonumber \\
	={}& \bar{\sigma}^2_{1,n} +\bar{\sigma}^2_{2,n}, \label{eq:interBlockRamp}
\end{align}
where $\bar{\sigma}^2_{2,n}$ is constructed from the sum of squared terms and the last incomplete block.
Then, we can use the same procedure from \crefrange{eq:interVarMean}{eq:boundVarMean} to argue that $\bar{\sigma}^2_{2,n}$ is asymptotically negligible in finding the variance.
First, by the Minkowski inequality,
\[
\norm{ \bar{\sigma}^2_n -\E(\bar{\sigma}^2_n) }
\le \norm{ \bar{\sigma}^2_{1,n} -\E(\bar{\sigma}^2_{1,n}) } +\norm{ \bar{\sigma}^2_{2,n} -\E(\bar{\sigma}^2_{2,n}) }.
\]
In view of Lemma \ref{lem:lastBlock},
\[
\norm{ \bar{\sigma}^2_{2,n} -\E(\bar{\sigma}^2_{2,n}) }^2
= o\big( n^{\psi +\max\{ 2q(\psi-\theta), 0\} -1} \big).
\]
Repeating the same arguments in \cref{eq:upperVarMean,eq:lowerVarMean}, it remains to prove
\[
\norm{ \bar{\sigma}^2_{1,n} -\E(\bar{\sigma}^2_{1,n}) }^2
\sim \mathcal{V}_{\psi,\Psi,\theta,\Theta,q,\phi} \sigma^4 n^{\psi +\max\{ 2q(\psi-\theta), 0\} -1},
\]
for the case $\phi > 1$.
If $\phi = 1$, we do not need to partition the subsamples as in \cref{eq:interBlockRamp} because the subsampling parameter is monotonically increasing instead of being ramped.
In this case, we consider
\begin{align}
	\bar{\sigma}^2_n &= \frac{1}{n} \sum_{i=1}^n X_i^2 + \frac{2}{n} \sum_{i=2}^n \sum_{k=1}^{s'_i} \left( 1 -\frac{k^q}{t_n^q} \right) X_i X_{i-k} \nonumber \\
	&= \frac{2}{n} \sum_{h=1}^{s_n} \sum_{i=\eta_h}^{n} \left( 1 -\frac{h^q}{t_n^q} \right) X_i X_{i-h} + \frac{1}{n} \sum_{i=1}^n X_i^2 \nonumber \\
	&= \bar{\sigma}^2_{3,n} +\bar{\sigma}^2_{4,n}, \label{eq:interBlock}
\end{align}
where $\bar{\sigma}^2_{4,n}$ is constructed from the sum of squared terms only.
Since $\norm{\bar{\sigma}^2_{4,n} -\E(\bar{\sigma}^2_{4,n})}^2 = O(n^{-1})$ by Lemma \ref{lem:lastBlock},
we can similarly argue that $\bar{\sigma}^2_{4,n}$ is asymptotically negligible in finding the variance.

With slight abuse of notation, we only write $\bar{\sigma}^2_n$ in \cref{sec:proof-var-mDepend,sec:proof-var-mart} for tidiness.
This is possible because the results on $\bar{\sigma}^2_n$ also apply to \hyperref[eq:interBlockRamp]{$\bar{\sigma}^2_{1,n}$} and \hyperref[eq:interBlock]{$\bar{\sigma}^2_{3,n}$}, which is equivalent to say that the last incomplete block or sum of squared terms does not exist.
We will separate the two cases when we try to find $\mathcal{V}_{\psi,\Psi,\theta,\Theta,q,\phi}$.

\subsection{\texorpdfstring{Approximation by $m$-dependent Process}{
		Approximation by m-dependent process}} \label{sec:proof-var-mDepend}

Recall the definition of $m$-dependent process approximated version $\tilde{\sigma}^2_n$ in \cref{eq:sigmaMDepend}.
We are going to approximate $\norm{\bar{\sigma}_n^2-\E(\bar{\sigma}_n^2)}$ by $\norm{\tilde{\sigma}_n^2-\E(\tilde{\sigma}_n^2)}$.
To this end, we can use the same procedure from \crefrange{eq:interVarMean}{eq:boundVarMean} again.
By the Minkowski inequality,
\[
\norm{ \bar{\sigma}^2_n -\E(\bar{\sigma}^2_n) }
\le \norm{ \tilde{\sigma}_n^2-\E(\tilde{\sigma}_n^2) } +\norm{ \bar{\sigma}^2_n -\E(\bar{\sigma}_n^2) -\tilde{\sigma}^2_n +\E(\tilde{\sigma}_n^2) }.
\]
Since $\psi > 0$, we can always choose $m = o(n^{\psi/6})$ such that $m \to \infty$ as $n \to \infty$.
Then, $d_{m, \alpha} = o(1)$ as $n \to \infty$ by Lemma \ref{lem:dMDepend}, and
$G_{5,n} = O(n^{\psi +\max\{2q(\psi -\theta),0\} })$ by \cref{eq:boundG1}, where
$d_{m,\alpha}$ and $G_{5,n}$ are defined in \cref{eq:dMDepend} and Lemma \ref{lem:eqvMDepend}, respectively.
For $\alpha \ge 4$, Lemma \ref{lem:eqvMDepend} states that
\[
\norm{ \bar{\sigma}^2_n -\E(\bar{\sigma}^2_n) -\tilde{\sigma}^2_n +\E(\tilde{\sigma}^2_n) }
= O\left( n^{-\frac{1}{2}} d_{m,\alpha} G_{5,n}^{\frac{1}{2}} \right)
= o\left( n^{\frac{\psi}{2} +\max\{q(\psi -\theta),0\} -\frac{1}{2}} \right).
\]
Repeating the same arguments as in \cref{eq:upperVarMean,eq:lowerVarMean}, it remains to show that
\[
\norm{ \tilde{\sigma}_n^2-\E(\tilde{\sigma}_n^2) }^2
\sim \mathcal{V}_{\psi,\Psi,\theta,\Theta,q,\phi} \sigma^4 n^{\psi +\max\{ 2q(\psi-\theta), 0\} -1}.
\]

\subsection{Approximation by Martingale Difference}\label{sec:proof-var-mart}

Recall the definition of martingale difference approximated version $\breve{\sigma}^2_n$ in \cref{eq:sigmaMart}.
We can use the same procedure from \crefrange{eq:interVarMean}{eq:boundVarMean} again
to approximate $\norm{\tilde{\sigma}_n^2-\E(\tilde{\sigma}_n^2)}$ by
$\norm{\breve{\sigma}_n^2-\E(\breve{\sigma}_n^2)}$.
By the Minkowski inequality,
\[
\norm{ \tilde{\sigma}_n^2-\E(\tilde{\sigma}_n^2) }
\le \norm{ \breve{\sigma}_n^2-\E(\breve{\sigma}_n^2) } +\norm{ \tilde{\sigma}^2_n -\E(\tilde{\sigma}_n^2) -\breve{\sigma}_n^2 +\E(\breve{\sigma}_n^2) }.
\]
Note that $G_{6,n} = O(n^{\max\{ q(\psi -\theta), 0 \}})$ by \cref{eq:boundG2},
where $G_{6,n}$ is defined in Lemma \ref{lem:eqvMart}.
In \cref{sec:proof-var-mDepend}, we have chosen $m = o(n^{\psi/6})$.
It follows from Lemma \ref{lem:eqvMart} that
\[
\norm{ \tilde{\sigma}^2_n -\E(\tilde{\sigma}^2_n) -\breve{\sigma}^2_n +\E(\breve{\sigma}^2_n) }
= O\left( m^{\frac{5}{2}} n^{-\frac{1}{2}} G_{6,n} \right)
= o\left( n^{\frac{\psi}{2} +\max\{ q(\psi -\theta), 0 \} -\frac{1}{2}} \right),
\]
as $\alpha \ge 4$ implies $\alpha'=2$.
Repeating the same arguments as in \cref{eq:upperVarMean,eq:lowerVarMean}, it suffices to prove
\[
\norm{ \breve{\sigma}^2_n -\E(\breve{\sigma}^2_n) }^2
\sim \mathcal{V}_{\psi,\Psi,\theta,\Theta,q,\phi} \sigma^4 n^{\psi +\max\{ 2q(\psi-\theta), 0\} -1}.
\]

\subsection{Exact Convergence Rate of Variance} \label{sec:proof-var}

Consider the case $\phi > 1$.
We similarly define $\breve{\sigma}^2_{1,n}$ as in \cref{eq:interBlockRamp} with $X_j$ replaced by martingale difference \hyperref[eq:dMart]{$D_j$}.
By \cref{sec:proof-var-block}, we can work on
$\norm{ \breve{\sigma}^2_{1,n} -\E(\breve{\sigma}^2_{1,n}) }$ instead
of $\norm{ \breve{\sigma}^2_n -\E(\breve{\sigma}^2_n) }$.
Then, by $\E(D_iD_j) = 0$ when $i \ne j$ (see \cref{sec:proof-varDMart}) and the Minkowski inequality,
\begin{align}
	& \norm{ \breve{\sigma}^2_{1,n} -\E(\breve{\sigma}^2_{1,n}) } \nonumber \\
	={}& \frac{2}{n} \norm{ \sum_{h=1}^{s_n-1} \sum_{i=1}^{\nu_h} \sum_{j=\eta_h+(i-1)\varpi_h}^{\eta_h+i\varpi_h-1} \sum_{k=1}^{j-\eta_h-(i-1)\varpi_h+h} \left( 1 -\frac{k^q}{t_n^q} \right) D_j D_{j-k} } \nonumber \\
	\le{}& \frac{2}{n} \norm{ \sum_{h=\lceil (m+1)/\phi \rceil}^{s_n-1} \sum_{i=1}^{\nu_h} \sum_{j=\eta_h+(i-1)\varpi_h+m+1}^{\eta_h+i\varpi_h-1} \sum_{k=m+1}^{j-\eta_h-(i-1)\varpi_h+h} \left( 1 -\frac{k^q}{t_n^q} \right) D_j D_{j-k} } \nonumber \\
	& +\frac{2}{n} \norm{ \sum_{h=1}^{s_n-1} \sum_{i=1}^{\nu_h} \sum_{j=\eta_h+(i-1)\varpi_h}^{\eta_h+i\varpi_h-1} \sum_{k=1}^{\min\{j-\eta_h-(i-1)\varpi_h+h,m\}} \left( 1 -\frac{k^q}{t_n^q} \right) D_j D_{j-k} } \nonumber \\
	={}& H_3 +H_4, \label{eq:interVarRamp}
\end{align}
where $m$ is the ancillary variable used in $m$-dependent process approximation.
In \cref{sec:proof-var-mDepend}, we have chosen $m = o(n^{\psi/6})$,
which means $m = o(s_n)$ under Definition \ref{def:winLASER}.
As the physical meaning of $j-\eta_h-(i-1)\varpi_h+h$ in the above is the ramped subsample size at the $j$-th observation,
$m = o(s_n)$ implies that $\min\{j-\eta_h-(i-1)\varpi_h+h-1,m\} = m$ most of the time.
Consequently, $H_4 = o(H_3)$.
Repeating the same arguments as in \cref{eq:upperVarMean,eq:lowerVarMean}, we only need to prove
\begin{equation} \label{eq:varExactRamp}
	H_3^2 \sim \mathcal{V}_{\psi,\Psi,\theta,\Theta,q,\phi} \sigma^4 n^{\psi +\max\{ 2q(\psi-\theta), 0\} -1}.
\end{equation}
For $\phi = 1$, we can similarly define $\breve{\sigma}^2_{3,n}$ as in \cref{eq:interBlock} so that
\begin{align}
	& \norm{ \breve{\sigma}^2_{3,n} -\E(\breve{\sigma}^2_{3,n}) } \nonumber \\
	={}& \frac{2}{n} \norm{ \sum_{h=1}^{s_n} \sum_{i=\eta_h}^{n} \left( 1 -\frac{h^q}{t_n^q} \right) D_i D_{i-h} } \nonumber \\
	\le{}& \frac{2}{n} \norm{ \sum_{h=m+1}^{s_n} \sum_{i=\eta_h}^{n} \left( 1 -\frac{h^q}{t_n^q} \right) D_i D_{i-h} } + \frac{2}{n} \norm{ \sum_{h=1}^m \sum_{i=\eta_h}^{n} \left( 1 -\frac{h^q}{t_n^q} \right) D_i D_{i-h} } \nonumber \\
	={}& H_5 +H_6. \label{eq:interVar}
\end{align}
Repeating the arguments from \crefrange{eq:interVarRamp}{eq:varExactRamp}, we only need to prove
\begin{equation}
	H_5^2 \sim \mathcal{V}_{\psi,\Psi,\theta,\Theta,q,\phi} \sigma^4 n^{\psi +\max\{ 2q(\psi-\theta), 0\} -1}. \label{eq:varExact}
\end{equation}
Lemma \ref{lem:varExact} confirms \cref{eq:varExactRamp,eq:varExact}.

\subsection{Exact Convergence Rate of Bias} \label{sec:proof-bias}

Now, we find the exact convergence rate of bias.
To approximate $\Bias(\hat{\sigma}_n^2)$ by $\Bias(\bar{\sigma}_n^2)$, we use the same procedure from \crefrange{eq:interVarMean}{eq:boundVarMean}.
Note that
\[
\Bias(\hat{\sigma}_n^2)
= \Bias(\bar{\sigma}_n^2) +\E(\hat{\sigma}_n^2) -\E(\bar{\sigma}_n^2).
\]
By \cref{eq:interVarMean},
\[
|\E(\hat{\sigma}_n^2) -\E(\bar{\sigma}_n^2)| = H_2 = O\left( n^{\psi +\max\{ q(\psi-\theta), 0\} -1} \right).
\]
Check that $\psi -1 < -\psi$ and $\psi +q(\psi-\theta) -1 < -q \theta$ always hold under Definition \ref{def:winLASER} because $\psi < 1/(1+q) \le 1/2$.
Repeating the same arguments as in \cref{eq:upperVarMean,eq:lowerVarMean}, we can assume $\mu = 0$, and it suffices to prove
\[
\Bias(\bar{\sigma}_n^2) \sim \left\{
\begin{array}{ll}
	o(n^{-q\psi}), & \psi < \theta; \\
	-\Theta^{-q} n^{-q \theta} v_q, & \psi \ge \theta. \\
\end{array}
\right.
\]
To better illustrate the idea, we consider the case $\phi = 1$ first.
Recall that we can rearrange the order of summations in \cref{eq:interBlock}.
Taking expectation yields
\begin{align}
	\E(\bar{\sigma}_n^2) &= \frac{1}{n} \sum_{i=1}^n \gamma_0 + \frac{2}{n} \sum_{h=1}^{s_n} \sum_{i=\eta_h}^{n} \left( 1 -\frac{h^q}{t_n^q} \right) \gamma_h \nonumber \\
	&= \gamma_0 + 2 \sum_{h=1}^{s_n} \frac{n-\eta_h+1}{n} \left( 1 -\frac{h^q}{t_n^q} \right) \gamma_h \nonumber \\
	&\sim \gamma_0 + 2 \sum_{h=1}^{s_n} \left\{ 1 -\frac{h^q}{t_n^q} -\left( \frac{h}{s_n} \right)^{\frac{1}{\psi}} \left( 1 -\frac{h^q}{t_n^q} \right) \right\} \gamma_h, \label{eq:interBias}
\end{align}
where we have used the order of $\eta_h$ in \cref{eq:blockStart} to arrive at \cref{eq:interBias}.
As $n \to \infty$, $\gamma_0 + 2 \sum_{h=1}^{s_n} \gamma_h \to \sigma^2$.
On the other hand, by Kronecker's lemma and Assumption \ref{asum:qWeakStable} that $u_q = \sum_{j \in \Z} |j|^q |\gamma_j| < \infty$,
\begin{align*}
	\sum_{h=1}^{s_n} \left( \frac{h}{s_n} \right)^{\frac{1}{\psi}} \gamma_h
	&= \frac{1}{s_n^q} \sum_{h=1}^{s_n} \frac{h^{1/\psi}}{s_n^{1/\psi-q}} \gamma_h \\
	&= \frac{1}{s_n^q} \sum_{h=1}^{\lceil s_n^{1-q\psi} \rceil} \frac{h^{1/\psi}}{\big(s_n^{1-q\psi} \big)^{1/\psi}} \gamma_h
	+\frac{1}{s_n^q} \sum_{h=\lceil s_n^{1-q\psi} \rceil +1}^{s_n} \frac{\big(h^{1-q\psi} \big)^{1/\psi}}{\big(s_n^{1-q\psi} \big)^{1/\psi}} h^q \gamma_h \\
	&= \frac{1}{s_n^q} o(1) +\frac{1}{s_n^q} o(1)
	= o(s_n^{-q}).
\end{align*}
Check that $1-q\psi > 0$ always holds under Definition \ref{def:winLASER}.
Similarly,
\begin{equation}
	\sum_{h=1}^{s_n} \left( \frac{h}{s_n} \right)^{\frac{1}{\psi}} \frac{h^q}{t_n^q} \gamma_h
	= \frac{1}{t_n^q} \sum_{h=1}^{s_n} \left( \frac{h}{s_n} \right)^{\frac{1}{\psi}} h^q \gamma_h
	= \frac{1}{t_n^q} o(1)
	= o(t_n^{-q}). \label{eq:boundBias}
\end{equation}
Consequently, \cref{eq:interBias} becomes
\begin{align*}
	\E(\bar{\sigma}_n^2) &\sim \left\{
	\begin{array}{ll}
		\sigma^2 -2 \sum_{h=1}^{s_n} \left( h/s_n \right)^{\frac{1}{\psi}} \gamma_h +O(t_n^{-q}), & \psi < \theta; \\
		\sigma^2 -2 t_n^{-q} \sum_{h=1}^{s_n} h^q \gamma_h +o(t_n^{-q}), & \psi \ge \theta \\
	\end{array}
	\right. \\
	&= \left\{
	\begin{array}{ll}
		\sigma^2 +o(n^{-q \psi}), & \psi < \theta; \\
		\sigma^2 -\Theta^{-q} n^{-q \theta} v_q +o(n^{-q \theta}), & \psi \ge \theta. \\
	\end{array}
	\right.
\end{align*}
For $\phi > 1$, note that the partitions in \cref{eq:interBlockRamp} cannot give a form like \cref{eq:interBias} after taking expectation.
By rearranging the order of summations, approximating the sums by Riemann integrals and using the order of $\nu_i$ in \cref{eq:blockNo},
\begin{align*}
	\E(\bar{\sigma}_n^2)
	={}& \gamma_0 + \frac{2}{n} \sum_{h=1}^{s_n-1} \sum_{i=1}^{\nu_h} \sum_{j=\eta_h+(i-1)\varpi_h}^{\eta_h+i\varpi_h-1} \sum_{k=1}^{j-\eta_h-(i-1)\varpi_h+h} \left( 1 -\frac{k^q}{t_n^q} \right) \gamma_k +\frac{2}{n} \sum_{j=\eta_{s_n}}^n \sum_{k=1}^{s'_j} \left( 1 -\frac{k^q}{t_n^q} \right) \gamma_k \\
	\sim{}& \gamma_0 + \frac{2}{n} \sum_{k=1}^{\lceil \phi s_n \rceil} \gamma_k \int_{k/\phi}^{s_n} \int_0^{\nu_h} \int_{\max(k-h,0)}^{(\phi-1)h} \left( 1 -\frac{k^q}{t_n^q} \right) \, \di j \ \di i\ \di h \\
	=& \gamma_0 + \frac{2}{n} \sum_{h=1}^{\lceil \phi s_n \rceil} \gamma_h \int_{h/\phi}^{s_n} \int_0^{\nu_i} \int_{\max(h-i,0)}^{(\phi-1)i} \left( 1 -\frac{h^q}{t_n^q} \right) \, \di k\ \di j\ \di i \\
	={}& \gamma_0 + \frac{2}{n} \sum_{h=1}^{\lceil \phi s_n \rceil} \left( 1 -\frac{h^q}{t_n^q} \right) \gamma_h \int_{h/\phi}^{s_n} \nu_i \{ (\phi-1) i -\max(h-i,0) \} \,\di i \\
	\sim{}& \gamma_0 + \frac{2}{(\phi-1)\psi} \left( \frac{1}{s_n} \right)^{\frac{1}{\psi}} \sum_{h=1}^{s_n} \left( 1 -\frac{h^q}{t_n^q} \right) \gamma_h \left\{ \int_{h/\phi}^h (\phi i-h) i^{\frac{1}{\psi}-2} \,\di i + \int_h^{s_n} (\phi-1) i^{\frac{1}{\psi}-1} \,\di i \right\} \\
	& +\frac{2}{(\phi-1)\psi} \left( \frac{1}{s_n} \right)^{\frac{1}{\psi}} \sum_{h=s_n+1}^{\lceil \phi s_n \rceil} \left( 1 -\frac{h^q}{t_n^q} \right) \gamma_h \int_{h/\phi}^{s_n} (\phi i-h) i^{\frac{1}{\psi}-2} \,\di i.
\end{align*}
Evaluating the integrals gives
\begin{align}
	\E(\bar{\sigma}_n^2)
	\sim{}& \gamma_0 + 2 \sum_{h=1}^{s_n} \left( 1 -\frac{h^q}{t_n^q} \right) \left\{ 1 - \frac{\psi \left(\phi^{1 - 1/\psi} - 1\right)}{\left(\phi - 1\right) \left(\psi - 1\right)} \left( \frac{h}{s_n} \right)^{\frac{1}{\psi}} \right\} \gamma_h \nonumber \\
	+{}& 2 \sum_{h=s_n+1}^{\lceil \phi s_n \rceil} \left( 1 -\frac{h^q}{t_n^q} \right) \left\{ \frac{\phi}{\phi-1} +\frac{1}{(\phi-1)(\psi-1)} \cdot \frac{h}{s_n} - \frac{\phi^{1-1/\psi} \psi}{\left(\phi - 1\right) \left(\psi - 1\right)} \left( \frac{h}{s_n} \right)^{\frac{1}{\psi}} \right\} \gamma_h. \label{eq:interBiasRamp}
\end{align}
Repeating the arguments from \crefrange{eq:interBias}{eq:boundBias}, we have
\[
\E(\bar{\sigma}_n^2) \sim \left\{
\begin{array}{ll}
	\sigma^2 +o(n^{-q\psi}), & \psi < \theta; \\
	\sigma^2 -\Theta^{-q} n^{-q \theta} v_q +o(n^{-q\theta}), & \psi \ge \theta, \\
\end{array}
\right.
\eqno\qedsymbol
\]

\section{\texorpdfstring{Proof of Corollary \ref{coro:mse}}{
		Proof of Corollary 3.3}} \label{sec:proof-optimPar} 

To begin with, we try to find $\Psi_\star$ assuming we do not separate the smoothing parameters, i.e., $\Psi = \Theta$.
If $\psi = \theta = 1/(1+2q)$, then as $n \to \infty$, the standardized mean squared error is given by
\[
n^{2q/(1+2q)} \MSE(\hat{\sigma}_n^2) \\
\to \left\{
\begin{array}{ll}
	\frac{\kappa_q^2}{\Psi^{2q}}
	+\frac{\Psi \left(\phi + 1\right) \left(2 q + 1\right)}{q + 1}
	-\frac{8 \Psi \left(\phi^{q + 2} -1\right) \left(2 q + 1\right)}{\left(\phi - 1\right) \left(q + 1\right) \left(q + 2\right) \left(3 q + 2\right)}
	+\frac{\Psi \left(\phi^{2 q + 2} - 1\right)}{\left(\phi - 1\right) \left(q + 1\right) \left(2 q + 1\right)}, & \phi > 1; \\
	\frac{\kappa_q^2}{\Psi^{2q}}
	+\frac{2 \Psi \left(2 q + 1\right)}{q + 1}
	-\frac{8 \Psi \left(2 q + 1\right)}{\left(q + 1\right) \left(3 q + 2\right)}
	+\frac{2 \Psi}{2 q + 1}, & \phi = 1. \\
\end{array}
\right.
\]
Differentiating the standardized mean squared error with respect to $\Psi$, we can derive the minimizer $\Psi_\star$ easily:
\[
\Psi_\star = \left\{
\begin{array}{ll}
	\left\{ \frac{(\phi + 1) (2 q + 1)}{2q (q + 1)}
	- \frac{4 (\phi^{q + 2} -1) (2 q + 1)}{(\phi - 1) q (q + 1) (q + 2) (3 q + 2)}
	+ \frac{(\phi^{2 q + 2} - 1)}{2 (\phi - 1) q (q + 1) (2 q + 1)} \right\}^{-\frac{1}{1+2q}} \kappa_q^{\frac{2}{1+2q}}, & \phi > 1; \\
	\left\{ \frac{2 q + 1}{q (q + 1)}
	- \frac{4 (2 q + 1)}{q (q + 1) (3 q + 2)}
	+ \frac{1}{q(2 q + 1)} \right\}^{-\frac{1}{1+2q}} \kappa_q^{\frac{2}{1+2q}}, & \phi = 1. \\
\end{array}
\right.
\]
When the smoothing parameters are separated, we can allow $\Theta = \rho \Psi_\star$, where $\rho \in \R^+$, to view $\Theta$ as an improvement to the optimal value without separation of parameters.
Consequently, we find $\rho$ that minimize
\[
n^{2q/(1+2q)} \MSE(\hat{\sigma}_n^2) \\
\to \left\{
\begin{array}{ll}
	\frac{\kappa_q^2}{\rho^{2q} \Psi_\star^{2q}}
	+\frac{\Psi_\star \left(\phi + 1\right) \left(2 q + 1\right)}{q + 1}
	-\frac{8 \rho^{-q} \Psi_\star \left(\phi^{q + 2} -1\right) \left(2 q + 1\right)}{\left(\phi - 1\right) \left(q + 1\right) \left(q + 2\right) \left(3 q + 2\right)}
	+\frac{\rho^{-2q} \Psi_\star \left(\phi^{2 q + 2} - 1\right)}{\left(\phi - 1\right) \left(q + 1\right) \left(2 q + 1\right)}, & \phi > 1; \\
	\frac{\kappa_q^2}{\rho^{2q} \Psi_\star^{2q}}
	+\frac{2 \Psi_\star \left(2 q + 1\right)}{q + 1}
	-\frac{8 \rho^{-q} \Psi_\star \left(2 q + 1\right)}{\left(q + 1\right) \left(3 q + 2\right)}
	+\frac{2 \rho^{-2q} \Psi_\star}{2 q + 1}, & \phi = 1. \\
\end{array}
\right.
\]
Differentiating the above with respect to $\rho$, the minimizer is
\[
\rho_\star = \left\{
\begin{array}{ll}
	\left\{ \frac{(q+2) (3q+2) (\phi^{2q+2}-1)}{4 (2q+1)^2 (\phi^{q+2}-1)}
	+ \frac{\Psi_\star^{-2q-1} \kappa_q^2 (\phi-1) (q+1) (q+2) (3q+2)}{4 (2q+1) (\phi^{q+2}-1)} \right\}^{\frac{1}{q}}, & \phi > 1; \\
	\left\{ \frac{(q+1) (3q+2)}{2 (2q+1)^2}
	+ \frac{\Psi_\star^{-2q-1} \kappa_q^2 (q+1) (3q+2)}{4 (2q+1)} \right\}^{\frac{1}{q}}, & \phi = 1. \\
\end{array}
\right.
\]
Note that the $\kappa_q^2$ from $\Psi_\star^{-2q-1}$ will cancel with the $\kappa_q^2$ in $\rho_\star$.
Therefore, the $\mathcal{L}^2$-optimal $\Theta$ is
\[
\Theta_\star = \left\{
\begin{array}{ll}
	\left\{ \frac{(q+2) (3q+2) (\phi^{2q+2}-1)}{4 (2q+1)^2 (\phi^{q+2}-1)}
	+ \frac{\Psi_\star^{-2q-1} \kappa_q^2 (\phi-1) (q+1) (q+2) (3q+2)}{4 (2q+1) (\phi^{q+2}-1)} \right\}^{\frac{1}{q}} \Psi_\star, & \phi > 1; \\
	\left\{ \frac{(q+1) (3q+2)}{2 (2q+1)^2}
	+ \frac{\Psi_\star^{-2q-1} \kappa_q^2 (q+1) (3q+2)}{4 (2q+1)} \right\}^{\frac{1}{q}} \Psi_\star, & \phi = 1. \\
\end{array}
\right.
\]
Using $\Psi_\star$ and $\Theta_\star$, we can check that the standardized mean squared error decreases when $\phi$ decreases, ceteris paribus.
Hence the $\mathcal{L}^2$-optimal $\phi$ is $\phi_\star=1$.
If $O(1)$-space update is required, the $\mathcal{L}^2$-optimal $\phi$ is $\phi_\star=2$ by Proposition \ref{prop:on}.
\hfill$\qedsymbol$

\section{\texorpdfstring{Proof of Proposition \ref{prop:on}}{
		Proof of Proposition 2.1}} \label{sec:proof-on} 

We separate the proof of time and space complexities.
For space complexity, we consider a more general setting that $\hat{\sigma}_n^2$ is any estimator that can be updated in $O(1)$ time.

\subsection{\texorpdfstring{Proof of $O(1)$ Time}{
		Proof of O(1) Time}}  \label{sec:proof-time}

From \cref*{eq:winDecom,eq:genClassOn}, we have
\[
\hat{\sigma}_n^2
= \frac{1}{n} \sum_{i=1}^n (X_i -\bar{X}_n) \left\{\sum_{j=1}^i (2 -\I_{i=j}) T\left( \frac{|i-j|}{t_n(i,j)} \right) S\left(  \frac{|i-j|}{s_n(i,j)} \right) (X_j -\bar{X}_n) \right\}.
\]
Under conditions (L), (A), (S) and (E) in Proposition \ref{prop:on}, $\hat{\sigma}_n^2$ can be further written as
\begin{align} \label{eq:genClassOn-expanded}
	\hat{\sigma}_n^2
	&= \frac{1}{n}\sum_{i=1}^n (X_i -\bar{X}_n) \left\{ \sum_{k=0}^{s_i}(2-\I_{k=0}) \sum_{r=0}^q a_r \left(\frac{k}{t_nt_i'}\right)^r (X_{i-k} -\bar{X}_n) \right\} \nonumber \\
	&= \frac{1}{n} \sum_{r=0}^q \frac{a_r}{t_n^r} \sum_{i=1}^n \frac{X_i}{(t'_i)^r} (K_{i,r} -k_{i,r} \bar{X}_n)
	-\frac{1}{n} \sum_{r=0}^q \frac{a_r}{t_n^r} \sum_{i=1}^n \frac{\bar{X}_n}{(t'_i)^r} (K_{i,r} -k_{i,r} \bar{X}_n),
\end{align}
where $K_{i,r} = \sum_{k=0}^{s_i}(2-\I_{k=0}) k^r X_{i-k}$ and $k_{i,r} = \sum_{k=0}^{s_i}(2-\I_{k=0}) k^r$.
To prove the $O(1)$-time property, it suffices to find a way such that $K_{i,r}$ in \cref{eq:genClassOn-expanded} can be computed in $O(1)$ time as $k_{i,r}$ is a special case.
To do so, note that
\begin{align*}
	K_{i,r} &= \sum_{k=0}^{s_i}(2-\I_{k=0}) k^r X_{i-k}
	= 0^r X_i +2 \sum_{j=i-s_i}^{i-1} (i-j)^r X_j \\
	&= 0^r X_i +2 \sum_{u=0}^r \binom{r}{u} i^{r-u} (-1)^u \sum_{j=i-s_i}^{i-1} j^u X_j \\
	&= 0^r X_i +2 \sum_{u=0}^r \binom{r}{u} i^{r-u} (-1)^u (P_{i-1,u} -P_{i-s_i-1,u}),
\end{align*}
where $P_{i,u} = \sum_{j=1}^i j^u X_j$.
If we store $P_{i,u}$ for $i=1,\ldots,n-1$ and $u=0,\ldots,q$, we can update $P_{n,u}$ by $P_{n-1,u} +n^u X_n$ for $u=0,\ldots,q$, and
consequently compute $K_{n,r}$ in $O(1)$ time, which completes the proof.

However, note that \cref{eq:genClassOn-expanded} requires $O(n)$ space as it has to store all $P_{i,u}$'s.
Furthermore, to support access of $P_{i,u}$ in $O(1)$ time, an array or hash table of unlimited size is necessary.
Otherwise, we may not be able to perform direct access, which is needed for computing $P_{n-1,u} -P_{n-s_n-1,u}$ with unrestricted $s_n$,
or we may need to resize the data structure, which typically costs $O(n)$ time under direct access; see Remark \ref{rmk:memory}.
To perform practical $O(1)$ time update, we recommend taking $t'_i=1$ for all $i \in \Z^+$, and restricting the growth of subsampling parameter such that $\sup_{i \in \Z^+} |s_{i+1}-s_i| \le 1$.
Denote the set of components used to update the $r$-th order term in the polynomial by $\mathcal{C}_{n,r}$. For instance,
\[
\mathcal{C}_{n,r} = \left\{ K_{n,r}, R_{n,r}, k_{n,r}, r_{n,r}, U_{n,r}, V_{n,r}, \left\{ d_{s_n,r}^{(b)}, D_{n,r}^{(b)} \right\}_{b=1,\ldots,r} \right\};
\]
see Algorithm \ref{algo:lase}.
Then, we have the Algorithm \ref{algo:time}.

\begin{algorithm}[!t]
	\caption{Practical $O(1)$-time update} \label{algo:time}
	\SetAlgoVlined
	\DontPrintSemicolon
	\SetNlSty{texttt}{[}{]}
	\small
	\textbf{input}: \;
	(i) $n, s$ -- the sample size and subsampling parameter from the last iteration\;
	(ii) $\mathcal{C}_{n,r}, r=0,1,\ldots,q$ -- the components for the $r$-th order term \;
	(iii) $\vec{x}$ -- the vector storing the last $s$ observations \;
	\Begin{
		Receive $X_{n+1}$ \;
		Set $n = n + 1$ \;
		Compute $s_n, t_n$ \;
		Retrieve $X_{n-1}, \ldots, X_{n-1-\max\{\min(s_n-1,q),0\}}$ and $X_{n-s-1}$ from $\vec{x}$ \;
		\eIf{$s_n == s$}{
			Update $\mathcal{C}_{n,r}$ with Prop. \ref{prop:lase-zm} and \ref{prop:lase} for $r=0,1,\ldots,q$ \;
			Pop $X_{n-s}$ from $\vec{x}$ \;
		}{
			Update $\mathcal{C}_{n,r}$ with Prop. \ref{prop:lase-zm} and \ref{prop:lase} for $r=0,1,\ldots,q$ \;
		}
		Push $X_n$ into $\vec{x}$ \;
		Set $s=s_n$ \;
		Output the online estimate computed with $\mathcal{C}_{n,r}$, where $r=0,1,\ldots,q$ \;
	}
\end{algorithm}

Now, we interpret the conditions in the first part of Proposition \ref{prop:on}.
Recall the backward finite difference operator $\nabla^{(b)} \cdot$ is defined by
\[
\nabla^{(1)} f(k) = f(k) -f(k-1) \quad \text{and} \quad
\nabla^{(b)} f(k) = \nabla^{(b-1)} f(k) -\nabla^{(b-1)} f(k-1),
\]
where $b \in \Z^+$ and $f$ is a function that takes an integer input $k$.
The implications of the polynomial $T(x)$ condition (E) are twofold.
First, the tapering parameter can be exteriorized, which has been briefly discussed in \cref*{sec:principle-driven-sufficient-conditions}.
This is important because the tapering parameter should depend on $n$ according to Principle A.
However, such tapered sum, e.g., $\sum_{k=1}^{s_i} \cos(\pi k/t_n) X_{i-k}$, may not be updated in $O(1)$ time when $t_n$ changes.
Polynomial $T(x)$ is free of this problem as we can work on $\sum_{k=1}^{s_i} T(k) X_{i-k}$ instead of $\sum_{k=1}^{s_i} T(k/t_n) X_{i-k}$.

Second, the backward finite difference method can be used to derive recursive formulas for polynomial $T(x)$.
This is because if $\nabla^q T(k)$ is a constant, we can recursively update $\sum_{k=1}^{s_i} T(k) X_{i-k}$ by rearranging the finite differences:
\begin{align*}
	\nabla^{q-1} T(k) &= \nabla^q T(k)+\nabla^{q-1} T(k-1), \\
	\nabla^{q-2} T(k) &= \nabla^{q-1} T(k)+\nabla^{q-2} T(k-1), \\
	&\hspace{0.4em} \vdots \\
	T(k) &= \nabla T(k) +T(k-1).
\end{align*}
Since the terms on right hand side are either constants, updated by their above formulas, or coming from the last iteration (for those with $T(k-1)$), we can update $\sum_{k=1}^{s_i} T(k) X_{i-k}$ in $O(1)$ time.
A working example is illustrated in \cref{sec:algo-lase}.
In contrast, if $\nabla^q T(k)$ is not a constant for any $q \in \Z^+$, the backward finite difference method will fail.
Consequently, it is uncertain whether $\sum_{k=1}^{s_i} \nabla^q T(k) X_{i-k}$ can be updated in $O(1)$ time.
A typical example is $T(k) = \exp(k)$ as its derivative, which is the continuous analogue of finite difference, is always $\exp(k)$.
Restricting $T(k)$ to be a $q$-th order polynomial avoids this problem as the constancy of $\nabla^q T(k)$ is well-known.

Next, we investigate the separable tapering parameter $t_n(i,j)=t_n t'_{i \lor j}$ condition (A).
Note that $i \lor j$ is taken in the smoothing parameter to make the window symmetric only.
If we do not assume $t'_i=1$ for all $i \in \Z^+$, $O(1)$-time update is still possible in Algorithm \ref{algo:time} by a similar decomposition as in \cref{eq:genClassOn-expanded}.
However, we do not consider it in Algorithm \ref{algo:time} because it may violate Principle A.
In contrast, if the tapering parameter is not separable, e.g., $t_n(i,j)=n+(i \lor j)$, the decomposition in \cref{eq:genClassOn-expanded} becomes
\[
\hat{\sigma}_n^2
= \frac{1}{n} \sum_{r=0}^q a_r \sum_{i=1}^n \frac{X_i}{(n+i)^r} (K_{i,r} -k_{i,r} \bar{X}_n)
-\frac{1}{n} \sum_{r=0}^q a_r \sum_{i=1}^n \frac{\bar{X}_n}{(n+i)^r} (K_{i,r} -k_{i,r} \bar{X}_n),
\]
Consequently, we need to re-compute $X_i K_{i,r}/(n+i)^r$ for $i=1,\ldots,n$ and $r=0,\ldots,q$, which likely cannot be done in $O(1)$ time.

Then, we discuss the subsampling function $S(x)=\I_{x\le 1}$ condition.
It is possible to use other subsampling functions such as $S(x)=\I_{x\le C}$ for some $C \in \R^+$.
Nevertheless, the form $S(x)=\I_{x\le 1}$ is commonly studied in the literature.
It also offers a direct relationship between $s_n(i,j)$ and the subsample size since
\[
\mathcal{S}_n(i)
= \left\{ X_j: S\left( \frac{|i-j|}{s_n(i,j)} \right) = 1, 1 \le j \le i\right\}
= \left\{ X_j: i-s_n(i,j) \le j \le i \right\},
\]
which means that the subsample size is approximately $s_n(i,j)$.
Note that we use the term ``approximately'' because $s_n(i,j)$ may not be an integer.
If we use other $S(x)$ such as $\I_{x<C}$, we will lose this nice interpretation of $s_n(i,j)$.

Finally, for the subsampling parameter condition (L), the importance of local subsample has been briefly elaborated in \cref*{sec:principle-driven-sufficient-conditions}.
In short, it makes online estimation possible by adapting the estimates to their corresponding points in time.
By further restricting the growth of subsampling parameter such that $\sup_{i \in \Z^+} |s_{i+1}-s_i| \le 1$,
Algorithm \ref{algo:time} only requires $O(s_n)$ space in contrast to the $O(n)$ space requirement of \cref{eq:genClassOn-expanded}.
In addition, Algorithm \ref{algo:time} can be implemented using a queue data structure, which supports insertions and the necessary access in $O(1)$ time.
We remark that the condition $\sup_{i \in \Z^+} |s_{i+1} -s_i| \le 1$ is sufficient throughout \cref{sec:algorithms} because the subsample size grows sublinearly and cannot exceed the sample size in practice; see Definition \ref{def:winLASER}.

\begin{remark} \label{rmk:memory}
	When an online algorithm exhausts the initial memory, its theoretically $O(1)$-time operations may need more than $O(1)$ time to deal with memory reallocation.
	Therefore, the choice of data structure matters in implementation, and $O(1)$-space update is desirable under memory constraint.
\end{remark}

\begin{remark} \label{rmk:time}
	We believe that the conditions in the first part of Proposition \ref{prop:on} are close to be necessary given their implications.
	In particular, it is clear that online estimation is closely related to the backward finite difference method.
	However, it will be difficult to prove that an online algorithm does not exist without certain conditions in general.
\end{remark}

\subsection{\texorpdfstring{Proof of $O(1)$ Space}{
		Proof of O(1) Space}}  \label{sec:proof-space}

To prove that the conditions are sufficient for having $O(1)$-space update, we consider the case that the intended subsampling parameter $s_i$ is monotonically increasing.
Denote the set of intended components and the set of precalculated components under ramping by $\mathcal{C}'_n$ and $\mathcal{C}''_n$, respectively.
In the more general setting that $\hat{\sigma}_n^2$ can be updated in $O(1)$ time, recursive formulas for updating $\mathcal{C}'_n$ and $\mathcal{C}''_n$ exist.
Then, the proof is completed by Algorithm \ref{algo:space}.
For the case that $s_i$ is fixed or prespecified, the algorithm is similar as the ramping upper bound $a_i$ is also prespecified.

\begin{algorithm}[!t]
	\caption{$O(1)$-space update via precalculation} \label{algo:space}
	\SetAlgoVlined
	\DontPrintSemicolon
	\SetNlSty{texttt}{[}{]}
	\small
	\textbf{input}: \;
	(i) $n, s$ -- the sample size and subsampling parameter from the last iteration\;
	(ii) $\mathcal{C}'_n, \mathcal{C}''_n$ -- the intended components and their precalculated versions \;
	(iii) $s', a$ -- the ramped subsampling parameter and its upper bound from the last iteration \;
	\Begin{
		Receive $X_{n+1}$ \;
		Set $n = n + 1$ \;
		Compute $s_n, t_n, a_n$ \;
		\eIf{$s'+1 < a$}{
			Set $s'_n = s' + 1, s_n = s, a_n = a$ \tcc*[f]{restrict $s_n$ when $s'_n>s'$} \;
			Update $\mathcal{C}'_n$ with their recursive formulas \;
			\eIf{$s'_n \ge a - s$}{
				Update $\mathcal{C}''_n$ with their recursive formulas \;
			}{
				Re-initialize $\mathcal{C}''_n$ with their recursive formulas \;
			}
		}{
			Set $s'_n = s_n, a = a_n$ \;
			Replace the applicable components in $\mathcal{C}'_n$ with $\mathcal{C}''_n$ \;
			Update $\mathcal{C}'_n$ with their recursive formulas \;
			Re-initialize $\mathcal{C}''_n$ with their recursive formulas \;
		}
		Set $s'=s'_n, s=s_n$ \;
		Output the recursive estimate computed with $\mathcal{C}'_n$ \;
	}
\end{algorithm}

A practical example of $\mathcal{C}'_n$ can be found in Algorithm \ref{algo:laser}:
\[
\mathcal{C}'_n = \left\{ Q_n, \bar{X}_n, \left\{ K'_{n,b}, R'_{n,b}, k'_{n,b}, r'_{n,b}, U'_{n,b}, V'_{n,b} \right\}_{b=0,1} \right\}.
\]
We further discuss the implications of the conditions in Proposition \ref{prop:on}.
Conditions (L), (A), (S) and (E) in the first part ensure the existence of recursive formulas for updating $\mathcal{C}'_n$ and $\mathcal{C}''_n$.
Consequently, we can use the precalculation technique to perform $O(1)$-time and
$O(1)$-space update.
A working example is illustrated in \cref{sec:algo-laser}.

For the ramped subsampling parameter condition (R), it ensures the correctness of the precalculation technique in the proof.
This is because we need to re-initialize $\mathcal{C}''_n$ to prepare for the next reset of $s'_n$.
If re-initialization cannot be done, we need to drop far observations in $\mathcal{C}''_n$, which ruins the $O(1)$-space update.
To re-initialize properly, we need $s'_n \le a -s$ for some $s'_n \in [s,a)$, which is guaranteed by the $\phi \ge 2$ condition (R).
If the intended subsampling parameter $s_i$ is not monotonically increasing but prespecified, we can still precalculate the components by computing $s_n$ at the next point of reset in advance.
\hfill$\qedsymbol$

\section{Proof of Lemmas} \label{sec:proof-of-lemmas}

The technical lemmas used in other sections are stated and proved here.

\subsection{\texorpdfstring{Lemma \ref{lem:eqvMean}}{
		Lemma F.1}}\label{proof-eqvMean} 

The following lemma shows that the unknown-mean version \hyperref[eq:sigmaClass]{$\hat{\sigma}^2_n$} and the known-mean version \hyperref[eq:sigmaMean]{$\bar{\sigma}^2_n$} are asymptotically equivalent in the $\mathcal{L}^{\alpha/2}$-sense.

\begin{lemma} \label{lem:eqvMean}
	Let $\alpha > 2$.
	Suppose that $\norm{X_1}_\alpha < \infty$.
	If Assumption \ref{asum:stability} holds, then
	\[
	\norm{ \hat{\sigma}^2_n - \bar{\sigma}^2_n }_{\frac{\alpha}{2}}
	= O(n^{-1}) +O\left( n^{-\frac{3}{2}} G_{3,n}^{\frac{1}{2}} \right) +O\left( n^{-2} G_{4,n} \right),
	\]
	where
	\[
	G_{3,n} = \sum_{i=2}^n \left| \sum_{j=1}^{i-1} W_n(i,j) \right|^2
	+ \sum_{j=1}^{n-1} \left| \sum_{i=j+1}^n W_n(i,j) \right|^2 \quad \text{and} \quad
	G_{4,n} = \left| \sum_{i=2}^n \sum_{j=1}^{i-1} W_n(i,j) \right|.
	\]
	In addition, if Assumptions \ref{asum:winSums} and \ref{asum:winGen}(a) also hold, then
	$
	\norm{ \hat{\sigma}^2_n - \bar{\sigma}^2_n }_{\alpha/2} = o(1).
	$
\end{lemma}

\begin{proof}
	By the Minkowski inequality and Hölder's inequality,
	\begin{align}
		\begin{split}
			& \norm{ \hat{\sigma}^2_n - \bar{\sigma}^2_n }_{\frac{\alpha}{2}} \\
			\le{}& \norm{ \frac{1}{n} \sum_{i=1}^n (2X_i -\bar{X}_n -\mu)(\mu -\bar{X}_n) }_{\frac{\alpha}{2}}
			+\norm{ \frac{2}{n} \sum_{i=2}^n \sum_{j=1}^{i-1} W_n(i,j) \Big\{ \bar{X}_n^2 -(\bar{X}_n-\mu)(X_i+X_j) -\mu^2 \Big\} }_{\frac{\alpha}{2}} \\
			\le{}& \norm{ (\bar{X}_n -\mu)(\mu -\bar{X}_n) }_{\frac{\alpha}{2}}
			+\norm{ \frac{2}{n} \sum_{i=2}^n \sum_{j=1}^{i-1} W_n(i,j) \Big\{ (\bar{X}_n-\mu)^2 -(\bar{X}_n-\mu)(X_i+X_j -2\mu) \Big\} }_{\frac{\alpha}{2}} \nonumber
		\end{split} \\
		\begin{split}
			\le{}& \norm{ \bar{X}_n -\mu }_\alpha^2 + \frac{2}{n} \left| \sum_{i=2}^n \sum_{j=1}^{i-1} W_n(i,j) \right| \norm{ \bar{X}_n -\mu }_{2\alpha}^2 \\
			& +\frac{2}{n} \norm{ \sum_{i=2}^n \sum_{j=1}^{i-1} W_n(i,j) (X_i+X_j-2\mu) }_\alpha \norm{ \bar{X}_n -\mu }_\alpha. \label{eq:interMean}
		\end{split}
	\end{align}
	By the moment inequality in Theorem 3 of \citet{wu2011asymptotic} and Assumption \ref{asum:stability}, $\norm{ \bar{X}_n -\mu }_\alpha = O(n^{-1/2})$.
	For the last term in \cref{eq:interMean}, applying the Minkowski inequality gives
	\begin{align*}
		\norm{ \sum_{i=2}^n \sum_{j=1}^{i-1} W_n(i,j) (X_i+X_j-2\mu) }_\alpha
		\le{}& \norm{ \sum_{i=2}^n \sum_{j=1}^{i-1} W_n(i,j) (X_i-\mu) }_\alpha \\
		& +\norm{ \sum_{j=1}^{n-1} \sum_{i=j+1}^{n} W_n(i,j) (X_j-\mu) }_\alpha.
	\end{align*}
	By Lemma 1 of \citet{liu2010asymptotic},
	\begin{align*}
		\norm{ \sum_{i=2}^n \sum_{j=1}^{i-1} W_n(i,j) (X_i-\mu) }_\alpha
		&\le C_\alpha \Delta_\alpha \sqrt{ \sum_{i=2}^n \left| \sum_{j=1}^{i-1} W_n(i,j) \right|^2 }
		\le C_\alpha \Delta_\alpha G_{3,n}^{\frac{1}{2}}, \\
		\norm{ \sum_{j=1}^{n-1} \sum_{i=j+1}^{n} W_n(i,j) (X_j-\mu) }_\alpha
		&\le C_\alpha \Delta_\alpha \sqrt{ \sum_{j=1}^{n-1} \left| \sum_{i=j+1}^{n} W_n(i,j) \right|^2 }
		\le C_\alpha \Delta_\alpha G_{3,n}^{\frac{1}{2}},
	\end{align*}
	where $C_\alpha \in \R^+$ is some constant that only depends on $\alpha$ and may change from occurrence to occurrence in the proofs.
	Hence under Assumption \ref{asum:stability},
	\[
	\norm{ \sum_{i=2}^n \sum_{j=1}^{i-1} W_n(i,j) (X_i+X_j-2\mu) }_\alpha
	= O\Big( G_{3,n}^{\frac{1}{2}} \Big).
	\]
	Combining the results, we have
	\begin{align*}
		\norm{ \hat{\sigma}^2_n - \bar{\sigma}^2_n }_{\frac{\alpha}{2}}
		&= O(n^{-1}) + n^{-1} G_{4,n} O(n^{-1}) +n^{-1} O\Big( G_{3,n}^{\frac{1}{2}} \Big) O\Big( n^{-\frac{1}{2}} \Big)\\
		&= O(n^{-1}) +O\left( n^{-\frac{3}{2}} G_{3,n}^{\frac{1}{2}} \right) +O\left( n^{-2} G_{4,n} \right).
	\end{align*}
	In addition, if Assumptions \ref{asum:winSums} and \ref{asum:winGen}(a) also hold, the Minkowski inequality gives
	\begin{align*}
		G_{3,n}
		&\le 2 \sum_{i=1}^n \left( \sum_{j=1}^n \big| W_n(i,j) \big| \right)^2
		\le 2 \sum_{i=1}^n G_{1,n}^2
		= O(n) o\left( n^{4-\frac{4}{\alpha'}} \right) = o(n^3), \\
		G_{4,n} &\le \sum_{i=1}^n \sum_{j=1}^n \big| W_n(i,j) \big|
		\le \sum_{i=1}^n G_{1,n}
		= O(n) o\left( n^{2-\frac{2}{\alpha'}} \right) = o(n^2),
	\end{align*}
	since $\alpha' \in (1,2]$. Therefore, $\norm{ \hat{\sigma}^2_n - \bar{\sigma}^2_n }_{\alpha/2} = o(1)$.
\end{proof}

\subsection{\texorpdfstring{Lemma \ref{lem:eqvMDepend}}{
		Lemma F.2}}\label{proof-eqvMDepend} 

Recall that we can assume $\mu = 0$ after applying Lemma \ref{lem:eqvMean}. Define
\begin{align}
	\Delta_{m,\alpha} &= \sum_{i=m}^\infty \delta_{i,\alpha}, \label{eq:deltaMDepend} \\
	\Upsilon_{m,\alpha} &= \sqrt{ \sum_{i=m}^\infty \delta_{i, \alpha}^2 }, \label{eq:upsilonMDepend} \\
	d_{m,\alpha} &= \sum_{i=0}^\infty \min ( \delta_{i, \alpha}, \Upsilon_{m+1, \alpha} ), \label{eq:dMDepend} \\
	\mathcal{P}_i \cdot &= \E(\cdot \mid \mathcal{F}_i) -\E(\cdot \mid \mathcal{F}_{i-1}). \label{eq:project}
\end{align}
For simplicity, write $\Delta_{0,\alpha} = \Delta_{\alpha}$.
The following lemma shows that the known-mean version \hyperref[eq:sigmaMean]{$\bar{\sigma}^2_n$} and the $m$-dependent process approximated version \hyperref[eq:sigmaMDepend]{$\tilde{\sigma}^2_n$} are asymptotically equivalent in the $\mathcal{L}^{\alpha/2}$-sense.

\begin{lemma} \label{lem:eqvMDepend}
	Let $\alpha > 2$ and $\alpha' = \min(\alpha/2 ,2)$.
	Suppose that $\norm{X_1}_\alpha < \infty$.
	If Assumption \ref{asum:stability} holds, then
	\[
	\norm{ \bar{\sigma}^2_n -\E(\bar{\sigma}^2_n) -\tilde{\sigma}^2_n +\E(\tilde{\sigma}^2_n) }_{\frac{\alpha}{2}}
	= O\left( n^{\frac{1}{\alpha'}-1} d_{m, \alpha} G_{5,n}^{\frac{1}{2}} \right),
	\]
	where
	\[
	G_{5,n} = \max_{1 \le i \le n} \sum_{j=1}^n \Big\{ W_n^2(i,j) +W_n^2(j,i) \Big\}.
	\]
	In addition, if Assumptions \ref{asum:winSums} and \ref{asum:winGen}(a) also hold, then
	$
	\norm{ \bar{\sigma}^2_n -\E(\bar{\sigma}^2_n) -\tilde{\sigma}^2_n +\E(\tilde{\sigma}^2_n) }_{\alpha/2} = o(1).
	$
\end{lemma}

\begin{remark} \label{rmk:eqvMDepend}
	The proof of Lemma \ref{lem:eqvMDepend} follows a similar procedure as the proof of Proposition 1 in \citet{liu2010asymptotic}.
	However, their Proposition 1 cannot be applied directly since our window functions may depend on the running indices $i,j$ in addition to the difference $i-j$.
	In addition, their Proposition 1 does not provide explicit approximation of $\sum_{i=1}^n X_i^2$, which is included in our estimators.
	Finally, their Proposition 1 does not consider the case of $2 < \alpha < 4$, which we will show later for proving consistency of our estimators.
\end{remark}

\begin{proof}
	Recall that the $m$-dependent process $\tilde{X}_j$ is defined in \cref{eq:xMDepend}.
	Let $Z_i = 2 \sum_{j=1}^{i-1} W_n(i,j) X_j$, and $\tilde{Z}_i = 2 \sum_{j=1}^{i-1} W_n(i,j) \tilde{X}_j$.
	Define the intermediate version
	\[
	\tilde{\sigma}^*_n
	= \frac{1}{n} \sum_{i=1}^n X_i \tilde{X}_i +\frac{2}{n} \sum_{i=2}^{n} X_i \sum_{j=1}^{i-1} W_n(i,j) \tilde{X}_j
	= \frac{1}{n} \sum_{i=1}^n X_i \tilde{X}_i +\frac{1}{n} \sum_{i=2}^{n} X_i \tilde{Z}_i.
	\]
	In addition, recall that $X_{i,\{h\}} = g(\mathcal{F}_{i,\{h\}})$ represents a coupled version of $X_i = g(\mathcal{F}_i)$ with $\epsilon_h$ replaced by a copy $\epsilon'_h$ in $\mathcal{F}_i$.
	We similarly define $Z_{i,\{h\}}$ and $\tilde{Z}_{i,\{h\}}$ as the coupled versions of $Z_i$ and $\tilde{Z}_i$, respectively.
	Note that for any $X_i \in \sigma(\mathcal{F}_i)$,
	$
	\norm{ \mathcal{P}_h X_i }_\alpha \le \norm{ X_i -X_{i,\{h\}}}_\alpha
	$
	\citep{wu2005asymptotic}.
	Therefore, by the Minkowski inequality,
	\begin{align}
		\begin{split}
			& n \norm{ \mathcal{P}_h (\bar{\sigma}^2_n -\tilde{\sigma}^*_n) }_{\frac{\alpha}{2}} \\
			\le{}& \Bigg\lVert \sum_{i=1}^n \Big\{ X_i (X_i -\tilde{X}_i) -X_{i,\{h\}} (X_{i,\{h\}} -\tilde{X}_{i,\{h\}}) \Big\} \\
			& +\sum_{i=2}^{n} \Big\{ X_i (Z_i -\tilde{Z}_i) -X_{i,\{h\}} (Z_{i,\{h\}} -\tilde{Z}_{i,\{h\}}) \Big\} \Bigg\rVert_{\frac{\alpha}{2}} \\
			\le{}& \sum_{i=1}^n \norm{ X_{i,\{h\}} \Big\{ (X_i -\tilde{X}_i) - (X_{i,\{h\}} -\tilde{X}_{i,\{h\}}) \Big\} }_{\frac{\alpha}{2}}
			+\sum_{i=1}^n \norm{ (X_i -X_{i,\{h\}}) (X_i -\tilde{X}_i) }_{\frac{\alpha}{2}} \\
			& +\norm{ \sum_{i=2}^{n} X_{i,\{h\}} \Big\{ (Z_i -\tilde{Z}_i) - (Z_{i,\{h\}} -\tilde{Z}_{i,\{h\}}) \Big\} }_{\frac{\alpha}{2}}
			+\sum_{i=2}^{n} \norm{ (X_i -X_{i,\{h\}}) (Z_i -\tilde{Z}_i) }_{\frac{\alpha}{2}} \nonumber
		\end{split} \\
		={}& I_1 +I_2 +I_3 +I_4. \label{eq:interMDepend}
	\end{align}
	By the Jensen's inequality,
	\[
	\norm{ \tilde{X}_i -\tilde{X}_{i,\{0\}} }_\alpha
	= \norm{ \E(X_i -X_{i,\{0\}} \mid \epsilon_{i-m}, \ldots, \epsilon_i, \epsilon'_0) }_\alpha \\
	\le \norm{ X_i -X_{i,\{0\}} }_\alpha
	= \delta_{i, \alpha}.
	\]
	For $\alpha > 2$, the Burkholder inequality gives (see the proof of Lemma 1 in \citealt{liu2010asymptotic})
	\begin{align*}
		\norm{ X_i -\tilde{X}_i }_\alpha^2
		&\le C_\alpha \sum_{l=m+1}^\infty \norm{ \E(X_i \mid \epsilon_{i-l}, \ldots, \epsilon_i) -\E(X_i \mid \epsilon_{i-l+1}, \ldots, \epsilon_i) }_\alpha^2 \\
		&\le C_\alpha \sum_{l=m+1}^\infty \delta_{l, \alpha}^2
		= C_\alpha \Upsilon_{m+1,\alpha}^2.
	\end{align*}
	Combining the results yields
	\begin{equation} \label{eq:boundMDepend}
		\norm{ X_i -\tilde{X}_i -X_{i,\{h\}} +\tilde{X}_{i,\{h\}} }_\alpha
		\le C_\alpha \min ( \delta_{i-h, \alpha}, \Upsilon_{m+1, \alpha} ).
	\end{equation}
	Now we begin to derive upper bounds for $I_1$ to $I_4$.
	By the Hölder's inequality, \cref{eq:boundMDepend} and stationarity,
	\begin{align*}
		I_1 &\le \sum_{i=1}^n \norm{ X_{i,\{h\}} }_\alpha
		\norm {X_i -\tilde{X}_i -X_{i,\{h\}} +\tilde{X}_{i,\{h\}} }_\alpha \\
		&\le C_\alpha \max_{1 \le i \le n} \norm{ X_{i,\{h\}} }_\alpha \sum_{i=1}^n \min ( \delta_{i-h, \alpha}, \Upsilon_{m+1, \alpha} ) \\
		&\le C_\alpha \norm{ X_1 }_\alpha \sum_{i=1}^n \min ( \delta_{i-h, \alpha}, \Upsilon_{m+1, \alpha} ).
	\end{align*}
	Since $\norm{ X_1 }_\alpha < \infty$, $\Delta_\alpha < \infty$ under Assumption \ref{asum:stability}, and $\delta_{i, \alpha} = 0$ for $i < 0$, we have
	\begin{align*}
		\sum_{h=-\infty}^n I_1^{\alpha'} 
		&\le C_\alpha \norm{ X_1 }_\alpha^{\alpha'} \sum_{h=-\infty}^n \Delta_\alpha^{\alpha'-1} \sum_{i=1}^n \min ( \delta_{i-h, \alpha}, \Upsilon_{m+1, \alpha} ) \\
		&= C_\alpha \norm{ X_1 }_\alpha^{\alpha'} \Delta_\alpha^{\alpha'-1} \sum_{i=1}^n \sum_{h=-\infty}^n \min ( \delta_{i-h, \alpha}, \Upsilon_{m+1, \alpha} )
		= O(n d_{m,\alpha}).
	\end{align*}
	For the second term in \cref{eq:interMDepend}, by the Hölder's inequality we have
	\[
	I_2
	\le \sum_{i=1}^n \norm{ X_i -X_{i,\{h\}} }_\alpha \norm{ X_i -\tilde{X}_i }_\alpha
	\le \max_{1 \le i \le n} \norm{ X_i -\tilde{X}_i }_\alpha \sum_{i=1}^n \delta_{i-h, \alpha}
	\le C_\alpha \Upsilon_{m+1, \alpha} \sum_{i=1}^n \delta_{i-h, \alpha}.
	\]
	Since
	$
	\Upsilon_{m+1,\alpha}
	= \sqrt{ \sum_{i=m+1}^\infty \delta_{i, \alpha}^2 }
	\le \sqrt{ \sum_{i=1}^\infty \delta_{i, \alpha}^2 }
	\le \sqrt{ \left| \sum_{i=1}^\infty \delta_{i, \alpha} \right|^2 }
	= \Delta_\alpha,
	$
	we have
	\[
	\sum_{h=-\infty}^n I_2^{\alpha'}
	\le C_\alpha \Delta_\alpha^{\alpha'} \Delta_\alpha^{\alpha'-1} \sum_{i=1}^n \sum_{h=-\infty}^n \delta_{i-h, \alpha}
	= O(n).
	\]
	For the third term in \cref{eq:interMDepend}, we first rearrange the order of summation:
	\begin{align*}
		I_3 &= \norm{ 2 \sum_{i=2}^{n} X_{i,\{h\}} \sum_{j=1}^{i-1} W_n(i,j) ( X_j -\tilde{X}_j -X_{j,\{h\}} +\tilde{X}_{j,\{h\}} ) }_{\frac{\alpha}{2}} \\
		&= \norm{ 2 \sum_{j=1}^{n-1} ( X_j -\tilde{X}_j -X_{j,\{h\}} +\tilde{X}_{j,\{h\}} ) \sum_{i=j+1}^n W_n(i,j) X_{i,\{h\}} }_{\frac{\alpha}{2}}.
	\end{align*}
	By the Minkowski inequality, Hölder's inequality, Lemma 1 of \citet{liu2010asymptotic} and \cref{eq:boundMDepend},
	\begin{align*}
		I_3 &\le 2 \sum_{j=1}^{n-1} \norm{ X_j -\tilde{X}_j -X_{j,\{h\}} +\tilde{X}_{j,\{h\}} }_\alpha \norm{ \sum_{i=j+1}^n W_n(i,j) X_{i,\{h\}} }_\alpha \\
		&\le C_\alpha \Delta_\alpha \sum_{j=1}^n \min ( \delta_{j-h, \alpha}, \Upsilon_{m+1, \alpha} ) \sqrt{ \sum_{i=j+1}^n W_n^2(i,j) } \\
		&\le C_\alpha \Delta_\alpha \max_{1 \le j \le n} \sqrt{ \sum_{i=j+1}^n W_n^2(i,j) } \sum_{j=1}^n \min ( \delta_{j-h, \alpha}, \Upsilon_{m+1, \alpha} ) \\
		&\le C_\alpha \Delta_\alpha G_{5,n}^{\frac{1}{2}} \sum_{j=1}^n \min ( \delta_{j-h, \alpha}, \Upsilon_{m+1, \alpha} ).
	\end{align*}
	Therefore, we have
	\[
	\sum_{h=-\infty}^n I_3^{\alpha'}
	\le C_\alpha \Delta_\alpha^{\alpha'} G_{5,n}^{\frac{\alpha'}{2}} \Delta_\alpha^{\alpha'-1} \sum_{j=1}^n \sum_{h=-\infty}^n \min ( \delta_{j-h, \alpha}, \Upsilon_{m+1, \alpha} )
	= O\left( n d_{m,\alpha} G_{5,n}^{\frac{\alpha'}{2}} \right).
	\]
	For the last term in \cref{eq:interMDepend}, note that by Lemma 1 of \citet{liu2010asymptotic},
	\[
	\norm{ Z_i -\tilde{Z}_i }_\alpha
	= \norm{ 2 \sum_{j=1}^{i-1} W_n(i,j) (X_j -\tilde{X}_j) }_\alpha
	\le C_\alpha \Delta_{m+1,\alpha} \sqrt{ \sum_{j=1}^{i-1} W_n^2(i,j) }
	\le C_\alpha \Delta_{m+1,\alpha} G_{5,n}^{\frac{1}{2}}.
	\]
	Thus, by the Hölder's inequality,
	\[
	I_4
	\le \sum_{i=1}^n \norm{ X_i -X_{i,\{h\}} }_\alpha \norm{ Z_i -\tilde{Z}_i }_\alpha \\
	\le C_\alpha \Delta_{m+1,\alpha} G_{5,n}^{\frac{1}{2}} \sum_{i=1}^n \delta_{i-h, \alpha}.
	\]
	Check that $\Delta_{m+1,\alpha} \le d_{m,\alpha}$ because
	\[
	\sum_{i=0}^\infty \Upsilon_{m+1,\alpha} -\Delta_{m+1,\alpha}
	= \sum_{i=0}^\infty \sqrt{ \sum_{i=m+1}^\infty \delta_{i, \alpha}^2 } -\sum_{i=m+1}^\infty \sqrt{ \delta_{i,\alpha}^2 }
	\ge 0.
	\]
	Therefore, we have
	\[
	\sum_{h=-\infty}^n I_4^{\alpha'}
	\le C_\alpha d_{m,\alpha}^{\alpha'} G_{5,n}^{\frac{\alpha'}{2}} \Delta_\alpha^{\alpha'-1} \sum_{i=1}^n \sum_{h=-\infty}^n \delta_{i-h, \alpha}
	= O\left( n d_{m,\alpha}^{\alpha'} G_{5,n}^{\frac{\alpha'}{2}} \right).
	\]
	Since $\alpha' \in (1,2]$, combining the results leads to
	\[
	\sum_{h=-\infty}^n \norm{ \mathcal{P}_h (\bar{\sigma}^2_n -\tilde{\sigma}^*_n) }_{\frac{\alpha}{2}}^{\alpha'}
	= O\left( n^{1-\alpha'} d_{m, \alpha}^{\alpha'} G_{5,n}^{\frac{\alpha'}{2}} \right).
	\]
	Finally, by the decomposition $X_n -\E(X_n) = \sum_{h=-\infty}^n \mathcal{P}_h X_n$ for stationary $\{X_i\}$ and the Burkholder inequality,
	\begin{align*}
		\norm{ \bar{\sigma}^2_n -\E(\bar{\sigma}^2_n) -\tilde{\sigma}^*_n +\E(\tilde{\sigma}^*_n) }_{\frac{\alpha}{2}}^{\alpha'}
		&\le C_\alpha \sum_{h=-\infty}^n \norm{ \mathcal{P}_h (\bar{\sigma}^2_n -\tilde{\sigma}^*_n) }_{\frac{\alpha}{2}}^{\alpha'} \\
		&= O\left( n^{1-\alpha'} d_{m, \alpha}^{\alpha'} G_{5,n}^{\frac{\alpha'}{2}} \right), \\
		\norm{ \bar{\sigma}^2_n -\E(\bar{\sigma}^2_n) -\tilde{\sigma}^*_n +\E(\tilde{\sigma}^*_n) }_{\frac{\alpha}{2}}
		&= O\left( n^{\frac{1}{\alpha'}-1} d_{m, \alpha} G_{5,n}^{\frac{1}{2}} \right).
	\end{align*}
	A similar inequality for $\norm{ \tilde{\sigma}^*_n -\E(\tilde{\sigma}^*_n) -\tilde{\sigma}^2_n +\E(\tilde{\sigma}^2_n) }_{\alpha/2}$ completes the proof of the general case.
	In addition, note that $d_{m, \alpha} = o(1)$ as $m \to \infty$ in view of Lemma \ref{lem:dMDepend}.
	To fulfill this condition, we can choose $m$ such that $m \to \infty$ as $n \to \infty$ by Assumption \ref{asum:winSums}; see \cref{sec:proof-eqvMart}.
	It also follows from Assumptions \ref{asum:winSums} and \ref{asum:winGen}(a) that
	$
	G_{5,n} = 2 \max_{1 \le i \le n} \sum_{j=1}^n W_n^2(i,j)
	\le 2 G_{1,n}
	= o(n^{2-2/\alpha'}).
	$
	Therefore, $\norm{ \bar{\sigma}^2_n -\E(\bar{\sigma}^2_n) -\tilde{\sigma}^*_n +\E(\tilde{\sigma}^*_n) }_{\alpha/2} = o(1)$.
	Finding a similar inequality for $\norm{ \tilde{\sigma}^*_n -\E(\tilde{\sigma}^*_n) -\tilde{\sigma}^2_n +\E(\tilde{\sigma}^2_n) }_{\alpha/2}$ completes the proof.
\end{proof}

\subsection{\texorpdfstring{Lemma \ref{lem:eqvMart}}{
		Lemma F.3}} \label{sec:proof-eqvMart} 

The following lemma shows that the $m$-dependent process approximated version \hyperref[eq:sigmaMDepend]{$\tilde{\sigma}^2_n$} and the martingale difference approximated version \hyperref[eq:sigmaMart]{$\breve{\sigma}^2_n$} are asymptotically equivalent in the $\mathcal{L}^{\alpha/2}$-sense.

\begin{lemma} \label{lem:eqvMart}
	Let $\alpha > 2$ and $\alpha' = \min(\alpha/2 ,2)$.
	Suppose that $\E|X_1|^\alpha < \infty$.
	If Assumption \ref{asum:stability} holds, then
	\[
	\norm{ \tilde{\sigma}^2_n -\E(\tilde{\sigma}^2_n) -\breve{\sigma}^2_n +\E(\breve{\sigma}^2_n) }_{\frac{\alpha}{2}}
	= O\left( m^{3-\frac{1}{\alpha'}} n^{\frac{1}{\alpha'}-1} G_{6,n} \right),
	\]
	where
	\begin{align*}
		G_{6,n} &= \max_{1 \le i,j \le n} \big| W_n(i,j) \big|
		+\max_{1 \le i \le n} \left[ \sum_{j=2}^n \Big\{ \big| \nabla_1 W_n(j,i) \big|^{\alpha'} +\big| \nabla_2 W_n(i,j) \big|^{\alpha'} \Big\} \right]^{\frac{1}{\alpha'}}, \\
		\nabla_1 W_n(j,i) &= W_n(j,i) -W_n(j-1,i) \quad \text{and} \quad
		\nabla_2 W_n(i,j) = W_n(i,j) -W_n(i,j-1).
	\end{align*}
	In addition, if Assumptions \ref{asum:winSums} and \ref{asum:winGen}(a) also hold, then
	$
	\norm{ \tilde{\sigma}^2_n -\E(\tilde{\sigma}^2_n) -\breve{\sigma}^2_n +\E(\breve{\sigma}^2_n) }_{\alpha/2} = o(1).
	$
\end{lemma}

\begin{remark} \label{rmk:eqvMart}
	The proof of Lemma \ref{lem:eqvMart} follows a similar procedure as the proof of Proposition 2 in \citet{liu2010asymptotic}.
	However, their Proposition 2 cannot be applied directly for similar reasons stated in Remark \ref{rmk:eqvMDepend}.
	It is also worth noting that $W_n(i,j)$ is implicitly assumed to be real throughout our paper.
	Therefore, we do not consider the case of complex coefficients in contrast to Proposition 2 of \citet{liu2010asymptotic}.
\end{remark}

\begin{proof}
	We first show some simple results that will be used later in this proof.
	They were similarly established in the proof of Proposition 2 in \citet{liu2010asymptotic} but we try to give more details here.
	Recall that $\tilde{X}_i$, $M_i$ and $D_i$ are defined in \cref{eq:xMDepend,eq:aMart,eq:dMart}, respectively.
	By the construction of $\tilde{X}_i$, we can rewrite $M_i$ by
	$
	M_i = \sum_{h=0}^m \E(\tilde{X}_{i+h} \mid \mathcal{F}_i)
	$
	since $\E(\tilde{X}_{i+h} \mid \mathcal{F}_i) = 0$ when $h > m$.
	In addition, $D_i$'s are $m$-dependent and form martingale differences with respect to $\mathcal{F}_i$.
	By the Minkowski inequality, Jensen's inequality and stationarity of $X_i$, we have
	\begin{equation} \label{eq:boundAMart}
		\norm{ M_i }_\alpha
		= \norm{ \sum_{h=0}^m \E(\tilde{X}_{i+h} \mid \mathcal{F}_i) }_\alpha
		\le 2m \norm{ X_1 }_\alpha.
	\end{equation}
	We also have
	\begin{align*}
		M_i -\E(M_{i+1} \mid \mathcal{F}_i)
		&= \sum_{h=0}^m \E(\tilde{X}_{i+h} \mid \mathcal{F}_i) -\sum_{h=0}^m \E(\tilde{X}_{i+h+1} \mid \mathcal{F}_i) \\
		&= \E(\tilde{X}_i \mid \mathcal{F}_i) = \tilde{X}_i.
	\end{align*}
	As a result,
	\begin{align}
		\tilde{X}_i -D_i
		&= M_i -\E(M_{i+1} \mid \mathcal{F}_i) -M_i +\E(M_i \mid \mathcal{F}_{i-1}) \nonumber \\
		&= \E(M_i \mid \mathcal{F}_{i-1}) -\E(M_{i+1} \mid \mathcal{F}_i). \label{eq:diffMart}
	\end{align}
	Now, we begin the main part of this proof. Let
	\[
	\begin{array}{ll}
		\displaystyle Y_{1,i} = \tilde{X}_i \sum_{j=1}^{i-8m} W_n(i,j) (\tilde{X}_j -D_j),
		& Y_{2,i} = \tilde{X}_i \sum_{j=i-8m+1}^{i-1} W_n(i,j) (\tilde{X}_j -D_j), \\
		\displaystyle Y_i = 2\tilde{X}_i \sum_{j=1}^{i-1} W_n(i,j) (\tilde{X}_j -D_j) = 2 (Y_{1,i} +Y_{2,i}),
		& Y'_i = 2 (\tilde{X}_i -D_i) \sum_{j=1}^{i-1} W_n(i,j) D_j.
	\end{array}
	\]
	By the Minkowski inequality,
	\begin{align}
		\begin{split}
			n \norm{ \tilde{\sigma}^2_n -\E(\tilde{\sigma}^2_n) -\breve{\sigma}^2_n +\E(\breve{\sigma}^2_n) }_{\frac{\alpha}{2}}
			\le{}& \norm{ \sum_{i=1}^n \Big\{ \tilde{X}_i^2 -\E (\tilde{X}_i^2) \Big\} }_{\frac{\alpha}{2}} + \norm{ \sum_{i=1}^n \Big\{ D_i^2 -\E (D_i^2) \Big\} }_{\frac{\alpha}{2}} \\
			& +\norm{ \sum_{i=1}^n \Big\{ Y_i -\E(Y_i) \Big\} }_{\frac{\alpha}{2}} + \norm{ \sum_{i=1}^n \Big\{ Y'_i -\E(Y'_i) \Big\} }_{\frac{\alpha}{2}} \nonumber
		\end{split} \\
		={}& I_5 +I_6 +I_7 +I_8. \label{eq:interMart}
	\end{align}
	To upper bound the first term in \cref{eq:interMart}, note that
	\begin{align*}
		\norm{ \tilde{X}_i^2 -\tilde{X}_{i,\{0\}}^2 }_{\frac{\alpha}{2}}
		&\le \norm{ \E(X_i -X_{i,\{0\}}|\epsilon_{i-m}, \ldots, \epsilon_i, \epsilon'_0) }_\alpha \norm{ \E(X_i +X_{i,\{0\}}|\epsilon_{i-m}, \ldots, \epsilon_i, \epsilon'_0) }_\alpha \\
		&\le \delta_{i,\alpha} \norm{ X_i +X_{i,\{0\}} }_\alpha.
	\end{align*}
	Hence under Assumption \ref{asum:stability} and $\norm{X_i}_\alpha < \infty$,
	\[
	\sum_{i=0}^\infty \norm{ \tilde{X}_i^2 -\tilde{X}_{i,\{0\}}^2 }_{\frac{\alpha}{2}}
	\le \sup_{i \in \N} \norm{ X_i +X_{i,\{0\}} }_\alpha \sum_{i=0}^\infty \delta_{i,\alpha}
	= \Delta_\alpha \sup_{i \in \N} \norm{ X_i +X_{i,\{0\}} }_\alpha
	< \infty.
	\]
	This means that $\{\tilde{X}_i^2\}$ is $\alpha/2$-stable \citep{wu2011asymptotic}.
	Thus, it follows from Lemma 1 of \citet{liu2010asymptotic} that
	$
	I_5 = O(n^{1/\alpha'}).
	$
	To upper bound the second term in \cref{eq:interMart}, we apply the Minkowski inequality:
	\begin{equation} \label{eq:blockDMart}
		I_6 \le \sum_{i=1}^{4m-1} \norm{ \sum_{h=0}^{\lfloor (n-i)/(4m)\rfloor} \Big\{ D_{i+4mh}^2 -\E (D_{i+4mh}^2) \Big\} }_{\frac{\alpha}{2}}.
	\end{equation}
	Since $D_i^2$'s are also $m$-dependent, \cref{eq:blockDMart} becomes
	$
	I_6 = O(m^{1-1/\alpha'} n^{1/\alpha'}).
	$
	For the third term in \cref{eq:interMart}, we try to find an upper bound for $Y_{1,i}$ first.
	By \cref{eq:diffMart}, the Minkowski inequality, Jensen's inequality and \cref{eq:boundAMart},
	\begin{align}
		\begin{split}
			\norm{ \sum_{j=1}^{i-8m} W_n(i,j) (\tilde{X}_j -D_j) }_{\frac{\alpha}{2}}
			={}& \norm{ \sum_{j=1}^{i-8m} W_n(i,j) \Big\{ \E(M_j \mid \mathcal{F}_{j-1}) -\E(M_{j+1} \mid \mathcal{F}_j) \Big\} }_{\frac{\alpha}{2}} \\
			\le{}& C_\alpha m \norm{ X_1 }_{\frac{\alpha}{2}} \max_{1 \le i,j \le n} \big| W_n(i,j) \big| \\
			& +\norm{ \sum_{j=2}^{i-8m} \Big\{ W_n(i,j) -W_n(i,j-1) \Big\} \E(M_j \mid \mathcal{F}_{j-1}) }_{\frac{\alpha}{2}} \nonumber
		\end{split} \\
		={}& C_\alpha m \norm{ X_1 }_{\frac{\alpha}{2}} \max_{1 \le i,j \le n} \big| W_n(i,j) \big| +I_9. \label{eq:interY1Mart}
	\end{align}
	To upper bound $I_9$, note that $\E(M_j \mid \mathcal{F}_{j-1}) = \sum_{h=1}^m \mathcal{P}_{j-h} \E(M_j \mid \mathcal{F}_{j-1})$ and
	$\mathcal{P}_{j-h} \E(M_j \mid \mathcal{F}_{j-1})$ are martingale differences with respect to $\mathcal{F}_{j-h}$.
	They can be verified since
	\begin{align*}
		\sum_{h=1}^m \mathcal{P}_{j-h} \E(M_j \mid \mathcal{F}_{j-1})
		&= \sum_{h=1}^m \Big\{ \E(M_j \mid \mathcal{F}_{j-h}) -\E(M_j \mid \mathcal{F}_{j-h-1}) \Big\} \\
		&= \E(M_j \mid \mathcal{F}_{j-1}) -\E(M_j \mid \mathcal{F}_{j-m-1})
		= \E(M_j \mid \mathcal{F}_{j-1}),
	\end{align*}
	and
	\[
	\E \Big\{ \mathcal{P}_{j-h} \E(M_j \mid \mathcal{F}_{j-1}) \mid \mathcal{F}_{j-h-1} \Big\}
	= \E(M_j \mid \mathcal{F}_{j-h-1}) -\E(M_j \mid \mathcal{F}_{j-h-1}) = 0.
	\]
	Hence, by the Burkholder inequality, Minkowski inequality, Jensen's inequality and \cref{eq:boundAMart},
	\begin{align*}
		I_9^{\alpha'}
		&\le C_\alpha \sum_{j=2}^{i-8m} \norm{ \Big\{ W_n(i,j) -W_n(i,j-1) \Big\} \sum_{h=1}^m \mathcal{P}_{j-h} \E(M_j \mid \mathcal{F}_{j-1}) }_{\frac{\alpha}{2}}^{\alpha'} \\
		&\le C_\alpha \sum_{j=2}^{i-8m} \big| W_n(i,j) -W_n(i,j-1) \big|^{\alpha'}
		\left\{ \sum_{h=1}^m \norm{ \mathcal{P}_{j-h} \E(M_j \mid \mathcal{F}_{j-1}) }_{\frac{\alpha}{2}} \right\}^{\alpha'} \\
		&\le C_\alpha m^{2\alpha'} \norm{ X_1 }_{\frac{\alpha}{2}}^{\alpha'}
		\sum_{j=2}^{i-8m} \big| W_n(i,j) -W_n(i,j-1) \big|^{\alpha'},
	\end{align*}
	which implies that
	\[
	I_9 \le C_\alpha m^2 \norm{ X_1 }_{\frac{\alpha}{2}} \left\{ \sum_{j=2}^{i-8m} \big| W_n(i,j) -W_n(i,j-1) \big|^{\alpha'} \right\}^{\frac{1}{\alpha'}}.
	\]
	\cref{eq:interY1Mart} becomes
	\begin{equation} \label{eq:boundY1Mart}
		\norm{ \sum_{j=1}^{i-8m} W_n(i,j) (\tilde{X}_j -D_j) }_{\frac{\alpha}{2}}
		\le C_\alpha m^2 G_{6,n} \norm{ X_1 }_{\frac{\alpha}{2}}.
	\end{equation}
	Similarly, we have
	\begin{equation} \label{eq:boundY2Mart}
		\norm{ \sum_{j=i-8m+1}^{i-1} W_n(i,j) (\tilde{X}_j -D_j) }_{\frac{\alpha}{2}}
		\le C_\alpha m^2 G_{6,n} \norm{ X_1 }_{\frac{\alpha}{2}}.
	\end{equation}
	Now to apply the idea of independence as in \cref{eq:blockDMart}, check that $Y_{1,i}, Y_{1,i+4m}, \ldots$ are martingale differences as
	\begin{align*}
		\E(Y_{1,i+4m} \mid \mathcal{F}_i) &= \E \left\{ \tilde{X}_{i+4m} \sum_{j=1}^{i-4m} W_n(i,j) (\tilde{X}_j -D_j)\ \Bigg| \ \mathcal{F}_i \right\} \\
		&= \sum_{j=1}^{i-4m} W_n(i,j) (\tilde{X}_j -D_j) \E (\tilde{X}_{i+4m}) = 0.
	\end{align*}
	Using Hölder's inequality and \cref{eq:boundY1Mart},
	\[
	\norm{ Y_{1,i} }_{\frac{\alpha}{2}}
	\le \norm{ \tilde{X}_i }_\alpha \norm{ \sum_{j=1}^{i-8m} W_n(i,j) (\tilde{X}_j -D_j) }_\alpha
	\le C_\alpha m^2 G_{6,n} \norm{ X_1 }_{\alpha}^2.
	\]
	Thus, by Burkholder inequality,
	\[
	\norm{ \sum_{h=0}^{\lfloor (n-i)/(4m) \rfloor} Y_{1,i+4mh} }_{\frac{\alpha}{2}}^{\alpha'}
	\le C_\alpha \sum_{h=0}^{\lfloor (n-i)/(4m) \rfloor} \norm{ Y_{1,i+4mh} }_{\frac{\alpha}{2}}^{\alpha'}
	\le C_\alpha m^{2\alpha'-1} n G_{6,n}^{\alpha'} \norm{ X_1 }_\alpha^{2\alpha'},
	\]
	which implies
	\[
	\norm{ \sum_{h=0}^{\lfloor (n-i)/(4m) \rfloor} Y_{1,i+4mh} }_{\frac{\alpha}{2}}
	\le C_\alpha m^{2-\frac{1}{\alpha'}} n^{\frac{1}{\alpha'}} G_{6,n} \norm{ X_1 }_\alpha^2.
	\]
	By the Minkowski inequality,
	\begin{equation}
		\norm{ \sum_{i=1}^n Y_{1,i} }_{\frac{\alpha}{2}}
		\le \sum_{i=1}^{4m-1} \norm{ \sum_{h=0}^{\lfloor (n-i)/(4m) \rfloor} Y_{1,i+4mh} }_{\frac{\alpha}{2}}
		\le C_\alpha m^{3-\frac{1}{\alpha'}} n^{\frac{1}{\alpha'}} G_{6,n} \norm{ X_1 }_\alpha^2. \label{eq:blockY1Mart}
	\end{equation}
	For $Y_{2,i}$, observe that
	\[
	\E \big\{ Y_{2,i+12m} -\E(Y_{2,i+12m}) \mid \mathcal{F}_i \big\}
	= \E(Y_{2,i+12m}) -\E(Y_{2,i+12m}) = 0.
	\]
	Therefore, by \cref{eq:boundY2Mart} and similar arguments as in \cref{eq:blockY1Mart}, we have
	\[
	\norm{ \sum_{i=1}^n \big\{ Y_{2,i} -\E(Y_{2,i}) \big\} }_{\frac{\alpha}{2}}
	\le C_\alpha m^{3-\frac{1}{\alpha'}} n^{\frac{1}{\alpha'}} G_{6,n} \norm{ X_1 }_\alpha^2.
	\]
	By the Minkowski inequality and $\E|X_1|^\alpha < \infty$, we now have
	\begin{equation} \label{eq:boundYMart}
		I_7 \le \norm{ \sum_{i=1}^n Y_{1,i} }_{\frac{\alpha}{2}} +\norm{ \sum_{i=1}^n \big\{ Y_{2,i} -\E(Y_{2,i}) \big\} }_{\frac{\alpha}{2}}
		= O\left( m^{3-\frac{1}{\alpha'}} n^{\frac{1}{\alpha'}} G_{6,n} \right).
	\end{equation}
	To upper bound the last term in \cref{eq:interMart}, rearrange the order of summation by
	\[
	\sum_{i=1}^n Y'_i
	= 2 \sum_{i=2}^n \big( \tilde{X}_i -D_i \big) \sum_{j=1}^{i-1} W_n(i,j) D_j
	= 2 \sum_{j=1}^{n-1} D_j \sum_{i=j+1}^n W_n(i,j) \big( \tilde{X}_i -D_i \big).
	\]
	By similar arguments as in \cref{eq:boundY1Mart,eq:blockY1Mart,eq:boundYMart},
	\[
	I_8 = O\left( m^{3-\frac{1}{\alpha'}} n^{\frac{1}{\alpha'}} G_{6,n} \right).
	\]
	Combining the results in \cref{eq:interMart} gives us
	\begin{align*}
		\norm{ \tilde{\sigma}^2_n -\E(\tilde{\sigma}^2_n) -\breve{\sigma}^2_n +\E(\breve{\sigma}^2_n) }_{\frac{\alpha}{2}}
		&= n^{-1} \left\{ O(n^{\frac{1}{\alpha'}}) +O(m^{1-\frac{1}{\alpha'}} n^{\frac{1}{\alpha'}}) +O\left( m^{3-\frac{1}{\alpha'}} n^{\frac{1}{\alpha'}} G_{6,n} \right) \right\} \\
		&= O\left( m^{3-\frac{1}{\alpha'}} n^{\frac{1}{\alpha'}-1} G_{6,n} \right),
	\end{align*}
	since $\alpha' \in (1,2]$.
	For the additional case, $m$ needs to be chosen such that $m \to \infty$ when $n \to \infty$ for $m$-dependent process approximation to work well.
	Under Assumptions \ref{asum:winSums} and \ref{asum:winGen}(a),
	$
	G_{6,n} \le 2 G_{2,n} = o(n^c)
	$
	for some $c \in (0, 1-1/\alpha')$.
	Therefore, We can choose $m = O(n^{(1-1/\alpha'-c)/3})$ such that $m \to \infty$ as $n \to \infty$, and
	$\norm{ \tilde{\sigma}^2_n -\E(\tilde{\sigma}^2_n) -\breve{\sigma}^2_n +\E(\breve{\sigma}^2_n) }_{\alpha/2} = o(1)$.
\end{proof}

\subsection{\texorpdfstring{Lemma \ref{lem:varDMart}}{
		Lemma F.4}} \label{sec:proof-varDMart} 

To establish the consistency or variance of our estimators under martingale difference approximation, we have the following lemma.

\begin{lemma} \label{lem:varDMart}
	Let $\alpha > 2$.
	If Assumption \ref{asum:stability} holds, then as $m \to \infty$,
	\[
	\norm{D_i} \to \sigma \quad \text{and} \quad
	\Var(D_i D_j) \to \sigma^4,
	\]
	for positive integers $i$ and $j$ where $|i-j| > m$.
\end{lemma}

\begin{proof}
	Recall that the long-run variance can be expressed in a probabilistic representation $\sigma = \norm{ \sum_{h=i}^\infty \mathcal{P}_i X_h }$; see, e.g., \citet{wu2009recursive}.
	By the Minkowski inequality, Theorem 1 of \citet{wu2005asymptotic} and \cref{eq:boundMDepend},
	\begin{align*}
		\norm{ D_i -\sum_{h=i}^\infty \mathcal{P}_i X_h }
		&= \norm{ \sum_{h=0}^\infty \Big\{ \E(\tilde{X}_{i+h} \mid \mathcal{F}_i) -\E(\tilde{X}_{i+h} \mid \mathcal{F}_{i-1}) \Big\}
			-\sum_{h=i}^\infty \mathcal{P}_i X_h } \\
		&= \norm{ \sum_{h=i}^\infty \mathcal{P}_i (\tilde{X}_h -X_h) } \\
		&\le \sum_{h=i}^\infty \norm{ \tilde{X}_h -X_h -\tilde{X}_{h, \{i\}} +X_{h, \{i\}} }
		\le C_2 d_{m,2},
	\end{align*}
	where $d_{m,2}$ is defined in \cref{eq:dMDepend}.
	As $m \to \infty$, $d_{m,2} = o(1)$ in view of Lemma \ref{lem:dMDepend}.
	Thus, we have
	\[
	\lim_{m \to \infty} \norm{D_i -\sum_{h=i}^\infty \mathcal{P}_i X_h} = 0.
	\]
	By the Minkowski inequality,
	\[
	\limsup_{m \to \infty} \norm{D_i}
	\le \limsup_{m \to \infty} \norm{ D_i -\sum_{h=i}^\infty \mathcal{P}_i X_h }
	+ \limsup_{m \to \infty} \norm{\sum_{h=i}^\infty \mathcal{P}_i X_h}
	= \sigma.
	\]
	At the same time,
	\[
	\liminf_{m \to \infty} \norm{\sum_{h=i}^\infty \mathcal{P}_i X_h}
	\le \liminf_{m \to \infty} \norm{D_i}
	+\liminf_{m \to \infty} \norm{ \sum_{h=i}^\infty \mathcal{P}_i X_h -D_i },
	\]
	which implies $\liminf_{m \to \infty} \norm{D_i} \ge \sigma$.
	By the Sandwich Theorem, $\lim_{m \to \infty} \norm{D_i} = \sigma$.
	For $\Var(D_i D_j)$, recall \hyperref[eq:dMart]{$D_i$}'s form martingale differences with respect to $\mathcal{F}_i$ and are $m$-dependent.
	Without loss of generality, assume $i > j$.
	Then $\E(D_i D_j) = \E\{ D_j \E(D_i \mid \mathcal{F}_j) \} = 0$ and
	\begin{align*}
		\Var(D_i D_j) &= \E(D_i^2 D_j^2) -\{ \E(D_i D_j) \}^2 \\
		&= \E(D_i^2) \E(D_j^2) \\
		&= \norm{D_i}^2 \norm{D_j}^2
		\to \sigma^4,
	\end{align*}
	where $|i-j| > m$ as $m \to \infty$.
\end{proof}

\subsection{\texorpdfstring{Lemma \ref{lem:dMDepend}}{
		Lemma F.5}} \label{sec:proof-dMDepend} 

The following lemma provides an asymptotic upper bound of \hyperref[eq:dMDepend]{$d_{m,\alpha}$}.

\begin{lemma} \label{lem:dMDepend}
	Let $\alpha > 2$.
	If Assumption \ref{asum:stability} holds, then as $m \to \infty$,
	$d_{m,\alpha} = o(1)$.
\end{lemma}

\begin{proof}
	Let $M = f(m)$ such that $\Upsilon_{m+1,\alpha} > \sup_{i \ge M} \delta_{i, \alpha}$ and $M \to \infty$ as $m \to \infty$.
	Then
	\begin{align*}
		d_{m,\alpha}
		&= \sum_{i=0}^\infty \min \big( \delta_{i, \alpha}, \Upsilon_{m+1, \alpha} \big) \\
		&\le \sum_{i=0}^{M-1} \Upsilon_{m+1, \alpha} +\sum_{i=M}^\infty \delta_{i, \alpha}
		= M \Upsilon_{m+1, \alpha} +\sum_{i=M}^\infty \delta_{i, \alpha}.
	\end{align*}
	As $m \to \infty$, $\sum_{i=M}^\infty \delta_{i, \alpha} \to 0$ and $\Upsilon_{m+1, \alpha} \to 0$ by Assumption \ref{asum:stability}.
	It suffices to show that $M \Upsilon_{m+1, \alpha} \to 0$, which we may apply the L'Hôpital's rule:
	\[
	\lim_{m \to \infty} M \Upsilon_{m+1, \alpha}
	= \lim_{m \to \infty} \frac{f(m)}{1/\Upsilon_{m+1, \alpha}}
	= \lim_{m \to \infty} -\frac{f'(m) \Upsilon_{m+1, \alpha}^2}{\Upsilon'_{m+1, \alpha}},
	\]
	which is 0 if
	\[
	f'(m) = -\frac{\Upsilon'_{m+1, \alpha}}{\Upsilon_{m+1, \alpha}}
	\Leftrightarrow f(m) = -\ln(\Upsilon_{m+1, \alpha}) +C,
	\]
	where $C$ is an arbitrary constant.
	We check that $-\ln(\Upsilon_{m+1, \alpha}) \to \infty$ since $\Upsilon_{m+1, \alpha} \to 0$.
	To ensure $f(m)$ is always positive, we can take $C = \ln(\Upsilon_{0, \alpha}) +1$.
\end{proof}

\subsection{\texorpdfstring{Lemma \ref{lem:lastBlock}}{
		Lemma F.6}} \label{sec:proof-lastBlock} 

The following lemma provides asymptotic upper bounds of the sum of squared terms and the last incomplete block \hyperref[eq:interBlockRamp]{$\bar{\sigma}_{2,n}^2$} when $\phi>1$,
and the sum of squared terms \hyperref[eq:interBlock]{$\bar{\sigma}_{4,n}^2$} when $\phi=1$.

\begin{lemma} \label{lem:lastBlock}
	Let $\alpha \ge 4$.
	Suppose that $\norm{X_1}_\alpha < \infty$.
	If Assumption \ref{asum:stability} holds, then
	\[
	\norm{ \bar{\sigma}_{2,n}^2 -\E(\bar{\sigma}_{2,n}^2) }^2 = o\big( n^{\psi +\max\{ 2q(\psi-\theta), 0\} -1} \big) \quad \text{and} \quad
	\norm{ \bar{\sigma}_{4,n}^2 -\E(\bar{\sigma}_{4,n}^2) }^2 = O(n^{-1}).
	\]
\end{lemma}

\begin{proof}
	We first establish the stability of $\{X_i X_{i-k}\}$ for a fixed $k \in [0,i) \cap \Z$.
	By the Minkowski inequality and Hölder's inequality,
	\begin{align*}
		\norm{ X_i X_{i-k} -X_{i, \{0\}} X_{i-k, \{0\}} }_{\frac{\alpha}{2}}
		&= \norm{ X_i (X_{i-k}-X_{i-k, \{0\}}) +X_{i-k, \{0\}}(X_i -X_{i, \{0\}}) }_{\frac{\alpha}{2}} \\
		&\le \delta_{i-k,\alpha} \norm{X_i}_\alpha +\delta_{i,\alpha} \norm{X_{i-k, \{0\}}}_\alpha.
	\end{align*}
	Hence under Assumption \ref{asum:stability} and $\norm{X_i}_\alpha < \infty$,
	\[
	\sum_{i=0}^\infty \norm{ X_i X_{i-k} -X_{i, \{0\}} X_{i-k, \{0\}} }_{\frac{\alpha}{2}}
	\le \Delta_\alpha \left( \sup_{i \in \N} \norm{X_i}_\alpha + \sup_{i \in \N} \norm{X_{i-k, \{0\}}}_\alpha \right)
	< \infty,
	\]
	which means $\{X_i X_{i-k}\}$ is $\alpha/2$-stable for a fixed $k \in [0,i) \cap \Z$ \citep{wu2011asymptotic}.
	For $\alpha \ge 4$, it follows from Lemma 1 of \citet{liu2010asymptotic} that
	\[
	\norm{ \bar{\sigma}_{4,n}^2 -\E(\bar{\sigma}_{4,n}^2) }
	= \frac{1}{n} \norm{ \sum_{i=1}^n \left\{ X_i^2 -\E(X_i^2) \right\} }
	= \frac{1}{n} \cdot O \left\{ \left(  \sum_{i=1}^n 1^2 \right)^{\frac{1}{2}} \right\}
	= O(n^{-\frac{1}{2}}).
	\]
	In addition, by the Minkowski inequality and Lemma 1 of \citet{liu2010asymptotic},
	\begin{align*}
		\norm{ \bar{\sigma}_{2,n}^2 -\E(\bar{\sigma}_{2,n}^2) }
		&\le \frac{2}{n} \norm{ \sum_{j=\eta_{s_n}}^n \sum_{k=1}^{s'_j} \left( 1 -\frac{k^q}{t_n^q} \right) \{ X_j X_{j-k} -\E(X_j X_{j-k}) \} } +O( n^{-\frac{1}{2}}) \\
		&\le \frac{2}{n} \sum_{k=1}^{\lceil \phi s_n \rceil} \norm{ \sum_{j=\eta_{s_n}}^n \left( 1 -\frac{k^q}{t_n^q} \right) \I_{k \le s'_j} \{ X_j X_{j-k} -\E(X_j X_{j-k}) \} } +O( n^{-\frac{1}{2}}) \\
		&\le O(n^{-1}) \sum_{k=1}^{\lceil \phi s_n \rceil} \sqrt{ \sum_{j=\eta_{s_n}}^n \left| 1 -\frac{k^q}{t_n^q} \right|^2 } +O( n^{-\frac{1}{2}}) \\
		&= O\big( n^{\psi +\max\{ q(\psi-\theta, 0 \} -1} \big) +O( n^{-\frac{1}{2}}).
	\end{align*}
	Check that $2\psi +\max\{ 2q(\psi-\theta), 0\} -2 < \psi +\max\{ 2q(\psi-\theta), 0\} -1$ and $-1 < \psi +\max\{ 2q(\psi-\theta), 0\} -1$ always hold under Definition \ref{def:winLASER}.
\end{proof}

\subsection{\texorpdfstring{Lemma \ref{lem:varExact}}{
		Lemma F.7}}\label{sec:proof-varExact} 

Define
\begin{equation} \label{eq:varBoth}
	H_{7,\phi} = \left\{
	\begin{array}{ll}
		\displaystyle \frac{2}{n} \sum_{h=\lceil (m+1)/\phi \rceil}^{s_n-1} \sum_{i=1}^{\nu_h} \sum_{j=\eta_h+(i-1)\varpi_h+m+1}^{\eta_h+i\varpi_h-1} \sum_{k=m+1}^{j-\eta_h-(i-1)\varpi_h+h} \left( 1 -\frac{k^q}{t_n^q} \right) D_j D_{j-k}, & \phi > 1; \\
		\displaystyle \frac{2}{n} \sum_{h=m+1}^{s_n} \sum_{i=\eta_h}^{n} \left( 1 -\frac{h^q}{t_n^q} \right) D_i D_{i-h}, & \phi = 1. \\
	\end{array}
	\right.
\end{equation}

Note that $\Var(H_{7,\phi}) = H_3^2$ for $\phi>1$, and $\Var(H_{7,\phi}) = H_5^2$ for $\phi=1$, where $H_3$ and $H_5$ are defined in \cref{eq:interVarRamp,eq:interVar}, respectively.
The following lemma provides the exact convergence rate of $\Var(H_{7,\phi})$.

\begin{lemma} \label{lem:varExact}
	Let $\alpha \ge 4$.
	Suppose that $\norm{X_1}_\alpha < \infty$.
	If Assumption \ref{asum:stability} holds, then as $n \to \infty$,
	\[
	\Var(H_{7,\phi}) \sim \mathcal{V}_{\psi,\Psi,\theta,\Theta,q,\phi} \sigma^4 n^{\psi +\max\{ 2q(\psi-\theta), 0\} -1},
	\]
	where $\mathcal{V}_{\psi,\Psi,\theta,\Theta,q,\phi}$
	\[
	= \left\{
	\begin{array}{ll}
		\frac{2 \Psi \left(\phi + 1\right)}{\psi + 1} \I_{\psi \le \theta}
		-\frac{8 \Psi^{q + 1} \Theta^{- q} \left(\phi^{q + 2} - 1\right)}{\left(\phi - 1\right) \left(q + 1\right) \left(q + 2\right) \left(\psi q + \psi + 1\right)} \I_{\psi = \theta}
		+\frac{2 \Psi^{2 q + 1} \Theta^{- 2 q} \left(\phi^{2 q + 2} - 1\right)}{\left(\phi - 1\right) \left(q + 1\right) \left(2 q + 1\right) \left(2 \psi q + \psi + 1\right)} \I_{\psi \ge \theta}, & \phi > 1; \\
		\frac{4 \Psi}{\psi + 1} \I_{\psi \le \theta}
		-\frac{8 \Psi^{q + 1} \Theta^{- q}}{\left(q + 1\right) \left(\psi q + \psi + 1\right)} \I_{\psi = \theta}
		+\frac{4 \Psi^{2 q + 1} \Theta^{- 2 q}}{\left(2 q + 1\right) \left(2 \psi q + \psi + 1\right)} \I_{\psi \ge \theta} , & \phi = 1. \\
	\end{array}
	\right.
	\]
\end{lemma}

\begin{proof}
	Consider $\phi > 1$ first.
	Recall that \hyperref[eq:dMart]{$D_i$}'s form martingale differences with respect to $\mathcal{F}_i$ and are $m$-dependent.
	Therefore, $\E(D_i D_{i-k}) = 0$ if $k > m$, and $\E(D_i D_{i-k} D_j D_{j-k}) = 0$ if $k > m$ and $i \neq j$.
	By Lemma \ref{lem:varDMart}, $\Var(D_i D_{i-k}) \to \sigma^4$ for $k > m$ as $m \to \infty$.
	Thus, as $n \to \infty$,
	\begin{align}
		\Var(H_{7,\phi}) &= \frac{4}{n^2} \Var\left\{ \sum_{h=\lceil (m+1)/\phi \rceil}^{s_n-1} \sum_{i=1}^{\nu_h} \sum_{j=\eta_h+(i-1)\varpi_h+m+1}^{\eta_h+i\varpi_h-1} \sum_{k=m+1}^{j-\eta_h-(i-1)\varpi_h+h} \left( 1 -\frac{k^q}{t_n^q} \right) D_j D_{j-k} \right\} \nonumber \\
		&\sim \frac{4\sigma^4}{n^2} \sum_{h=\lceil (m+1)/\phi \rceil}^{s_n-1} \sum_{i=1}^{\nu_h} \sum_{j=\eta_h+(i-1)\varpi_h+m+1}^{\eta_h+i\varpi_h-1} \sum_{k=m+1}^{j-\eta_h-(i-1)\varpi_h+h} \left( 1 -\frac{k^q}{t_n^q} \right)^2. \label{eq:riemannVarRamp}
	\end{align}
	Since we can choose $m=o(n^{\psi/6})$ as in \cref{sec:proof-var-mDepend}, we can approximate \cref{eq:riemannVarRamp} by Riemann integrals and use change of variables to obtain
	\begin{align*}
		\Var(H_{7,\phi}) &\sim \frac{4\sigma^4}{n^2} \int_0^{s_n} \int_0^{\nu_h} \int_{\eta_h+(i-1)\varpi_h}^{\eta_h+i\varpi_h} \int_0^{j-\eta_h-(i-1)\varpi_h+h} \left( 1 -\frac{k^q}{t_n^q} \right)^2 \,\di k\ \di j\ \di i\ \di h \\
		&= \frac{4\sigma^4}{n^2} \int_0^{s_n} \int_0^{\nu_h} \int_0^{\varpi_h} \int_0^{h+j} \left( 1 -\frac{k^q}{t_n^q} \right)^2\, \di k\ \di j\ \di i\ \di h \\
		&= \left\{
		\begin{array}{ll}
			4 \sigma^4 n^{\psi-1} \int_0^\Psi \int_0^{\nu_w} \int_0^{(\phi-1)w} \int_0^{w+y} 1 \,\di z\ \di y\ \di x\ \di w, & \psi < \theta; \\
			4 \sigma^4 n^{\psi-1} \int_0^\Psi \int_0^{\nu_w} \int_0^{(\phi-1)w} \int_0^{w+y} (1 -z^q \Theta^{-q})^2 \, \di z\ \di y\ \di x\ \di w, & \psi = \theta; \\
			4 \sigma^4 n^{\psi+2q(\psi-\theta)-1} \int_0^\Psi \int_0^{\nu_w} \int_0^{(\phi-1)w} \int_0^{w+y} z^{2q} \Theta^{-2q} \, \di z\ \di y\ \di x\ \di w, & \psi > \theta \\
		\end{array}
		\right. \\
		&= \mathcal{V}_{\psi,\Psi,\theta,\Theta,q,\phi} \sigma^4 n^{\psi +\max\{ 2q(\psi-\theta), 0\} -1}.
	\end{align*}
	Evaluating the integral gives
	\[
	\textstyle
	\mathcal{V}_{\psi,\Psi,\theta,\Theta,q,\phi}
	= \frac{2 \Psi \left(\phi + 1\right)}{\psi + 1} \I_{\psi \le \theta}
	- \frac{8 \Psi^{q + 1} \Theta^{- q} \left(\phi^{q + 2} - 1\right)}{\left(\phi - 1\right) \left(q + 1\right) \left(q + 2\right) \left(\psi q + \psi + 1\right)} \I_{\psi = \theta}
	+ \frac{2 \Psi^{2 q + 1} \Theta^{- 2 q} \left(\phi^{2 q + 2} - 1\right)}{\left(\phi - 1\right) \left(q + 1\right) \left(2 q + 1\right) \left(2 \psi q + \psi + 1\right)} \I_{\psi \ge \theta}.
	\]
	Now consider $\phi = 1$.
	By similar arguments as in \cref{eq:riemannVarRamp},
	\[
	\Var(H_{7,\phi}) \sim \frac{4 \sigma^4}{n} \sum_{h=m+1}^{s_n} \sum_{i=\eta_h}^{n} \left( 1 -\frac{h^q}{t_n^q} \right)^2.
	\]
	Using Riemann integral approximation and change of variables,
	\begin{align*}
		\Var(H_{7,\phi}) &\sim \frac{4\sigma^4}{n^2} \int_0^{s_n} \int_{\eta_h}^n \left( 1 -\frac{h^q}{t_n^q} \right)^2 \, \di i\ \di h \\
		&= \left\{
		\begin{array}{ll}
			4 \sigma^4 n^{\psi-1} \int_0^\Psi \int_{(w/\Psi)^{1/\psi}}^1 1 \,\di x\ \di w, & \psi < \theta; \\
			4 \sigma^4 n^{\psi-1} \int_0^\Psi \int_{(w/\Psi)^{1/\psi}}^1 (1 -w^q \Theta^{-q})^2 \,\di x\ \di w, & \psi = \theta; \\
			4 \sigma^4 n^{\psi+2q(\psi-\theta)-1} \int_0^\Psi \int_{(w/\Psi)^{1/\psi}}^1 w^{2q} \Theta^{-2q} \,\di x\ \di w, & \psi > \theta \\
		\end{array}
		\right. \\
		&= \mathcal{V}_{\psi,\Psi,\theta,\Theta,q,\phi} \sigma^4 n^{\psi +\max\{ 2q(\psi-\theta), 0\} -1}.
	\end{align*}
	Evaluating the integral gives
	\[
	\mathcal{V}_{\psi,\Psi,\theta,\Theta,q,\phi}
	= \frac{4 \Psi}{\psi + 1} \I_{\psi \le \theta}
	- \frac{8 \Psi^{q + 1} \Theta^{- q}}{\left(q + 1\right) \left(\psi q + \psi + 1\right)} \I_{\psi = \theta}
	+ \frac{4 \Psi^{2 q + 1} \Theta^{- 2 q}}{\left(2 q + 1\right) \left(2 \psi q + \psi + 1\right)} \I_{\psi \ge \theta}.
	\]
\end{proof}

\section{Additional Results} \label{sec:additional-results}

Some results that complement the main text are given here.

\subsection{\texorpdfstring{Example \ref{eg:on}: Window of Subsample Selection Rules}{
		Example 2: Window of Subsample Selection Rules}} \label{sec:omit-on} 

In this subsection, we try to find the window associated with the subsampling estimator defined in \cref*{eq:subClass}.
Recall that
\[
\hat{\sigma}^2_{n,\sub} = \frac{\sum_{i=1}^n \left\{ \sum_{j=i-\ell_i+1}^i (X_j -\bar{X}_n) \right\}^2}{\sum_{i=1}^n \ell_i}.
\]
Let $\bar{\ell}_n = n^{-1} \sum_{i=1}^n \ell_i$ be the global average batch size.
Assume that $t\mapsto t-\ell_t$ is a monotonically increasing sequence, which is satisfied by existing subsample selection rules.
Then, we can write
\begin{align*}
	\hat{\sigma}^2_{n,\sub}
	&= \frac{1}{n \bar{\ell}_n} \sum_{i=1}^n \sum_{j=i-\ell_i+1}^i \sum_{j'=i-\ell_i+1}^i (X_j -\bar{X}_n) (X_{j'} -\bar{X}_n) \\
	&= \frac{1}{n \bar{\ell}_n} \sum_{j=1}^n \sum_{j'=1}^n \sum_{i\in I_n(j,j')} (X_j -\bar{X}_n) (X_{j'} -\bar{X}_n)  \\
	&= \frac{1}{n} \sum_{i=1}^n \sum_{j=1}^n \frac{\left\vert I_n(i,j) \right\vert}{\bar{\ell}_n} (X_i -\bar{X}_n) (X_j -\bar{X}_n),
\end{align*}
where $I_n(i,j) = \{ t\in\{1,\ldots,n\}: i,j\in\{t-\ell_t+1, \ldots, t\}\}$.
Denote $f_n(i) = \max\{t\in\{1,\ldots,n\}: i\in\{t-\ell_t+1, \ldots, t\}\}$ and $b_n(i) = f_n(i)-i+1$.
Physically, the meaning of $f_n(i)$ is the largest index such that the $i$-th observation is still included in the subsample,
while $b_n(i)$ is the subsample size of the $f_n(i)$-th subsample.
We can express the cardinality of $I_n(i,j)$ as
\[
|I_n(i,j)| = \bigg\{ b_n(i \land j) - |i-j|\bigg\}\I_{|i-j| \le b_n(i \land j)}.
\]
Consequently, $\hat{\sigma}^2_{n,\sub}$ is exactly equivalent to
$\hat{\sigma}^2_n(W_\sub)$ with
\[
W_\sub = W_{n,\sub}(i,j)
= \left( 1 -\frac{|i-j|+\bar{\ell}_n-b_n(i \land j)}{\bar{\ell}_n} \right) \I_{|i-j| \le b_n(i \land j)}.
\]
Note that the subsampling parameter admits the form $b_n(i \land j)$ because the way to subsample differs under the subsampling form in \cref*{eq:subClass} and the quadratic form in \cref*{eq:genClass}.
Nevertheless, we can approximate $W_\sub$ under $\PSR$ asymptotically by
\[
W_\PSR = W_{n,\PSR}(i,j)
= \left( 1 -\frac{|i-j|+\bar{\ell}_n-\ell_{i \lor j}}{\bar{\ell}_n}\right) \I_{|i-j| \le \ell_{i \lor j}}.
\]
We formalize this result with the following proposition.

\begin{proposition}[Asymptotic equivalence of $\PSR$] \label{prop:sub}
	Let $\alpha \ge 4$.
	Suppose that $\norm{X_1}_\alpha < \infty$.
	If Assumptions \ref{asum:stability} and \ref{asum:qWeakStable} hold for $q=1$, and
	$\ell_i = \min(\lfloor \Lambda i^\lambda \rfloor, i)$ for some $\Lambda \in \R^+$ and $\lambda \in (0,1/2)$, i.e., $\PSR$ is used, then
	\[
	\frac{\norm{\hat{\sigma}^2_n(W_\sub)-\hat{\sigma}^2_n(W_\PSR)}^2}{\norm{\hat{\sigma}^2_n(W_\sub)-\sigma^2}^2}
	= o(1).
	\]
\end{proposition}

\begin{proof}
	When Assumptions \ref{asum:stability} and \ref{asum:qWeakStable} hold for $q=1$, we can apply Theorem 3 of \citet{rtacm} to obtain $\norm{\hat{\sigma}^2_n(W_\sub)-\sigma^2}^2 = O(n^{\max(-2\lambda, \lambda-1)})$.
	It suffices to show that
	\[
	\norm{\hat{\sigma}^2_n(W_\sub)-\hat{\sigma}^2_n(W_\PSR)}
	= o(n^{\max(-\lambda, \frac{\lambda-1}{2})}).
	\]
	Recall the definition of the known-mean version $\bar{\sigma}_n^2$ in \cref{eq:sigmaMean}.
	We can follow \cref{sec:proof-var-mean} to replace the sample mean similarly and assume $\mu = 0$ without loss of generality.
	This is because our Lemma \ref{lem:eqvMean} and the Lemma E.1 in \citet{rtacm} give
	\[
	\norm{\hat{\sigma}^2_n(W_\PSR)-\bar{\sigma}^2_n(W_\PSR)} = O(n^{\lambda-1}) \qquad \text{and} \qquad
	\norm{\hat{\sigma}^2_n(W_\sub)-\bar{\sigma}^2_n(W_\sub)} = O(n^{\lambda-1})
	\]
	Check that $\lambda-1 < \max(-\lambda, \lambda-1/2)$ always hold because $\lambda \in (0,1/2)$.
	Therefore, we only need to prove
	\[
	\norm{\bar{\sigma}^2_n(W_\sub)-\bar{\sigma}^2_n(W_\PSR)}
	= o(n^{\max(-\lambda, \frac{\lambda-1}{2})}).
	\]
	Express $\bar{\sigma}^2_n(W_\sub)$ and $\bar{\sigma}^2_n(W_\PSR)$ by
	\begin{align*}
		\bar{\sigma}^2_n(W_\sub)
		&= \frac{1}{n \bar{\ell}_n} \sum_{i=1}^n \sum_{j=1}^n \bigg\{ b_n(i \land j) - |i-j|\bigg\}\I_{|i-j| \le b_n(i \land j)} X_i X_j, \\
		\bar{\sigma}^2_n(W_\PSR)
		&= \frac{1}{n \bar{\ell}_n} \sum_{i=1}^n \sum_{j=1}^n \bigg( \ell_{i \lor j} - |i-j|\bigg)\I_{|i-j| \le \ell_{i \lor j}} X_i X_j.
	\end{align*}
	Denote $F_n(i) = \max\{t \in \Z^+: i\in\{t-\ell_t+1, \ldots, t\}\}$,
	$\Xi_{1,n} = \{i \in \{1, \ldots, n\}: b_n(i) -\ell_i > 1 \}$,
	$\Xi_{2,n} = \{i \in \{1, \ldots, n\}: f_n(i) < F_n(i) \}$, and
	$\Xi_{3,n} = \{1,\ldots,n\} \setminus \{\Xi_{1,n} \cup \Xi_{2,n}\}$.
	The meaning of $\Xi_{1,n}$ is the set of indices in the beginning where $\ell_i$ is strictly increasing,
	while the meaning of $\Xi_{2,n}$ is the set of indices that belongs to the last incomplete block where $F_n(i)$ is truncated by $n$.
	Then, we partition $\bar{\sigma}^2_n(W_\sub)$ into
	\begin{align*}
		\bar{\sigma}^2_{1,n}(W_\sub)
		={}& \frac{1}{n \bar{\ell}_n} \sum_{i \in \Xi_{1,n}} \sum_{j=1}^n \bigg\{ b_n(i \land j) - |i-j|\bigg\}\I_{|i-j| \le b_n(i \land j)} X_i X_j \\
		&+ \frac{1}{n \bar{\ell}_n} \sum_{i \in \Xi_{3,n}} \sum_{j \in \Xi_{1,n}} \bigg\{ b_n(i \land j) - |i-j|\bigg\}\I_{|i-j| \le b_n(i \land j)} X_i X_j, \\
		\bar{\sigma}^2_{2,n}(W_\sub)
		={}& \frac{1}{n \bar{\ell}_n} \sum_{i \in \Xi_{2,n}} \sum_{j=1}^n \bigg\{ b_n(i \land j) - |i-j|\bigg\}\I_{|i-j| \le b_n(i \land j)} X_i X_j \\
		&+ \frac{1}{n \bar{\ell}_n} \sum_{i \in \Xi_{3,n}} \sum_{j \in \Xi_{2,n}} \bigg\{ b_n(i \land j) - |i-j|\bigg\}\I_{|i-j| \le b_n(i \land j)} X_i X_j, \\
		\bar{\sigma}^2_{3,n}(W_\sub)
		={}& \frac{1}{n \bar{\ell}_n} \sum_{i \in \Xi_{3,n}} \sum_{j \in \Xi_{3,n}} \bigg\{ b_n(i \land j) - |i-j|\bigg\}\I_{|i-j| \le b_n(i \land j)} X_i X_j,
	\end{align*}
	and similarly for $\bar{\sigma}^2_n(W_\PSR)$ with
	$\{ b_n(i \land j) - |i-j|\}\I_{|i-j| \le b_n(i \land j)}$ replaced by
	$( \ell_{i \lor j} - |i-j|)\I_{|i-j| \le \ell_{i \lor j}}$.
	By the Minkowski inequality,
	\begin{align*}
		\norm{\bar{\sigma}^2_n(W_\sub)-\bar{\sigma}^2_n(W_\PSR)}
		\le{}& \norm{\bar{\sigma}^2_{1,n}(W_\sub)-\bar{\sigma}^2_{1,n}(W_\PSR)}
		+\norm{\bar{\sigma}^2_{2,n}(W_\sub)-\bar{\sigma}^2_{2,n}(W_\PSR)} \\
		& +\norm{\bar{\sigma}^2_{3,n}(W_\sub)-\bar{\sigma}^2_{3,n}(W_\PSR)} \\
		={}& H_8 +H_9 +H_{10}.
	\end{align*}
	For $H_8$, since $\Xi_{1,n}$ is the set of indices in the beginning where $\ell_i$ is strictly increasing, $\max_{i \in \Xi_{1,n}} i \le C$ for some constant $C \in \R^+$ that may change from line to line in the remaining of this proof.
	Without loss of generality, consider $i \ge j$ as in \cref*{eq:genClassOn}.
	By the Minkowski inequality and the meaning of $\Xi_{1,n}$,
	\begin{align*}
		H_8
		\le{}& \frac{2}{n \bar{\ell}_n} \norm{\sum_{i=1}^{C+\ell_n} \sum_{j=1}^C \bigg\{ b_n(i \land j) -|i-j| \bigg\}\I_{|i-j| \le b_n(i \land j)} X_i X_j} \\
		& +\frac{2}{n \bar{\ell}_n} \norm{\sum_{i=1}^{C+\ell_n} \sum_{j=1}^C \bigg( \ell_{i \lor j} - |i-j|\bigg)\I_{|i-j| \le \ell_{i \lor j}} X_i X_j} \\
		={}& H_{11} +H_{12}.
	\end{align*}
	Consider $H_{11}$. The Minkowski inequality gives
	\begin{align}
		H_{11}
		\le{}& \frac{2}{n \bar{\ell}_n} \norm{\sum_{i=1}^{C+\ell_n} \sum_{j=1}^C \bigg\{ b_n(i \land j) -|i-j| \bigg\}\I_{|i-j| \le b_n(i \land j)} \left\{X_i X_j -\E(X_i X_j)\right\} } \nonumber \\
		& +\frac{2}{n \bar{\ell}_n} \left| \sum_{i=1}^{C+\ell_n} \sum_{j=1}^C \bigg\{ b_n(i \land j) -|i-j| \bigg\}\I_{|i-j| \le b_n(i \land j)} \E(X_i X_j) \right| \nonumber \\
		={}& H_{13} +H_{14}. \label{eq:interSub}
	\end{align}
	Since the stability of $\{X_i X_j\}$ for $j \in \{1,\ldots,i\}$ has been established in \cref{sec:proof-lastBlock}, it follows from Assumption \ref{asum:stability} and Lemma 1 of \citet{liu2010asymptotic} that
	\begin{equation} \label{eq:boundH13}
		H_{13}
		\le O(n^{-1-\lambda}) \sqrt{\sum_{i=1}^{C+\ell_n} \left| \sum_{j=1}^C C \right|^2}
		= O(n^{-1-\lambda}) O(n^{\frac{\lambda}{2}})
		= o(n^{\max(-\lambda, \frac{\lambda-1}{2})}).
	\end{equation}
	By the Minkowski inequality and $\sum_{k \in \Z} |\gamma_k| < \infty$ under Assumption \ref{asum:stability} \citep{wu2009recursive},
	\begin{equation} \label{eq:boundH14}
		H_{14} \le O(n^{-1-\lambda}) \sum_{i=1}^{C+\ell_n} \sum_{j=1}^C C |\gamma_{|i-j|}|
		= O(n^{-1-\lambda}) O(1)
		= o(n^{\max(-\lambda, \frac{\lambda-1}{2})}).
	\end{equation}
	Using the same procedure from \crefrange{eq:interSub}{eq:boundH14}, we have
	\begin{align*}
		H_{12}
		&\le O(n^{-1-\lambda}) \sqrt{\sum_{i=1}^{C+\ell_n} \left| \sum_{j=1}^C \ell_n \right|^2} +O(n^{-1-\lambda}) \sum_{i=1}^{C+\ell_n} \sum_{j=1}^C \ell_n |\gamma_{|i-j|}| \\
		&= O(n^{-1-\lambda}) O(n^{\frac{3\lambda}{2}}) +O(n^{-1-\lambda}) O(n^{\lambda})
		= o(n^{\max(-\lambda, \frac{\lambda-1}{2})}).
	\end{align*}
	For $H_9$, since the meaning of $\Xi_{2,n}$ is the set of indices that belongs to the last incomplete block where $F_n(i)$ is truncated by $n$,
	$\min_{i \in \Xi_{2,n}} i \ge (\ell_n/\Lambda)^{1/\lambda}-C = n-C$; see \cref{eq:blockStart}.
	Repeating the same arguments, we have
	$
	H_9 = o(n^{\max(-\lambda, \frac{\lambda-1}{2})}).
	$
	Finally, note that for $i,j \in \Xi_{3,n}$, we have
	$\{b_n(i \land j)-\ell_{i \lor j}\}\I_{|i-j| \le b_n(i \land j)}=1$ or
	$\{b_n(i \land j)-\ell_{i \lor j}\}\I_{|i-j| \le b_n(i \land j)}=0$.
	By the Minkowski inequality,
	\begin{align*}
		H_{10}
		={}& \frac{1}{n \bar{\ell}_n} \norm{\sum_{i \in \Xi_{3,n}} \sum_{j \in \Xi_{3,n}} \bigg\{ b_n(i \land j) -\ell_{i \lor j} \bigg\}\I_{|i-j| \le b_n(i \land j)} X_i X_j} \\
		\le{}& \frac{1}{n \bar{\ell}_n} \norm{\sum_{i \in \Xi_{3,n}} \sum_{j \in \Xi_{3,n}} \bigg\{ b_n(i \land j) -\ell_{i \lor j} \bigg\} \I_{|i-j| \le b_n(i \land j)} \left\{ X_i X_j -\E(X_i X_j) \right\} } \\
		& +\frac{1}{n \bar{\ell}_n} \left| \sum_{i \in \Xi_{3,n}} \sum_{j \in \Xi_{3,n}} \bigg\{ b_n(i \land j) -\ell_{i \lor j} \bigg\} \I_{|i-j| \le b_n(i \land j)} \E(X_i X_j) \right| \\
		={}& H_{15} +H_{16}.
	\end{align*}
	Without loss of generality, consider $i \ge j$ again.
	By similar reasons as in \cref{eq:boundH13,eq:boundH14},
	\begin{align*}
		H_{15} &\le O(n^{-1-\lambda}) \sqrt{\sum_{i \in \Xi_{3,n}} \left| \sum_{j=\min_{k \in \Xi_{3,n}} k}^i \bigg\{ b_n(i \land j) -\ell_{i \lor j} \bigg\} \I_{|i-j| \le b_n(i \land j)} \right|^2} \\
		&= O(n^{-\frac{1}{2}}) = o(n^{\max(-\lambda, \frac{\lambda-1}{2})}), \\
		H_{16} &\le O(n^{-1-\lambda}) \sum_{i \in \Xi_{3,n}} \sum_{j \in \Xi_{3,n}} \I_{|i-j| \le \ell_n} |\gamma_{|i-j|}|
		= O(n^{-1-\lambda}) O(n^{\lambda})
		= o(n^{\max(-\lambda, \frac{\lambda-1}{2})}).
	\end{align*}
	Combining the results completes the proof.
\end{proof}

\subsection{\texorpdfstring{Exact Convergence Rate of $\hat{\sigma}^2_n(W_{\Lest})$}{
		Exact Convergence Rate under the Simple Modification}} \label{sec:omit-winL}

For the exact convergence rate of $\hat{\sigma}^2_n(W_{\Lest})$, Theorem \ref{thm:mse} does not apply directly.
Nevertheless, we can use the same procedure in \cref{sec:proof-mse} to prove a similar theorem particularly for this simple modification.
After going through the steps, we have
\[
n^{1-\lambda} \Var(\hat{\sigma}_{n,\Lest}^2)
\to \frac{4\Lambda}{3(\lambda+1)} \quad \text{and} \quad
\Bias(\hat{\sigma}_{n,\Lest}^2)
\sim - \frac{\Lambda^{-1}}{(1-\lambda)} n^{-\lambda} v_1.
\]
Clearly, $\lambda_\star = 1/3$ optimize the order of $\MSE(\hat{\sigma}_{n,\Lest}^2)$.
We can also find that $\Lambda_\star = (9/2)^{1/3} \kappa_1^{2/3}$ by differentiating $\MSE(\hat{\sigma}_{n,\Lest}^2)$.
As a result, the  optimal mean squared error improves from $1.12B_n$ ($\PSR$) to $1.08B_n$.

\subsection{\texorpdfstring{Exact Convergence Rate of $\hat{v}_{q,n,\LASER(p,\phi)}$}{
		Exact Convergence Rate of the Nuisance Parameter Estimator}} \label{sec:omit-vq}

We can also use the same procedure in \cref{sec:proof-mse} to prove the exact convergence rate of $\hat{v}_{q,n,\LASER(p,\phi)}$.
Consider the case of interest $\psi = \theta$.
If $u_{p+q} < \infty$, we have
\[
\Bias(\hat{v}_{q,n,\LASER(p,\phi)}) \sim -\frac{1}{\Theta^p n^{p\theta}} v_{p+q}.
\]
Similarly,
\begin{align*}
	& 4^{-1} \sigma^{-4} n^{1-\psi} \Var(\hat{v}_{q,n,\LASER(p,\phi)}) \\
	\sim{}& \left\{
	\begin{array}{ll}
		\int_0^\Psi \int_0^{\nu_w} \int_0^{(\phi-1)w} \int_0^{w+y} (z^q -z^{p+q} \Theta^{-p})^2 \, \di z\ \di y\ \di x\ \di w, & \phi > 1; \\
		\int_0^\Psi \int_{(w/\Psi)^{1/\psi}}^1 (w^q -w^{p+q} \Theta^{-p})^2 \,\di x\ \di w, & \phi = 1; \\
	\end{array}
	\right. \\
	\sim{}& \left\{
	\begin{array}{ll}
		\begin{array}{l}
			\frac{2 \Psi^{2 q + 1} (\phi^{2 q + 2} - 1)}{(\phi - 1) (q + 1) (2 q + 1) (2 \psi q + \psi + 1)}
			-\frac{8 \Psi^{p + 2 q + 1} \Theta^{- p} (\phi^{p + 2 q + 2}-1)}{(\phi - 1) (p + 2 q + 1) (p + 2 q + 2) (\psi p + 2 \psi q + \psi + 1)} \\
			+ \frac{2 \Psi^{2 p + 2 q + 1} \Theta^{- 2 p} (\phi^{2 p + 2 q + 2} - 1)}{(\phi - 1) (p + q + 1) (2 p + 2 q + 1) (2 \psi p + 2 \psi q + \psi + 1)}
		\end{array}
		, & \phi > 1; \\
		\frac{4 \Psi^{2 q + 1}}{(2 q + 1) (2 \psi q + \psi + 1)}
		- \frac{8 \Psi^{p + 2 q + 1} \Theta^{- p}}{(p + 2 q + 1) (\psi p + 2 \psi q + \psi + 1)}
		+ \frac{4 \Psi^{2 p + 2 q + 1} \Theta^{- 2 p}}{(2 p + 2 q + 1) (2 \psi p + 2 \psi q + \psi + 1)}, & \phi = 1. \\
	\end{array}
	\right.
\end{align*}
The standardized mean squared error when $\psi = \theta = 1/(1+2p)$ is
\begin{align*}
	\left\{
	\begin{array}{ll}
		\begin{array}{l}
			\frac{\kappa_{p+q}^2}{\Theta^{2p}}
			+\frac{\Psi^{2 q + 1} (\phi^{2 q + 2} - 1) (2p+1)}{(\phi - 1) (q + 1) (2 q + 1) (p+q+1)}
			-\frac{8 \Psi^{p + 2 q + 1} \Theta^{- p} (\phi^{p + 2 q + 2}-1) (2p+1)}{(\phi - 1) (p + 2 q + 1) (p + 2 q + 2) (3p+2q+2)} \\
			+ \frac{\Psi^{2 p + 2 q + 1} \Theta^{- 2 p} (\phi^{2 p + 2 q + 2} - 1) (2p+1)}{(\phi - 1) (p + q + 1) (2 p + 2 q + 1) (2p+q+1)}
		\end{array}
		, & \phi > 1; \\
		\frac{\kappa_{p+q}^2}{\Theta^{2p}}		
		+\frac{2 \Psi^{2 q + 1} (2p+1)}{(2 q + 1) (p+q+1)}
		- \frac{8 \Psi^{p + 2 q + 1} \Theta^{- p} (2p+1)}{(p + 2 q + 1) (3p+2q+2)}
		+ \frac{2 \Psi^{2 p + 2 q + 1} \Theta^{- 2 p} (2p+1)}{(2 p + 2 q + 1) (2p+q+1)}, & \phi = 1. \\
	\end{array}
	\right.
\end{align*}
When $\Psi = \Theta$, we have
\[
\Psi_\star = \left\{
\begin{array}{ll}
	\begin{array}{l}
		\Big\{ \frac{(\phi^{2 q + 2} - 1) (2p+1)}{2 (\phi - 1) p (q + 1) (p+q+1)}
		- \frac{4 (\phi^{p + 2 q + 2}-1) (2p+1) (2q+1)}{(\phi - 1) p (p + 2 q + 1) (p + 2 q + 2) (3p+2q+2)} \\
		+ \frac{(\phi^{2 p + 2 q + 2} - 1) (2p+1) (2q+1)}{2 (\phi - 1) p (p + q + 1) (2 p + 2 q + 1) (2p+q+1)} \Big\}^{-\frac{1}{1+2p+2q}} \kappa_{p+q}^{\frac{2}{1+2p+2q}}
	\end{array}
	, & \phi > 1; \\
	\Big\{ \frac{(2p+1)}{p (p+q+1)}
	- \frac{4 (2p+1) (2q+1)}{p (p + 2 q + 1) (3p+2q+2)}
	+ \frac{(2p+1) (2q+1)}{p (2 p + 2 q + 1) (2p+q+1)} \Big\}^{-\frac{1}{1+2p+2q}} \kappa_{p+q}^{\frac{2}{1+2p+2q}}, & \phi = 1. \\
\end{array}
\right.
\]
If we allow $\Theta = \rho \Psi_\star$, then the standardized mean squared error becomes
\begin{align*}
	\left\{
	\begin{array}{ll}
		\begin{array}{l}
			\frac{\kappa_{p+q}^2}{\rho^{2p} \Psi_\star^{2p}}
			+\frac{\Psi_\star^{2 q + 1} (\phi^{2 q + 2} - 1) (2p+1)}{(\phi - 1) (q + 1) (2 q + 1) (p+q+1)}
			-\frac{8 \rho^{-p} \Psi_\star^{2 q + 1} (\phi^{p + 2 q + 2}-1) (2p+1)}{(\phi - 1) (p + 2 q + 1) (p + 2 q + 2) (3p+2q+2)} \\
			+ \frac{\rho^{-2p} \Psi_\star^{2 q + 1} (\phi^{2 p + 2 q + 2} - 1) (2p+1)}{(\phi - 1) (p + q + 1) (2 p + 2 q + 1) (2p+q+1)}
		\end{array}
		, & \phi > 1; \\
		\frac{\kappa_{p+q}^2}{\rho^{2p} \Psi_\star^{2p}}		
		+\frac{2 \Psi_\star^{2 q + 1} (2p+1)}{(2 q + 1) (p+q+1)}
		- \frac{8 \rho^{-p} \Psi_\star^{2 q + 1} (2p+1)}{(p + 2 q + 1) (3p+2q+2)}
		+ \frac{2 \rho^{-2p} \Psi_\star^{2 q + 1} (2p+1)}{(2 p + 2 q + 1) (2p+q+1)}, & \phi = 1, \\
	\end{array}
	\right.
\end{align*}
and we have
\[
\rho_\star = \left\{
\begin{array}{ll}
	\begin{array}{l}
		\Big\{ \frac{(\phi^{2 p + 2 q + 2} - 1) (p + 2 q + 1) (p + 2 q + 2) (3p+2q+2)}{4 (\phi^{p + 2 q + 2}-1) (p + q + 1) (2 p + 2 q + 1) (2p+q+1)} \\
		+ \frac{\Psi_\star^{-2p-2q-1} \kappa_{p+q}^2 (\phi - 1) (p + 2 q + 1) (p + 2 q + 2) (3p+2q+2)}{4 (\phi^{p + 2 q + 2}-1) (2p+1)} \Big\}^{\frac{1}{p}}
	\end{array}
	, & \phi > 1; \\
	\left\{ \frac{(p + 2 q + 1) (3p+2q+2)}{2 (2 p + 2 q + 1) (2p+q+1)}
	+ \frac{\Psi_\star^{-2p-2q-1} \kappa_{p+q}^2 (p + 2 q + 1) (3p+2q+2)}{4 (2p+1)}	 \right\}^{\frac{1}{p}}, & \phi = 1. \\
\end{array}
\right.
\]

\subsection{Positive Definiteness Adjustment} \label{sec:omit-pd}

A practical concern in long-run covariance matrix estimation is that the finite-sample estimates may not be positive definite.
Indeed, classical estimators such as those utilizing the overlapping batch means and Bartlett kernel are only guaranteed to be positive semi-definite.
To address this issue, we follow \citet{jentsch2015pd} to provide an adjustment that retains asymptotic properties of long-run covariance matrix estimators:

\begin{enumerate}[label=(\arabic*)]
	\item Obtain the diagonal matrix $\hat{V} = \mathrm{diag}(\hat{\Sigma}_n)$.
	\item Compute the correlation matrix $\hat{R}_n = \hat{V}^{-1/2} \hat{\Sigma}_n \hat{V}^{-1/2}$.
	\item Eigendecompose the correlation matrix into $\hat{R}_n = \hat{Q} \hat{\Lambda}_n \hat{Q}^\T$, where $\hat{Q}$ is an orthogonal matrix whose columns are the eigenvectors and $\hat{\Lambda}_n = \mathrm{diag}(\hat{\lambda}_1, \ldots, \hat{\lambda}_d)$ is the diagonal matrix whose diagonal elements are the eigenvalues.
	\item Set $\hat{\Lambda}_n^+ = \mathrm{diag}(\hat{\lambda}_1^+, \ldots, \hat{\lambda}_d^+)$, where $\hat{\lambda}_i^+ = \max(\hat{\lambda}_i, a n^{-b})$ for some user-specified $a>0$ and $b>1/2$.
	We use $a=\{\ln(n)/d\}^{1/2}$ and $b=9/10$ as suggested in \citet{vats2021kernel}.
	\item Return the positive definite estimate $\hat{\Sigma}_n^+ = \hat{V}^{1/2} \hat{Q}^\T \hat{\Lambda}_n^+ \hat{Q} \hat{V}^{1/2}$.
\end{enumerate}

\subsection{Online Quantile Regression: Independent Case} \label{sec:omit-quant}

Figure \ref{fig:quant-supp} reports the results in the independent case, which are similar to the results in the dependent case in \cref*{sec:online-quantile-regression}.

\begin{figure}[!t]
	\centering
	\includegraphics[width=\linewidth]{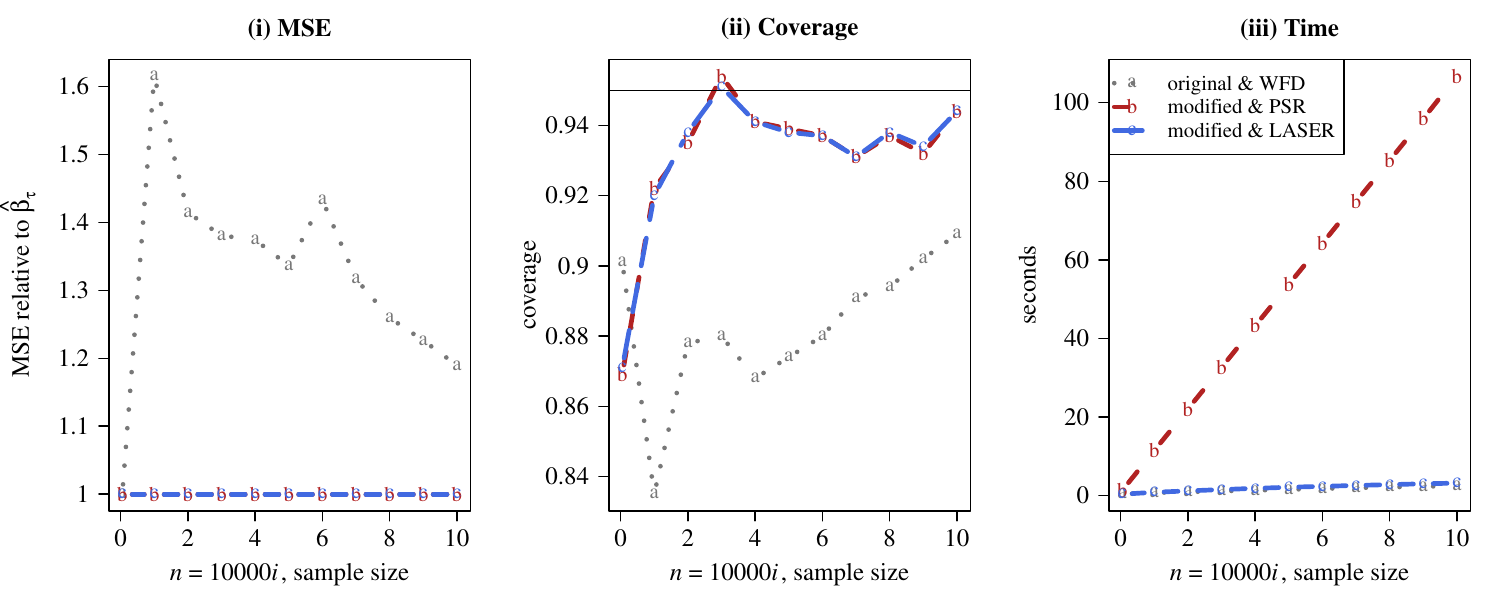}
	\caption{Online quantile regression under dependence using original and modified inference procedures at $5\%$ nominal size: 
		(a) original and $\WFD$ (dotted gray); 
		(b) modified and $\PSR$ (dashed red); 
		(c) modified and $\LASER(1,1)$ (longdash blue).
		The plot in (i) shows the value of $\MSE(\cdot)/\MSE(\hat\beta_{\tau,n})$.}
	\label{fig:quant-supp}
\end{figure}

\subsection{Online Change Point Detection: Other Cases} \label{sec:omit-cp}

Figure \ref{fig:cp-supp} reports the results for $\nu+k^* = 601$ (top panel) and $\nu+k^* = 1001$ (bottom panel), which are similar to the results for $\nu+k^* = 1401$ in \cref*{sec:online-change-point-detection}.

\begin{figure}[!t]
	\centering
	\includegraphics[width=\linewidth]{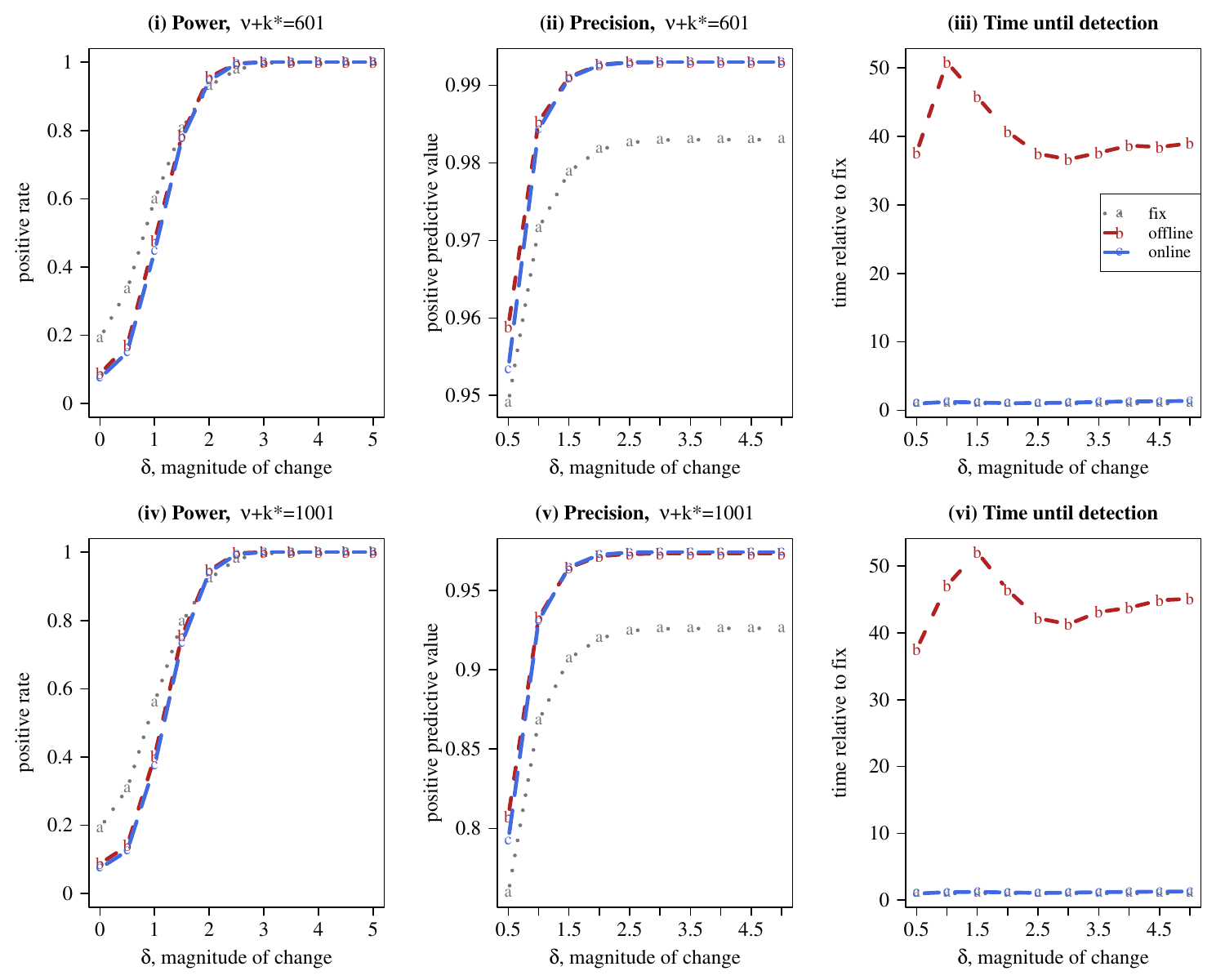}
	\caption{Online change point detection at $5\%$ nominal size using different long-run variance estimation methods: (a) fix (dotted gray); (b) offline (dashed red); (c) online (longdash blue).
		The locations of the change point are $\nu+k^* = 601$ (top panel) and $\nu+k^* = 1001$ (bottom panel).}
	\label{fig:cp-supp}
\end{figure}

\end{document}